%% file: ConvFieldsTheory_arxiv.tex
\documentclass[11pt]{article}
\usepackage{amsmath,amssymb,amsthm,bbm}
\usepackage{mathtools}
\usepackage{dsfont}
\usepackage{verbatim}
\usepackage{graphicx,rotating}
\usepackage[dvipsnames]{xcolor}
\usepackage{natbib}
\usepackage{placeins}
\usepackage{enumitem} 
\usepackage{xr}
\usepackage[colorlinks,citecolor=blue,urlcolor=blue]{hyperref}
\usepackage{multirow}
\usepackage{makecell}

\input{self_shortcuts.tex}

\input{self_environments.tex}

\newcommand{\supp}{ {\rm supp} }
\newcommand{\ska}[2]{ \langle{#1},\, {#2}\rangle }
\newcommand{\Cov}[2]{ {\rm Cov}\left[ {#1},\, {#2}\right] }

\newcommand{\figurepath}{}
\newcommand{\figurepathh}{}
\newcommand{\figurepathfwer}{}

\begin{document}

\title{ Precise FWER Control for Gaussian Related Fields: Riding the SuRF to continuous land - Part 1 }
\author{  Fabian J.E. Telschow$^1$, Samuel Davenport$^2$ \\[5mm]
	$^1$Department of Mathematics, Humboldt Universit\"at zu Berlin \\
	$^2$Division of Biostatistics, University of California, San Diego \\
}

\date{\today}
\maketitle

\begin{abstract} 
	\input{abstract.tex}
\end{abstract}

\section{Introduction}
\input{introduction.tex}

\input{main_body.tex}

\section*{Acknowledgments}
\input{acknowledgements.tex}

\bibliographystyle{unsrtnat}
\bibliography{literature}
\newpage

\section*{Appendix}
\input{appendix.tex}

\end{document}

%% file: self_shortcuts.tex
\newcommand{\mR}{\mathbb{R}}
\newcommand{\mS}{\mathbb{S}}

\newcommand{\mN}{\mathbb{N}}
\newcommand{\mZ}{\mathbb{Z}}

\newcommand{\1}{\mathbbm{1}}

\newcommand{\cA}{\mathcal{A}}
\newcommand{\cB}{\mathcal{B}}
\newcommand{\cC}{\mathcal{C}}

\newcommand{\cF}{\mathcal{F}}

\newcommand{\cL}{\mathcal{L}}
\newcommand{\cM}{\mathcal{M}}

\newcommand{\cT}{\mathcal{T}}
\newcommand{\cU}{\mathcal{U}}
\newcommand{\cV}{\mathcal{V}}
\newcommand{\cW}{\mathcal{W}}

\newcommand{\fc}{\mathfrak{c}}

\newcommand{\tr}{{\rm tr}}

\newcommand{\Var}{{\rm Var}}

\renewcommand{\P}{{\rm P}}

\newcommand{\vol}{{\rm vol}}

%% file: self_environments.tex
\newtheorem{theorem}{Theorem}
\newtheorem{lemma}{Lemma}
\newtheorem{proposition}{Proposition}

\theoremstyle{definition}
\newtheorem{definition}{Definition}

\theoremstyle{plain}

\theoremstyle{remark}

\newtheorem{remark}{Remark}

%% file: abstract.tex
The Gaussian Kinematic Formula (GKF) is a powerful and computationally efficient tool to perform statistical inference on random fields and became a well-established tool in the analysis of neuroimaging data.
Using realistic error models, \cite{Eklund2016} recently showed that GKF based methods for \emph{voxelwise inference} lead to conservative control of the familywise error rate (FWER) and to inflated false positive rates for cluster-size inference. In this series of articles we identify and resolve the main causes of these shortcomings in the traditional usage of the GKF for voxelwise inference. In particular applications of Random Field Theory have depended on the \textit{good lattice assumption}, which is only reasonable when the data is sufficiently smooth. We address this here, allowing for valid inference under arbitrary applied smoothness and non-stationarity, for Gaussian and Gaussian related fields. We address the assumption of Gaussianity in Part 2, where we also demonstrate that
our GKF based methodology is non-conservative under realistic error models.

%% file: introduction.tex
In  experiments in neuroscience using functional Magnetic Resonance Imaging (fMRI) of the brain,
data consists of 3D-images representing the time dynamics of the blood-oxygen-level dependence (BOLD).  After an extensive pre-processing pipeline, including corrections for head motion and respiratory effects
the time-series for each subject is combined into subject level maps (3D-images, typically consisting of coefficients of a linear model) via a first-level analysis step, see e.g., \cite{Poldrack2011} and the references therein.
These high resolution images contain hundreds of thousands of voxels and allow for the detection of the locations of differences in \%BOLD in the brain across different tasks and subjects. 
A major practical challenge, however, is the low signal-to-noise ratio which is often
increased by convolving each subject level map with a
smoothing kernel, typically an isotropic 3D-Gaussian kernel with Full Width at Half Maximum (FWHM)
$\approx 4$-$8\,$mm, i.e., $ \approx 2$-$4$ voxels, \cite[Supplementary material]{Eklund2016}.
A standard statistical analysis to identify areas of activation in the brain
is the mass-univariate approach which combines hypotheses tests at each voxel of a
test statistic image $\tilde T$ obtained from
the smoothed subject level maps with a multiple testing procedure.
A practical solution to the multiple testing of the thousands of voxels in the brain
in terms of controlling the \textit{family-wise error rate} (FWER) was
pioneered in \cite{Worsley1992, Friston1994, Worsley1996} and
is known as \textit{Random Field Theory} (RFT) in the neuroimaging community. 
The key innovation was to use the \textit{Gaussian Kinematic Formula} (GKF)
(e.g.,\cite{Taylor2006}) to approximate the probability that
the maximum of a Gaussian random field with  $C^3$ sample paths over a compact Whitney stratified (WS) manifold
$\mathcal{M} \subset \mathbb{R}^D$, $D\in\mathbb{N}$, exceeds a given threshold. 
The accuracy of this approximation has been theoretically studied in \citep{Taylor2005}.
We refer to the setting from these papers, where $\cM$ can be either the brain or its surface, as
\textit{traditional RFT}.
Originally assuming stationarity of the data, traditional RFT became a standard tool in the
analysis of fMRI data, \cite{Worsley1996},
because of its low computational costs. In the late 1990s it was
the core inference methodology in standard software such as
\textit{Statistical Parametric Mapping}, \citep{SPM} and it
is still the backbone of cluster-size and peak inference \cite{Friston1994,Chumbley2009,Schwartzman:2019Peak}.
However, in the early 2000s it was shown in Monte-Carlo-simulations that RFT
for voxel-wise inference is very conservative \citep{Nichols2003}, especially at the levels of applied
smoothing used in applications.
Therefore it has been superseded by time costly, yet accurate
permutation tests, among others, \citep{Nichols2001,Winkler2016,Winkler2016a}.
Using resting state data combined with fake task designs the seminal article \cite{Eklund2016} showed
that voxelwise RFT inference is conservative (see their Figure 1) which confirms the
findings of \citep{Nichols2003} on a realistic data model. They also
showed that cluster-size inference based on RFT can have inflated false positive rates.
They argued that -- in addition to the lack of smoothness-- a lack of stationarity in real data is another
major cause of their findings.
Although these factors play a role in the conservativeness of voxel-wise inference using RFT
we identify in this article and in Part 2 -  \cite{Davenport2023}  - the true causes
and present an elegant solution.
Here we assume that the data is Gaussian, while Part 2 generalizes our approach
to non-Gaussian data and validates its performance using
a large resting-state dataset from the UK Biobank \citep{Alfaro2018}.

Traditional RFT assumes that an underlying continuous random field is
approximated well by the random field observed on the voxel lattice, compare \cite[p.430]{Nichols2003}.
More precisely, it assumes
$
	\max_{v\in\cV} \tilde T(v) \approx \max_{s\in \cM} \tilde T(s)
$.
We call this the \textit{good lattice assumption} which was expressed, albeit less
clearly, earlier in \citep{Worsley1996}:
\textit{
	``[...] the search region
was regarded as a region with a smooth boundary
defined at every point in 3-D. In practice only voxel
data are available, and this will be regarded as a
continuous image sampled on a lattice of equally
spaced points. Thus a voxel is treated as a point in 3-D
with zero volume, although it is often displayed on
computer screens and in publications as a volume
centered at that point."}
As we shall argue, this view does not carefully distinguish between the \textit{data}
and the \textit{atoms of a probabilistic model}.

(Raw) data is the information about nature which is extracted from measurement
devices and is without exception currently discrete. In contrast,
the atoms (or data objects\footnote{We favor the term atoms over data objects
	because this nomenclature better separates the two concepts.}) of
probabilistic models used to describe data and perform statistical
inference are functions, curves, images, shapes, trees or other
mathematical objects \citep{Wang:2007, Marron:2014}.
In our neuroimaging example, the data consists of the subject level maps obtained
from the first level model observed on a discrete set $\cV$ of voxels which ``belong'' to
the brain.\footnote{This is not ``raw"  data as it is the result of a pre-processing pipeline
and some statistical modeling. It is another topic that complexity of today's data
often makes statistical modeling starting with raw data infeasible. But consistency along the data analysis  pipeline
should be guaranteed as much as feasible.} In traditional RFT the test statistic $\tilde T$ used for inference is
a random field over $\cV$ obtained from the smoothed data, while the atoms of RFT
are random fields with twice differentiable sample paths over the brain, i.e.,
a Whitney stratified (WS) manifold $\cM$.
Historically, the good lattice assumption hides this mismatch between the data and the atoms of
RFT.
To be independent of the good lattice assumption, \cite{Worsley2005} and \cite{Taylor2007b}
propose voxel based methods under the (unrealistic) model that the data consists
of a signal plus stationary Gaussian noise which has a separable covariance structure.

In this article we resolve the inconsistency between the data and the atoms of traditional RFT differently.
Our key observation is that the smoothing operation which is used to increase
the signal-to-noise ratio in the preprocessing naturally transforms the subject
level maps (the data), into functions over a WS manifold $\cM \supset \cV$ (atoms of RFT).
Thus, the test statistic $\tilde T$ can be interpreted as a function over $\cM$, too.
This allows us to solve the two main problems of the good
lattice assumption, i.e.,
\begin{itemize}
	\item[(i)]
	$\max_{v \in \cV} \tilde T(v) \leq \max_{s \in \cM} \tilde T(s)$ for all $\cV\subset\cM$
	\item[(ii)] continuous quantities, e.g., derivatives, are replaced by voxel-based,
				discrete counterparts.
\end{itemize}
In particular, $(i)$ is at the core of the conservativeness in traditional RFT as it implies
	\begin{equation}\label{eq:conserve}
	\P\Big( \max_{v \in \mathcal{V}} \tilde T(v) > u_\alpha \Big)
	\leq
	\P\Big( \max_{x \in \cM} \tilde T(x) > u_\alpha \Big) \approx \alpha
	\end{equation}
	for $u_\alpha \in \mathbb{R}$ being the $ \alpha $ level threshold obtained
	from the GKF.

In order to resolve $(i)$ and $(ii)$ we interpret the smoothing step from pre-processing
as transforming the data into functions over a WS manifold $\cM$, i.e., atoms of RFT.
More precisely, given an observed Gaussian random field $X$ (``a subject level map" aka data)
over a finite set $ \cV \subset \mathbb{R}^D $, $D\in\mathbb{N}$, a compact,
$D$-dimensional WS manifold $\cM$  and a smoothing kernel function $K: \mathcal{V} \times \mathcal{M} \rightarrow \mathbb{R}$, we define  the random field $\tilde X$ (``a smoothed subject level map" aka atom of RFT) given by
$
\tilde X(s) = \sum_{v\in\mathcal{V}} X(v) K(v,s)
$, $s\in\cM$.
We call $\tilde X$ a \emph{SUper Resolution Field} (SuRF) since -- given data $X$ -- it can be evaluated at any $ s \in \mathcal{M}$.
As many path properties of SuRFs are inherited from $K$, we shall show in
Theorem \ref{thm:GKF} that the GKF holds for $\tilde T$ (``a statistic depending on smoothed subject level maps"
aka built from atoms of RFT) under mild conditions on $K$ and $X$.

In order to perform inference using the GKF we must specify the domain $\cM$. The latter has not been properly
defined in the neuroimaging literature so far because of the reliance
on the good lattice assumption. This leads to a further mismatch between theory and practice in traditional RFT as the GKF is not valid for random fields over discrete sets.
We choose $\cM$ to be the \textit{voxel manifold} $\cM_\cV$ which we define to be a disjoint
union of $D$-dimensional hyperrectangles aligned with the coordinate axes and centered at the
elements of $\mathcal{V}$. This is a natural way to define a $D$-dimensional WS manifold from $ \mathcal{V} $,
see Figure \ref{fig:voxmf} in the supplementary material.
The main benefit of using $\cM_\cV$ is that it simplifies the estimation of the Lipschitz Killing Curvatures (LKCs)
-- the unknown parameters in the GKF. Moreover, the statistic $\max_{s\in \cM_\cV} \tilde T(s)$ can be well approximated using numerical optimizers, which allows us to resolve problem (i), the dominant factor in the observed conservativeness.

SuRFs also allow us to address problem $(ii)$.
In current software packages the LKCs are estimated from the random
fields evaluated on $\cV$ assuming the good lattice
assumption and stationarity (\cite{Forman1995}, \citep{Kiebel1999}, \cite{Worsley1996}).
These approaches use discrete approximations of the derivatives of $\tilde X$
instead of the available true derivatives. Less widely used estimators
drop the restrictive assumption of stationarity but still rely on discrete
approximation of derivatives, these are the warping estimator \citep{Taylor2007},
and the Hermite projection estimator (HPE) \citep{Telschow2020}.\footnote{In theory the HPE does not depend on discrete derivatives. However the HPE does requires the critical values of residual fields
and the current implementation of it uses discrete derivatives to do so.}
In Section \ref{scn:SuRFLKCestim} we propose a consistent LKC estimator based on the explicitly
available derivatives of a SuRF, which is computational fast as it requires only
a minimal resolution increase to obtain accurate estimates of the LKCs in practice.

Inherent in any method that applies smoothing is that the resulting inference will only provide
strong control of the FWER with respect to the smoothed signal. Since smoothing is regularly
used in neuroimaging this affects a wide range of methods\footnote{Compare for example \cite[e.g., pp.48-49]{Lazar:2008fMRI} for detailed explanations why data analysis in neuroimage
requires smoothing.}.  Widely cited articles such as \cite{Nichols2003} talk about FWER control in the strong sense but this is implicitly with respect to the smooth signal. In order to clarify this issue,
we show in Section \ref{scn:FWER} that kernel smoothing results in a FWER control which lies in between weak and strong control. In particular localization of a significant signal on the level of data is limited by the support of the kernel.

The article is structured as follows. In Section \ref{scn:notations} we define the notation.
Section \ref{scn:theory} studies theoretical properties of SuRFs.
In particular, Section \ref{scn:FWER} explains in detail how voxelwise inference based on the GKF using SuRFs is carried out and provides an explanation of the type of the resulting FWER control.
In Section \ref{Scn:Sims} we verify our findings using a simulation study. In particular,
Section \ref{scn:simLKC} reports the results of the SuRF LKC estimates under different models
and compares them to other published LKC estimators.
Simulations which demonstrate that the SuRF framework avoids the conservativeness of traditional voxelwise RFT inference can be found in Section \ref{scn:simFWER}. In Section \ref{scn:Discussion} we discuss our findings and
review other potential applications of SuRFs.
A Matlab implementation of the SuRF methodology is available in the RFTtoolbox \citep{RFTtoolbox}. Scripts reproducing our simulation results are available at \url{https://github.com/ftelschow/ConvolutionFieldsTheory}.

%% file: main_body.tex
\section{Notation and Definitions }\label{scn:notations}
In this section, we establish notation used in the article. We assume that $v \in \mR^D$ is a column vector, i.e., we identify $\mR^D$ with $\mR^{D\times 1} $. 
With $\cM$ we denote a $ D $-dimensional, compact $ \cC^2 $-\textit{Whitney-stratified} (WS) manifold
isometrically embedded into a $D$-dimensional manifold $\overline{\cM}$ without boundary.
Recall that a WS manifold of dimension $D$ is a space $\cM = \bigcup_{d=1}^D \partial_d\cM$ decomposed into strata where the stratum $ \partial_d\cM \subset
\overline{\cM}$ is a manifold of dimension $d$ and all the strata are disjoint, i.e.,
$\partial_d\cM\cap \partial_{d'}\cM = \emptyset$ for all $d=d'$,
compare \citep[Chapter 8]{Adler2007} for more details.
A $\cC^2$-chart (which gives local $\cC^2$-coordinates) around
$ s \in \overline{\cM} $ is given
by a tuple $ \big(\, \overline{U}, \overline{\phi} \,\big) $ where
$\overline{U}\subset \overline{\cM}$ open,
$ s \in \overline{U} $ and $ \overline{\phi} \in \mathcal{C}^2(\,\overline{U}, \overline{V}\,) $
is a diffeomorphism onto an open set $ \overline{V} \subset \mathbb{R}^D $. By the compactness of $\cM$ there exists a set of finitely many charts
$(\overline{U}_\alpha,\overline{\phi}_\alpha)_{\alpha \in \{1,\ldots ,P\}}$ of
$ \mathcal{M} $, $ P \in \mathbb{N} $, such that $\cM\subset\bigcup_{\alpha=1}^P\overline{U}_\alpha$. This union is the only relevant part of $\overline{\cM}$ since we are
only interested in properties of $\cM$. The surrounding manifold $\overline{\cM}$ is only introduced for
technical requirements in the formulation of the GKF.

We denote with $f_\alpha = f\circ\phi_\alpha^{-1}$ the coordinate representation
of $f$ in the chart $U_\alpha$ and with
$
\nabla f_\alpha$ the gradient of $ f_\alpha $, i.e.,
$\nabla f_\alpha(x)
	= \Big( \frac{\partial f_\alpha}{\partial x_1}(x),\ldots,
	  \frac{\partial f_\alpha}{\partial x_D}(x) \Big) \in \mathbb{R}^{1\times D}$
for $x \in V_\alpha$,
and with $ \nabla^2{f_\alpha} $ the Hessian of $ f_\alpha $, i.e.,
$\nabla^2{f_\alpha}(x) \in \mathbb{R}^{D \times D}$ is the matrix with $d$-$d'$th entry
$\frac{\partial^2 f_\alpha}{\partial x_d\partial x_{d'}}(x)$ for $x \in V_\alpha$.
If the gradient $\nabla$ or the Hessian $\nabla^2$ is applied to a function with
two arguments, then it is always assumed to be with respect to the first argument. 
 We will use $s, s'$ for
points in $\cM$ or $\overline{\cM}$ and $x,y$ for points in local coordinates. For simplicity in notation, given $h\in C^1(\mathbb{R}^D)$ and a multi-index
$\beta\in \mathbb{N}^{d}$, $d\leq D$, we write
\begin{equation}
	\partial_{\beta}h(x) = \frac{\partial^d h(x)}{\partial x_{\beta_1},\ldots, \partial x_{\beta_d}}(x)\,,
	~~~~x\in \mathbb{R}^D\,.
\end{equation}
If $h\in C^1(\mathbb{R}^D\times \mathbb{R}^D)$, we sometimes write $\big(\partial_{\beta}^x h\big)(x,y)$
for $(\partial_{\beta}h(\cdot, y))(x)$ evaluated at $(x,y)\in \mathbb{R}^D\times \mathbb{R}^D$.
Note that "$\partial\,$" also appears in our notation for the strata of a WS manifold.
For a symmetric matrix
$A \in \mathbb{R}^{D \times D}$, $D \in \mathbb{N}$, we define its
\emph{half-vectorization}  as
$
\mathbb{V}( A ) = (\, A_{11}, \ldots, A_{D1}, A_{22}, \ldots, A_{D2}, \ldots, A_{D-1D-1}, A_{DD-1}, A_{DD}\,)
$
and the set
$
a \cdot \mZ^D = \big\{ x\in \mR^D ~\vert~ x = (a_1z_1,\ldots, a_Dz_N),\,z\in\mZ^D \big\}
$
for $ D\in \mN$ and any vector $a\in\mR^D$ which has positive entries.
We also assume throughout the article that $ \mathcal{V} $ is a finite, discrete set.

\section{ Theory }\label{scn:theory}
In this section we define  Super-Resolution Fields (SuRFs),
introduce some of their basic properties and
show that they satisfy the Gaussian kinematic formula (GKF) under mild conditions.
Moreover, we derive computable estimators of the Lipschitz Killing curvatures (LKCs) of a SuRF defined
over a voxel manifolds even if the data is non-stationary.

\subsection{ Super-resolution Fields }\label{scn:SuRF}

\begin{definition}\label{def:kernel}
    We call a map $K: \overline{\cM} \times \cV \rightarrow \mR$ a \emph{kernel}. We say it is continuous/differentiable, if $s \mapsto K(s,v)$ is continuous/differentiable
    for all $v \in \cV$.
\end{definition}

\begin{remark}\label{rmk:conv_kernel}
This kernel definition is broader than the standard one in statistics, where it is typically a function $k:\overline{\cM} \rightarrow \mR$ with a normalization property. The function $ k $ can be viewed as a kernel from $ \overline{\cM} \times \overline{\cM} \rightarrow \mR$ by setting $K(s,s')=k(s-s')$.
\end{remark}

\begin{definition}\label{def:SuRF}
    Let $X$ be a $\mR$-valued random field
    on $\cV$ with covariance function $\fc(u,v) = \Cov{X(u)}{ X(v) }$, $u,v\in \cV$.
    For a kernel $K:\overline{\cM}\times\cV\rightarrow \mR$
    the random field over $\overline{\cM}$,
    \begin{equation}
        \tilde X(s) = \sum_{v\in\cV} K(s,v)X(v)\,,~ ~ ~s\in \overline{\cM}\,,
    \end{equation}
    is termed a \emph{Super-Resolution Field} (SuRF) linked to $K$ and $\cV$. For brevity, we will refer to it as $(\tilde X, X, K, \cV)$ or simply $\tilde X$ when the context is evident regarding $X$, $K$, and $\cV$.
    
    Given a SuRF $\big( \tilde X, X, K, \cV \big)$ we call $\big( \tilde X, X, \tilde K, \cV \big)$ a
    \emph{normalized SuRF} if
    $$
    	\tilde K:~\overline{\cM}\times \cV \rightarrow \mR\,,~~~
    	(s,v) \mapsto \frac{ K(s,v) }{ \sqrt{ \sum_{u\in\cV} \sum_{ v \in \cV } K( s, u )K( s, v ) \fc( u, v ) } }\,.
    $$
\end{definition}
\begin{remark}
    Let $\tilde X$ be a SuRF linked to $K$ and
    $\mathbb{E}[ X(v) ] < \infty$ for all $ v \in \cV $.
    Then by linearity
	\begin{equation}
        \mathbb{E}\big[ \tilde X( s ) \big] = 
            \sum_{ v \in \cV } K( s, v ) \mathbb{E}[ X( v ) ]\,,~~~s \in \overline{\cM}\,,
    \end{equation}
   	and if additionally $\mathbb{E}[ X( v )^2 ] < \infty$ for all
   	$v\in\cV$ then
    \begin{equation}\label{eq:CovConvFieldI}
        \Cov{ \tilde X(s)}{ \tilde X(s') } = 
    			  \sum_{u\in\cV} \sum_{ v \in \cV } K( s, u )K( s', v ) \fc( u, v ) < \infty\,,~~~s,s' \in \overline{\cM}\,.
    \end{equation}
    The last formula explains the definition of the normalized SuRF, since it shows
    that any normalized SuRF satisfies $\Var\big[ \tilde X(s) \big] = 1$ for all $ s \in \overline{\cM} $.
\end{remark}
\begin{remark}
	Equation \eqref{eq:CovConvFieldI} is similar to an inner product,
	where $(\fc(u,v))_{u,v\in\cV}$ is the representing ``matrix''.
	As such, we introduce, for all $s,s' \in \overline{\cM}$, the simplifying abbreviations
	\begin{equation}
 	\begin{split}\label{eq:InnerProdRepr}
     	\ska{ K_s}{K_{s'} } = \sum_{ u,v \in \cV } K(s,u) K(s',v) \fc(u,v)\,,\quad
    	\Vert K_s \Vert^2   = \sum_{ u,v \in \cV } K(s,u) K(s,v) \fc(u,v).
    \end{split}
\end{equation}

\end{remark}
\begin{remark}
	If $\cV \subseteq \overline{\cM} $ and $k:\overline{\cM} \rightarrow \mR$, then 
	we call the SuRF obtained from the kernel given in Remark \ref{rmk:conv_kernel} a \emph{convolution field}.
	Convolution fields appear
	naturally in many applications since convolving the observed data with a smoothing kernel is
	often a preprocessing step in signal processing or neuroimaging to improve the signal to noise ratio
	\citep{Turin1960matchedFilter, Worsley2002}.
\end{remark}

SuRFs are random fields with nice properties as most path properties are directly inherited
from the kernel $K$.
This is theoretically advantageous, as one can often establish the assumptions of results
like the GKF by imposing conditions like differentiability on the kernel $K$. The following proposition is self-evident but provided here for convenience.
\begin{proposition}\label{prop:KernelPropConvFields}
    Let $(\tilde X, X, K, \cV)$ be a SuRF and $k\geq 0$.
    If $K(\cdot, v) \in \mathcal{C}^k(\overline{\mathcal{M}})$ for all $ v \in \mathcal{V} $,
    then $\tilde X$ has sample paths of class $\mathcal{C}^k$.
\end{proposition}

\subsection{ Gaussian Kinematic Formula for SuRFs }\label{scn:GKF}
To demonstrate the benefit of thinking in terms of SuRFs we pose assumptions
on a kernel $K$ and the discrete field $X$ such that the corresponding
normalized SuRF $\tilde X$ satisfies the assumptions of the GKF \citep[Theorem 12.4.1, 12.4.2]{Adler2007}.
Thus, we first state the assumptions on a random field $f$ defined over $\overline{\cM}$
such that the GKF over $\cM$ holds. Recall that $f_\alpha$ is
the representation of $f$ in the chart $ \big( \overline{U}_\alpha, \overline{\phi}_\alpha\big)$
as introduced in Section \ref{scn:notations}. In this notation the GKF holds, if for all $\alpha \in \{1,\ldots, P\}$ and $(\overline{U}_\alpha,\overline{\phi}_\alpha)$ from the atlas of $\overline{\cM}$ we have that
\begin{enumerate}[leftmargin=1.4cm]
    \item[\textbf{(G1)}]  $f$ is a zero-mean, unit-variance and Gaussian on $\overline{\cM}$ with a.s. $\mathcal{C}^2$-sample paths.
    \item[\textbf{(G2)}] 
    $\big(\, {\nabla} f_\alpha(x),
		\mathbb{V}\big( {\nabla}^2{f_\alpha}(x) \big) \,\big)$
     is non-degenerate for all $x\in \overline{\phi}_\alpha\big( \overline{U}_\alpha ) \cap \cM$.
    \item[\textbf{(G3)}] There exist constants
    $\kappa,\gamma,\epsilon>0$ such that for each $d,d'\in \{1,\ldots,D\}$,
     $$\mathbb{E}\!\left[ \Big( \partial_{dd'}f_\alpha(x) - \partial_{dd'}f_\alpha(y) \Big)^2 \right] \leq \kappa \Big\vert \log \Vert x - y \Vert  \,\Big\vert^{-1-\gamma}, $$
     for all $x,y \in \overline{\phi}_\alpha\big( \overline{U}_\alpha \big) \cap \cM$ for which $\vert x - y \vert < \epsilon$.
\end{enumerate}
\begin{remark}
	Conditions \textbf{(G1)}-\textbf{(G3)} imply that the sample paths of
	$f$ are almost surely Morse functions, compare \cite[Corollary 11.3.2.]{Adler2007}.
	Moreover, by Lemma 1 from \cite{Davenport2022} these conditions do not depend
	on the particular choice of the $\mathcal{C}^3$ charts $(U_\alpha,\phi_\alpha)$,
	 $\alpha\in\{1,\ldots, P\}$, but rather hold for all $\mathcal{C}^3$ charts
	$(V,\varphi)$ of $\overline{\mathcal{M}}$.
\end{remark}
Applying Proposition \ref{prop:KernelPropConvFields}, \textbf{(G1)} holds for any normalized SuRF derived from a zero-mean Gaussian random field $ X $ on $ \mathcal{V} $ and a twice continuously differentiable kernel $ K $. The next two propositions establish that Condition \textbf{(G3)} is satisfied for a SuRF with
a $\mathcal{C}^3$-kernel.
\begin{proposition}\label{prop:HoelderCont}
	Let $ \gamma \in (0,1] $ and for all $ v \in \cV $ let $ K_\alpha(\cdot, v) $
	be $ \gamma $-H\"older continuous for all $\alpha\in\{1, \ldots, P\}$ 	with H\"older constants
	bounded above by $A > 0$ and $\mathbb{E}\left[ X(v)^p \right] < \infty$,
	$ p \in [1,\infty) $.
	Then $\tilde X$ has almost surely $\cL^p$-H\"older continuous paths, i.e.,
	\begin{equation}
    	\vert \tilde X_\alpha(x) - \tilde X_\alpha(y) \vert
    	\leq L \Vert x - y \Vert^\gamma
	\end{equation}
	for the charts $(\overline{U}_\alpha,\overline{\phi}_\alpha)$, $\alpha \in \{1,\ldots,P\}$,
	in the atlas of $\overline{\cM}$ covering $\cM$,
	all $x,y \in \overline{\phi}\big( \overline{U}_\alpha \big) \cap \cM$ and some
	random variable $L$ with finite $p$-th moment.
\end{proposition}
\begin{proposition}\label{prop:GaussG3}
	Let $ K(\cdot, v) \in \cC^1\big( \overline{\cM} \big) $ and
	$ \mathbb{E}\left[ X( v )^2 \right] < \infty $ for all
	$\nu\in\cV$. Then there exists a constant $ \kappa > 0 $ for
	$\alpha \in \{1,\ldots,P\}$ such that
	\begin{equation}
    	\mathbb{E}\left[ \left( \tilde X_\alpha(x) - \tilde X_\alpha(y) \right)^2 \right]
    						\leq \kappa \Big\vert \log \Vert x - y \Vert  \,\Big\vert^{-2}
	\end{equation}
	for all $x,y \in \overline{\phi}\big( \overline{U}_\alpha \big) \cap \cM$ such that
	$ 0 < \Vert x - y \Vert < 1$.
\end{proposition}
\begin{remark}
	Proposition \ref{prop:GaussG3} implies condition \textbf{(G3)}
	for a $\mathcal{C}^3$-kernel $K$ because each second order
	partial derivative of a normalized SuRF is a SuRF with continuous sample paths.
\end{remark}

\begin{definition}
	Given a set $ \cW \subset \mR^D $, we say that functions
	$ f_1, \dots, f_J: \cW \rightarrow \mR $ are $ \cW $-linearly
	independent if given constants $ a_1, \dots, a_{J} \in \mR $, the relation 
	\begin{equation*}
		\sum_{ j = 1 }^J a_j f_j(w) = 0
	\end{equation*}
	holding for all $ w \in \cW$ implies that $ a_j = 0 $ for all $ j \in \{ 1, \dots, J \}$.
\end{definition}
The following propositions link the linear independence of functions $K_\alpha(x,\cdot)$ for $x \in \overline{\phi}_\alpha(\overline{U}_\alpha\cap\cM)$ with Condition \textbf{(G2)} for SuRFs. We demonstrate that the needed linear independence condition for achieving \textbf{(G2)} is satisfied, for example, by Gaussian kernels.
\begin{proposition}\label{prop:nondegen}
	Let $ (\tilde X, X, K, \cV) $ be a SuRF with $K(\cdot,v) \in C^2\big(\overline{\cM}\big)$
	for all $v \in \mathcal{V}$  and $ \tilde Z$ be the corresponding normalized SuRF.
	For $\alpha \in \{ 1, \ldots, P \}$ and
	$x\in \overline{\phi}_\alpha\big( \overline{U}_\alpha \big) \cap \cM$ define
	$K_\alpha( x, \cdot ):\cV \rightarrow \mathbb{R}$
	by $v \mapsto K\big( \overline{\phi}_\alpha^{-1}(x), v \big)$ and
	set $ \cV_x = \left\lbrace v \in \cV: K_\alpha(x,v) \neq 0 \right\rbrace$.
	Assume that
	$ K_\alpha(x,\cdot)$, $\partial_d^x K_\alpha(x,\cdot)$
	and $\partial_{d'd''}^xK_\alpha( x,\cdot )$ for $d,d',d'' \in\{1, \ldots,D\} $ with
	$1  \leq d'\leq d''\leq D $ are $ \cV_x $-linearly
	independent and the random vector $ (X(v): v \in \cV_x) $ is non-degenerate. Then
	\begin{equation*}
	\Big( \tilde X_\alpha(x), \nabla \tilde X_\alpha(x), \mathbb{V}\big( \nabla^2 \tilde X_\alpha(x) \big) \Big)
	\text{ and }
	\Big(\tilde{Z}_\alpha(x), \nabla \tilde{Z}_\alpha(x), \mathbb{V}\big(\nabla^2 \tilde{Z}_\alpha(x)\big) \Big)
	\end{equation*}
	are non-degenerate Gaussian random vectors.
\end{proposition}

\begin{proposition}\label{prop:G2_2}
	Let $\tilde D = D+1+D(D+1)/2$ and $\alpha \in \{1\ldots, P\}$. Taking the
	gradient and Hessian w.r.t.
	$x\in \overline{\phi}_\alpha(\overline{U}_\alpha\cap\cM)$
	we define the vector valued functions
	\begin{equation}
		\mathbf{K}_{\alpha,x}(v) = \Big( K_\alpha(x,v), \nabla K_\alpha(x,v), \mathbb{V}\big( \nabla^2 K_\alpha(x,v) \big) \Big)
	\end{equation}
	indexed by $x \in \overline{\phi}_\alpha(\overline{U}_\alpha\cap\cM)$.
	If for each $x\in \overline{\phi}_\alpha(\overline{U}_\alpha\cap\cM)$
	there exist $v_1,\ldots, v_{\tilde D} \in \cV_x$
	such that $\mathbf{K}_{\alpha,x}(v_1),\ldots, \mathbf{K}_{\alpha,x}(v_{\tilde D})$
	are linearly independent,
	then $ K_\alpha(x,\cdot)$, $\partial_d^x K_\alpha(x,\cdot)$
	and $\partial_{d'd''}^xK_\alpha(x,\cdot)$ for $d,d',d'' \in\{1,\ldots, D\} $ with
	$1  \leq d'\leq d''\leq D $ are $\cV_x$ linearly independent.
\end{proposition}
\begin{remark}
	The results in Propositions \ref{prop:nondegen} and \ref{prop:G2_2}
	are stronger than \textbf{(G2)} as we want to emphasize that SuRFs often satisfy the assumptions
	of the expectation Metatheorem 11.2.1 from \cite{Adler2007}, see Corollary 11.2.2
	for the Gaussian version.
	However, our proofs show that \textbf{(G2)} follows already from  the assumption that
	$\partial_d^x K(x,\cdot)$ and $\partial_{d'd''}^xK(x,\cdot)$
	for $d\in\{1,\ldots, D\}$ and $1\leq d'\leq d'' \leq D $
	is $ \cV_x $-linearly independent for all
	$x \in \overline{\phi}\big( \overline{U}_\alpha \big) \cap \cM$ and $\alpha\in\{1,\ldots,P\}$.
\end{remark}

\begin{proposition}\label{prop:GaussKernelLinearIndependence}
	Let $ K(s,v) = e^{-(s-v)^T\Omega(s-v) / 2}$, for some positive definite matrix $ \Omega \in \mathbb{R}^{D\times D} $ be the $ D $-dimensional  Gaussian kernel. Assume that $ \cV $ is a $ D $-dimensional lattice which contains
	an element $ v $ such that
	$$ \left\lbrace v+ \sum_{d = 1}^D \lambda_d \gamma_d e_d: \gamma_d \in \lbrace -1,0,1 \rbrace\right\rbrace \subset \mathcal{V}
	$$
	where $ (e_d)_{1\leq d \leq D} $ is the standard
	basis and $ \lambda \in \mathbb{R}^D_{>0}$.
	Then $ K(s,\cdot)$, $\partial_d^s K(s,\cdot)$
	and $\partial_{d'd''}^sK(s,\cdot)$ for $d\in\{1,\ldots, D\}$ and $1\leq d'\leq d'' \leq D $
	are $ \cV_s $-linearly
	independent for each $s\in\mathbb{R}^D$.
\end{proposition}

The above conditions allow us to consider the GKF for Gaussian related fields
obtained from Gaussian SuRFs as a corollary to Theorem 12.4.2 from \cite{Adler2007}.
We define $\cA_u(f) = \{ s \in \cM:~f(s) \geq u  \}$ to be the excursion set
of a random field $f$ above the threshold $u$ on $\cM$ and write $\chi_f(u)$ to denote the
Euler characteristic (EC) of the excursion set $\cA_u(f)$.
\begin{theorem}\label{thm:GKF}
	Let $(\tilde X_1, X_1, K, \cV),\ldots, (\tilde X_N, X_N, K, \cV) \sim (\tilde X, X, K, \cV)$
	be i.i.d. SuRFs and $F\in \cC^2\big( \overline{\cM} \big)$.
	Assume that $X$ is a Gaussian field on $\cV$ with covariance
	function $\mathfrak{c}$ and that for all $v\in \cV$ it holds that
	$K(\cdot, v) \in C^3\big( \overline{\cM} \big)$
	and $\mathfrak{c}(v,v) > 0$.
	Furthermore, for all $\alpha\in\{1,\ldots, P\}$ and
	$x\in \overline{\phi}\big( \overline{U}_\alpha \big) \cap \cM$
	assume that the random vector $(X(v): v\in \cV_x)$ is non-degenerate
	for the $\cV_x$ defined in Proposition \ref{prop:nondegen}
	and that
	$\partial_d^x K_\alpha(x,\cdot)$
	and $\partial_{d'd''}^xK_\alpha(x,\cdot)$ for $d\in\{1,\ldots, D\}$ and $1\leq d'\leq d'' \leq D $
	are $ \cV_x $-linearly independent.
	Define a random field $T$ such that
	$T(s) = F\big( \tilde X_1(s)/\Vert K_s \Vert, \ldots, \tilde X_N(s)/\Vert K_s \Vert\big)$
	for all $s\in\overline{\cM}$. Then
\begin{equation}
	\mathbb{E}\left[ \chi_T(u) \right] = \sum_{d = 0}^D \cL_d \rho_d^T(u)\,,~ ~ ~u \in \mR\,,
\end{equation}
where $\cL_0,\ldots, \cL_D\in\mR$ are the LKCs of $\overline{\cM}$ endowed with the induced Riemannian metric from $\tilde X(s)/\Vert K_s \Vert$ and $\rho_d^T$'s are functions depending solely on the marginal distribution of $T$.
\end{theorem}

\subsection{Voxel Manifolds}\label{scn:voxelmanifolds}
Until now we did not discuss what domain $\cM\subset \overline{\cM}$
we should choose for a SuRF derived from a kernel $K$ and a
discrete random field $X$ on a finite grid $\cV$. 
In the case that $\overline{\cM}$ is a $D$-dimensional submanifold of $\mathbb{R}^D$ and $\cV \subset \mR^D$ and $K(\cdot, v): \mR^D \rightarrow \mathbb{R}$, for each
$v\in\cV$, we propose to use a practical domain which we call the \textit{voxel manifold associated with} $\cV$.
\begin{definition}
    Suppose that $ \cV\subset\mR^D$ is a discrete set and define
    $\delta\in \mR^D$ such that its $d$-th component is
    $
    	\delta_d = \min\big\{\, \vert  v_d - w_d \vert: v,w \in \cV\,,~v_d \neq w_d  \,\big\}
    $. Moreover, let
    \begin{equation}
         \cB_v( \delta )
         	=  \Big\{\, x \in \mathbb{R}^D ~\Big\vert~
         			\max_{ d \in \{1, \ldots, D\} } \vert x_d - v_d \vert - \delta_d/2 \leq  0
         			\,\Big\}.
    \end{equation}
    Then the \emph{voxel manifold associated with $\cV$} is the set
    $ \cM_\cV = \bigcup_{v\in\cV} \cB_v( \delta )$.
\end{definition}
A voxel manifold is a stratified space, for example, for
$D = 3$ by  the three dimensional stratum is the union of all the open
cubes $\text{int}(\cB_v(\delta))$ with
$v\in\cV$, while the two dimensional, the one dimensional and the zero dimensional strata
are the unions over all $v\in V$ of all faces, edges and corners of the cubes $\cB_v(\delta)$ respectively.
In fact, a voxel manifold as the union of polyhedra is even a WS
manifold \citep[p.187]{Adler2007}.

\begin{remark}
Figure I in \cite{Worsley1996} is somewhat suggestive of a voxel manifold. However they used in a different choice of domain, instead defining a cube by $8$ neighbouring voxels while in $\cM_\cV$ each voxel $v$ defines a cube. Crucially, they only used this choice of domain in order to estimate the intrinsic
 volumes while their test statistic is still $\max_{v \in \mathcal{V}} T(v)$, i.e. the maximum over the values that the test-statistic takes at the voxels.
\end{remark}

Throughout the rest of this article we assume that $\cB_v(\delta) \subset {\rm supp} \big( K(\cdot, v) \big)$ for all $v \in \cV$ because otherwise \textbf{(G2)} cannot be true as the corresponding SuRF is zero on parts of
$\cM_\cV$. This condition on the support, however, is usually satisfied
as the kernel $K$ is typically used to increase the signal-to-noise ratio through
averaging observations at different $v \in \cV$.

In Riemannian geometry most geometric quantities can be derived from
the Riemmanian metric and the Christoffel symbols.
Using Theorem \ref{thm:RiemannianMetric} and \ref{thm:Christoffel}
from Appendix \ref{app:InducedRiemann} these quantities  for the
Riemannian metric induced by a normalized SuRF on $\overline{\cM}$
can be written in terms of the inner products introduced in \eqref{eq:InnerProdRepr}.
 \begin{proposition}\label{cor:geometry-prop}
    For $\overline{\cM}\subset\mathbb{R}^D$ the Riemannian metric
    $\boldsymbol{\Lambda}$ induced by a normalized SuRF expressed
    in the local coordinates induced by $\iota:~\overline{\cM} \hookrightarrow \mR^D$ is given by
    \begin{equation}\label{eq:riem_SuRF}
        \Lambda_{dd'}(x) = \frac{ \ska{\partial_{d}K_x}{ \partial_{{d'}}K_x} }{ \Vert K_x \Vert^2 }
       - \frac{ \ska{ \partial_{d} K_x }{ K_x } \ska{ K_x }{ \partial_{{d'}}K_x } }{\Vert K_x \Vert^4}
    \end{equation}
     and the Christoffel symbols of the first kind are
    \begin{equation}\label{eq:Christoffel_lin_smoothers}
        \begin{aligned}
                \Gamma_{kdd'}(x)
                &= \frac{ \ska{\partial_{k}\partial_{d}K_x}{ \partial_{{d'}}K_x}  }{ \Vert K_x \Vert^2}
                 - \frac{ \ska{ \partial_{k}\partial_{d} K_x }{ K_x } \ska{ K_x }{ \partial_{{d'}}K_x } }{\Vert K_x \Vert^4} \\
              &\quad - \frac{\ska{ \partial_{k} K_x}{ K_x }\ska{\partial_{d}K_x}{ \partial_{{d'}}K_x}}{\Vert K_x \Vert^4}
              - \frac{\ska{ \partial_d K_x }{ K_x } \ska{ \partial_{k} K_x }{ \partial_{d'} K_x } }{\Vert K_x \Vert^4}\\
              &\quad+ 2\frac{ \ska{\partial_{k}K_x}{K_x} \ska{ \partial_{d}K_x}{K_x} \ska{K_x}{\partial_{d'}K_x}}{\Vert K_x \Vert^6}
              \end{aligned}
    \end{equation}
 \end{proposition}

The advantage of using the voxel manifold domain for a SuRF is that the numerical implementation of geometric quantities is feasible and therefore estimators of the LKCs can be calculated efficiently,
because all $d$-dimensional boundaries lie in hyperplanes parallel to the coordinate axes.
It is helpful here that we can construct for any $x\in \mathcal{U}$ and a small open neighborhood
$\mathcal{U}\ni x$ an orthonormal basis of $\cT_x\cM_\cV$, $x \in \cM_\cV\cap \mathcal{U}$,
such that a subset of this basis is an orthonormal frame of $\cT_x\partial_d\cM_\cV$, $d\in \{1,\ldots, D-1\}$, if
$x \in \partial_d\cM_\cV\cap \mathcal{U}$.
More concretely, for $D=3$ the Gram-Schmidt procedure on the Euclidean basis $e_1,e_2,e_3$ yields the following orthonormal vector fields with respect to the Riemannian metric induced by a SuRF at $x \in \cM_\cV$ and
$ k, l, m \in \{1,2,3\} $ such that $k < l$ and $ \{ k, l, m \} = \{ 1, 2, 3 \} $:
\begin{equation}\label{prop:ONF}
\begin{split}
        U_x = \Lambda_{kk}^{-1/2}(x) &e_k\,,~~
        V_x = \frac{\Lambda_{kl}(x)}{\sqrt{ c(x) \Lambda_{kk}(x) }}
        		e_k - \sqrt{ \frac{\Lambda_{kk}(x)}{c(x)}  } e_l\,,\\
        &N_x = \frac{\boldsymbol{\Lambda}^{-1}(x)}{\sqrt{ e_m^T \boldsymbol{\Lambda}^{-1}(x) e_m }} e_m\,.
\end{split}
\end{equation}
Here $c(x) = \det\big( \boldsymbol{\Lambda}^I(x)\big)$ for $ I = (k, l) $ and $U_x$, $V_x$ are a basis of $\cT_x\cF_I$,
where $\cF_{I}$ is the subset of $ \partial_{2} \cM_\cV $
such that the coordinates with indices not contained in $I$ are constant, and $N_x$
is in the one dimensional vector space orthogonal to $\cT_x\cF_I$ with respect to $\boldsymbol{\Lambda}$.

\subsection{ LKCs of Voxel Manifolds }\label{sec:LKC_VM}
General formulas for the LKCs of WS manifolds can be
found in \cite[Theorem 12.4.2]{Adler2007} and formulas for
WS manifolds of dimension $D \leq 3$ are given
in Appendix \ref{A:LKCdefn}.
Since voxel manifolds are embedded into $\mathbb{R}^D$
and have a simple geometric structure, the highest two LKCs
can be expressed as integrals of (sub-)determinants of the Riemannian
metric $\boldsymbol{\Lambda}$, i.e.,
\begin{align*}
	\cL_{D} =
		\sum_{ \nu\in\cV } \int_{ \cB_\nu( \delta ) }
			\sqrt{ \det\big( \boldsymbol{\Lambda}(x) \big) }\,dx\,,
	~ ~ ~
	\cL_{D-1} =
		\sum_{I:\vert I \vert = D-1}
					\int_{\cF_I}
						\sqrt{ \det\big( \boldsymbol{\Lambda}^{I}(x) \big) }\,dx^I\,.
\end{align*}
Here $ I $ is any ordered subset of $ \lbrace 1, \dots, D\rbrace $,
$dx^I = dx_{I_1}\ldots dx_{I_{\vert I \vert}}$ and
$\boldsymbol{\Lambda}^{I}(x)$ is the submatrix of $\boldsymbol{\Lambda}(x)$
consisting of the columns and rows given by the entries of $I$.
Finally, $\cF_{I}$ is the subset of $ \cM_\cV \setminus \partial_{D} \cM_\cV $
such that the coordinates with indices not contained in $I$ are constant.

The LKCs $\cL_{1},\ldots,\cL_{D-2}$ of a $D$-dimensional voxel manifold with $D>2$
are substantially harder to express explicitly, compare \cite[Theorem 12.4.2]{Adler2007} and
Theorem \ref{thm:L1VM} in Appendix \ref{A:LKCdefn} for $3$-dimensional voxel manifolds.

\subsection{SuRF Estimator for LKCs}\label{scn:SuRFLKCestim}
Given a kernel
$
K:\cM_\cV \times \cV \rightarrow \mathbb{R}
$ and an i.i.d. sample $X_1, \ldots, X_N $ of a Gaussian random field $X$
over $ \cV $, we obtain from formula \eqref{eq:riem_SuRF}
for the Riemannian metric $\boldsymbol{\Lambda}(x)$ 
the SuRF-Riemannian metric estimator
\begin{equation}\label{eq:riem_SuRFest}
\begin{split}
        \hat{\Lambda}_{dd'}(x) =
        	&\frac{ \Cov{ \partial_{d}\tilde{\mathbb{X}}(x)}{ \partial_{d'}\tilde{\mathbb{X}}(x) } }
        		 { \Var\left[ \tilde{\mathbb{X}}(x) \right] }\\
       	  &\quad- \frac{ \Cov{ \partial_{d}\tilde{\mathbb{X}}(x)}{ \tilde{\mathbb{X}}(x) }
       	  		   \Cov{ \tilde{\mathbb{X}}(x)}{ \partial_{d'}\tilde{\mathbb{X}}(x) } }
       	  		 { \Var\left[ \tilde{\mathbb{X}}(x) \right]^2 }\,.
\end{split}
\end{equation}
Here the variances and covariances are sample variances and covariances of the sample
$\tilde{\mathbb{X}} = ( \tilde X_1, \ldots, \tilde X_N )$ and its derivatives
$\partial_{d}\tilde{\mathbb{X}}(x)$. The latter can
be computed from the derivatives of the kernel $K$, compare Proposition \ref{prop:KernelPropConvFields}.
Hence, \eqref{eq:riem_SuRFest} uses the exact derivatives of the sample fields. We denote with
$\hat{\boldsymbol{\Lambda}}(x) \in \mathbb{R}^{D \times D}$, $x\in \cM_\cV$,
the matrix with $(d,d')$-th entry $\hat{\Lambda}_{dd'}(x)$.
To estimate the LKCs we evaluate $\hat{\boldsymbol{\Lambda}}(x)$ on a grid
$\cM_\cV^{(r)} \subset \cM_\cV$ given by
\begin{equation}\label{eq:grid}
	\cM_\cV^{(r)}
		= \bigcup_{v\in\cV} \cB_v( \delta )
						\cap \Big( v + \tfrac{ \delta }{ r + 1 } \cdot \mathbb{Z}^D \Big)\,.
\end{equation}
Here $r \in \big\{\, 2r'+1 ~\vert~ r'\in\mathbb{N} \,\big\}$
is called the \emph{added resolution}. The restriction to odd numbers greater zero is necessary to ensure
that the boundary of $\cM_\cV$ is sampled.
Consequentially, the \emph{SuRF-LKC estimators with added resolution $r$} for $\mathcal{L}_D$ and $\mathcal{L}_{D-1}$  are given by
\begin{equation}\label{eq:LKCestim_VM}
\begin{split}
	\hat{\cL}_D^{(r)}
		&= \sum_{ x \in \cM_\cV^{(r)} }
				\sqrt{ \det\left( \hat{\mathbf{\Lambda}}(x) \right) }
									\prod_{d=1}^D \frac{\delta_d}{r+1}\,.\\
	\hat{\cL}_{D-1}^{(r)}
		&= \sum_{ \vert I \vert = D-1 }\sum_{ x \in \cF_{I}^{(r)} }
						\sqrt{ \det\left( \hat{\boldsymbol{\Lambda}}\vphantom{c}^{I}(x) \right) }
							\prod_{i \in I} \frac{\delta_i}{r+1}\,.
\end{split}
\end{equation}
Here $\cF_{I}^{(r)} = \cF_{I} \cap \cM_\cV^{(r)}$.
These formulas are easy to implement and
as $r$ goes to infinity the numerical error in approximating the
integral by a Riemann sum becomes arbitrary small. In the
same fashion it is possible to obtain estimators $\hat{\cL}_{D-d}^{(r)}$, $d \in \{ 1,\ldots D-2 \}$,
from Theorem 12.4.2 in \cite{Adler2007}
and our formulas of the geometric quantities induced by a
SuRF on $\cM_\cV$,
compare Corollary \ref{cor:geometry-prop} and Appendix \ref{A:LKCdefn}.
In practice these estimators are tedious to implement and
they are computationally costly,
because even for a $3$-dimensional voxel manifold
we need for the Christoffel symbols $27$ convolutions on top of the $9$ convolutions required to estimate the Riemannian metric. Computing the Riemannian curvature tensor needs another $36$
convolutions.
A solution is to approximate the
lower LKCs by their locally stationary counterparts, i.e., for $D=3$
only the first integral in Theorem \ref{thm:L1VM} remains:
\begin{equation}
\begin{aligned}
	\hat{\cL}_1^{(r)}
		= \sum_{ \vert I \vert = 1 }\sum_{ v \in \cF_{I}^{(r)} }
			\hat\Theta(v)\sqrt{ \hat{\boldsymbol{\Lambda}}\vphantom{v}^{I}(v) }
				\frac{ \delta_I }{ r + 1 }\,.
\end{aligned}
\end{equation}
Here $\hat\Theta$ is the plug-in estimate of $\Theta$ defined in Theorem \ref{thm:L1VM}. 

\begin{theorem}\label{thm:unbiased}
	Let $K$ be the kernel of the SuRF and $d\in \{D-1,D\}$. Assume that $K(\cdot, v)\in  \mathcal{C}^3$ for all $v\in\mathcal{V}$ and that
	$\mathbb{E}\left[ X(v) \right] < \infty$  for all $ v \in \cV $.
	Then
	\begin{equation*}
		\lim_{r\rightarrow \infty} \mathbb{E}\left[ \hat{\cL}_d^{(r)} \right]
		= \mathbb{E}\left[  \lim_{r\rightarrow \infty} \hat{\cL}_d^{(r)} \right]
		= {\cL}_d\,.
	\end{equation*}
\end{theorem}

\begin{theorem}\label{thm:consistency}
	Let $d\in\! \{D-1,D\}$, $ K(\cdot, v) \in \mathcal{C}^3\big(\overline{\cM}\big) $
	and $\mathbb{E}\left[ X(v)^2 \right] < \infty$ for $ v \in \cV $. 
	Then
	\begin{equation}
		\lim_{N\rightarrow \infty} \lim_{r\rightarrow \infty}\hat{\cL}_d^{(r)}
		= \lim_{r\rightarrow \infty} \lim_{N\rightarrow \infty} \hat{\cL}_d^{(r)} = \cL_d\,.
	\end{equation}
\end{theorem}
\begin{remark}
	We expect that a similar result can be derived for the plugin estimator  $\cL_1$
	resulting from Theorem \ref{thm:L1VM} where
	also the $\Gamma_{kdd'}$'s from \eqref{eq:Christoffel_lin_smoothers} and the Riemannian
	curvature, Appendix \ref{app:InducedRiemann},
	are estimated using the corresponding sample covariances.
	This result could be established along the same lines as the consistency
	in \cite[Section 3]{Telschow2020}, but we leave this for future work since,
	currently, implementing this estimator seems infeasible.
\end{remark}

\subsection{FWER Control Using SuRFs}\label{scn:FWER}
In this section we illustrate how the GKF is used to construct a test controlling the family-wise error rate (FWER) which dates back to \cite{Worsley1992} and combine this approach with SuRFs in order to
improve the power of RFT based voxelwise inference.
We illustrate this on the problem of detecting areas of non-zero signal $\mu$, given
an $ i.i.d. $ sample $X_1,\ldots,X_N \sim X$, in a signal plus noise model
$X(v) = \mu(v) + \epsilon(v)$ for $v\in\cV$
and $\epsilon$ a zero-mean random field. The same approach, however, can be applied to
the linear model, e.g., Supplementary of \cite{Telschow2020}, and other probabilistic models for which there is a GKF, e.g., \cite{Worsley1994, Taylor2007, Taylor2008}.

Let $\tilde X_1, \ldots, \tilde X_N$ denote a sample of SuRFs over $\cM_\cV$
derived from random fields $X_1,\ldots, X_N$ over $\cV$.\footnote{Technically, stating the GKF requires the SuRFs to be defined on a $D$-dimensional, compact manifold without boundary $\overline{\cM_\cV} \supset \cM_\cV$}
Assume the setting of Theorem \ref{thm:GKF} and define $ F:\mathbb{R}^N \rightarrow \mathbb{R} $
 by
\begin{equation*}
F(a_1, \dots, a_N) = \frac{1}{\sqrt{N}}\sum_{i = 1}^N a_i\left( \frac{1}{N-1}\sum_{i = 1}^N \left( a_i - \frac{1}{N}\sum_{i = 1}^N a_i \right)^2 \right)^{-1/2}
\end{equation*}
for $ (a_1, \dots, a_N) \in \mathbb{R}^N $. Let $ \mu(v) =\mathbb{E}\left[  X(v)  \right] $ for $ v \in \cV $
and for all $ x \in \cM_\cV $, let
\begin{equation*}
	\tilde\mu(x) = \sum_{v \in \cV} K(x,v) \mu(v).
\end{equation*}
 Due to the invariance of $ F $ to scaling, for $ x \in \cM_\cV$, we can write
\begin{equation}\label{eq:tfield}
	\tilde T(x) := F\big( \tilde X_1(x)/\Vert K_x \Vert, \ldots, \tilde X_N(x)/\Vert K_x \Vert\big) = \frac{\sqrt{N}\hat{\mu}_N(x)}{\hat{\sigma}_N(x)}
\end{equation}
where $ \hat{\mu}_N(x) = \frac{1}{N}\sum_{i = 1}^N \tilde{X}_i(x) $ and $ \hat{\sigma}_N(x) = \left(\frac{1}{N-1}\sum_{i = 1}^N (\tilde{X}_i(x) - \hat{\mu}_N(x))^2\right)^{1/2} $.

Our goal is to construct based on $\tilde T$ a multiple hypothesis test for
the hypotheses
\begin{equation}\label{eq:testetss}
\mathbf{H}_{0}^x:~\tilde\mu(x) \leq 0
	\quad vs. \quad 
\mathbf{H}_{1}^x:~\tilde\mu(x) > 0\,\quad\quad x\in \cM_\cV
\end{equation}
which controls the FWER in the strong sense at a significance level
$\alpha\in (0,1)$.\footnote{The two-sided hypothesis can be treated similarly.} 
Denote with $\mathcal{H}_0 = \big\{ x \in \cM_\cV ~\vert~ \tilde\mu(x) \leq 0 \big\} \subseteq \cM_\cV$ the set of true null hypotheses and consider the test that rejects $\mathbf{H}_{0}^x$, $x\in\cM_\cV$,
whenever $ \tilde T(x)  > u_\alpha$.
Then its FWER at $ u \in \mathbb{R} $ is
\begin{equation*}
	\text{FWER}_{\tilde T}(u) = \mathbb{P}\left( \sup_{x \in \mathcal{H}_0} \tilde T(x)  > u \right).
\end{equation*}
Defining $\tilde T_0 = \sqrt{N}\big( \hat\mu_N(x) - \tilde\mu(x) \big)/\hat{\sigma}_N(x)$, then any $u_\alpha$ satisfying
$$
\mathbb{P}\Big( \max_{x \in \cM_\cV} \tilde T_0(x)  > u_\alpha \Big) \leq \alpha
$$
for $ \alpha \in (0,1) $, controls the FWER in the strong sense at the level $\alpha$ because
$$
	\text{FWER}_{\tilde T}(u_\alpha) \leq \mathbb{P}\left( \max_{x \in \cM_\cV} \tilde T_0(x)  > u_\alpha \right)  \leq \alpha
$$
for all $\mathcal{H}_0\subseteq \cM_\cV$ and $\tilde T_0(x) = \tilde T(x)$ for all $x\in \mathcal{H}_0$.
To find such a threshold $u_\alpha$ we use the EEC heuristic \cite{Taylor2005} to the excursion probability.
In particular, letting $ M_u(\tilde T_0) $ be the number of local maxima of $ \tilde T_0 $
over $\cM_\cV$ that lie above the level $ u $ and $\chi_{\tilde T_0}(u)$ be the EC
of the excursion set $\big\{ x \in\cM_\cV~\vert~\tilde T_0 (x) > u \big\}$  we have that
\begin{equation}\label{eq:empGKF}
	\mathbb{P}\left( \max_{x \in \cM_\cV}  \tilde T_0(x) >  u \right)
		\leq \mathbb{E}\big[ M_{u}(\tilde T_0) \big]
		\approx \mathbb{E}\left[ \chi_{\tilde T_0}(u) \right]
		= \sum_{d = 0}^D \mathcal{L}_d \rho_d^{\tilde T_0}(u).
\end{equation}
Here $\rho_d^{\tilde T_0}$ are the EC densities of a (centered) $t$-field given, compare \cite[p.915]{Taylor2007}.
In order to control the FWER in the strong sense to a level $ \alpha \in (0,1) $,
we find the largest $ u_\alpha $ such that
$ \sum_{d = 0}^D \hat{\mathcal{L}}_d \rho_d^T(u_\alpha) = \alpha$ using the $\hat{\mathcal{L}}_d$'s
from Section \ref{scn:SuRFLKCestim}.
At high thresholds $ u_\alpha $ the number of local maxima is either zero or one
and so $ \mathbb{E}\left[ \chi_{\tilde{T}_0}(u_\alpha) \right] $ is an
extremely good approximation to $ \mathbb{E}\big[ M_{u_\alpha}(\tilde{T}_0) \big] $.
Lower values of $ \alpha $ yield higher thresholds $u_\alpha$. At a typical value $ \alpha \leq 0.05 $,
we expect the approximation in \eqref{eq:empGKF} to be accurate.

Traditional RFT inference in neuroimaging \citep{Worsley1992,Worsley1996,Taylor2007} uses the same framework but only evaluates the fields on the lattice $\cV$ and uses the LKC estimators given in \cite{Forman1995}, \cite{Kiebel1999} or \cite{Taylor2007} which are based on discrete derivatives. More precisely, for each $ n \in\{1,\ldots, N\} $ it takes data $ X_n $ on a lattice
$ \cV $ (corresponding to the centers of voxels making up the brain), smoothes it with a kernel $ K $ to obtain
$ \lbrace \tilde X_n(v): v \in \cV\rbrace $ and rejects all $v\in \cV$ such that
$\tilde T(v) > u_\alpha$ where $u_\alpha$ is obtained from the GKF approximation such that
$\mathbb{P}\big( \max_{s \in \cM}\tilde T_0(s) > u_\alpha \big) \approx \alpha$.
Here $\cM\subset\mathbb{R}^3$ represents for example the brain, but
 has never been defined precisely in the literature as it probably
was assumed to be unnecessary by the good lattice assumption.
By construction this leads to valid, but conservative inference since
$
\max_{v \in \cV} \tilde T_0(s) \leq \max_{s \in \cM} \tilde T_0(s)
$
and $u_\alpha$ approximates the tails of the distribution of
$\max_{s \in \cM} \tilde T_0(s)$ at level $\alpha$ and not the tails of
the distribution of $\max_{v \in \cV} \tilde T_0(v)$.
SuRFs allow to remove this conservativeness by specifying $\cM$ to be
the voxel manifold $\cM_\cV$ and testing $\tilde T(x) > u_\alpha$ for all $x\in\cM_\cV$ meaning that
$u_\alpha$ approximates the quantiles of the test statistic.
Consequentially, our SuRF framework has strong FWER control over $\cM_\cV$ at
level $\alpha$ up to the approximation in \eqref{eq:empGKF} and thus has a
higher power than traditional RFT.

\paragraph{Effect Localization}
In the discussed framework as often in applications the discrete data has been smoothed before carrying out statistical inference. Therefore precise localization
of significant effects, i.e., finding $x \in \cM_\cV$ such that $\tilde \mu(x) \neq 0$, is
only possible for the smoothed signal; yet weaker localization results
for the $v\in \cV$ such that $\mu(v) \neq 0$ hold.\footnote{Our arguments carry over to tests controlling the FWER in the strong sense w.r.t.
the smoothed signal, for example, permutation tests in fMRI as they are typically applied to
smoothed data.}
The key observation is that, if $K(x,v) \geq 0$ for all $x\in \cM_\cV$ and all $v\in\cV$, it holds that
\begin{equation}\label{eq:identity_support}
	\tilde\mu(x) > 0
	~~ \Longleftrightarrow ~~
	\exists v\in\cV \cap \supp\big( K(x,\cdot) \big):~ \mu(v) > 0\,.
\end{equation}
Thus, if we reject $\mathbf{H}_{0}^x$ we can conclude by \eqref{eq:identity_support}
that there is at least one $v\in\cV \cap \supp\big( K(x,\cdot) \big)$ such
that $\mu(v) > 0$ and the strong control on the hypotheses \eqref{eq:testetss} yields
\begin{equation*}
\begin{split}
	\mathbb{P}\Big(\, \big\{ x'\in \cM_\cV ~\vert~ \mathbf{H}_{0}^{x'}\text{ is rejected} \big\}
						\subseteq \cM_\cV \setminus \mathcal{H}_0 \,\Big)
	&= 1 - \mathbb{P}\Big(\, \exists x\in \mathcal{H}_0:~
				 \mathbf{H}_{0}^x\text{ is rejected} \,\Big)\\			
	&\geq 1-\alpha\,.
\end{split}
\end{equation*}
Consequentially, $\tilde T(x) > u_\alpha$ implies that the probability of incorrectly
claiming that $\mu(v) > 0$ for a $v \in \cV \cap \supp\big( K(x,\cdot) \big)$ is at most $\alpha$.
This is a weaker form of FWER control than controlling it in the strong sense,
but stronger than controlling it in the weak sense as long as $\supp\big( K(x,\cdot) \big) \neq \cM_\cV$ for at least one $x\in \cM_\cV$.

\begin{remark}
Losing strong control with respect to $ \mu $ due to smoothing the data
is natural in applications such as fMRI. Smoothing is needed to increase the low
signal-to-noise ratio and it is even debatable whether it is plausible to talk about the BOLD
signal at a single voxel in a realistic fMRI experiment as the analyzed BOLD signal at $v$ is
a distorted version of the observed data due to an extensive preprocessing
pipeline which includes, among other aspects, motion correction and warping to a standardized brain.
\end{remark}

\section{Simulations}\label{Scn:Sims}
In this section we compare using simulations the performance of
the SuRF estimator of the LKCs from Section \ref{scn:SuRFLKCestim}
to existing LKC estimators; namely, the Hermite projection estimator (HP) and its bootstrap improvement (bHP) from \citep{Telschow2020} and the LKC estimators developed for stationary processes from \cite{Kiebel1999} and \cite{Forman1995} which are used in established software, e.g., SPM and FSL. We abbreviate the latter two estimators as \textit{Kie} and \textit{For} respectively.
We do not compare to the warping estimator \cite{Taylor2007} since the estimates,
although computed differently, are almost identical to the estimates of the bHP, see \citep{Telschow2020}.

Our second set of simulations show that the FWER control of our SuRF framework is not conservative
for Gaussian data, while the traditional approach implemented in software such as SPM and FSL is conservative.
In part 2 \cite{Davenport2023} we extend our method to non-Gaussian data and demonstrate that it accurately controls the FWER even
on a gold-standard data set consisting of $7.000$ resting state experiments from the UK Biobank.

\subsection{Simulation Setup}\label{scn:setup}
In our simulations we consider the data to be samples of the Gaussian random field $\lbrace X(v): v \in \cV\rbrace$ where $\mathcal{V} \subset \mathbb{R}^D$ and the $X(v)$'s are i.i.d. $\mathcal{N}(0,1)$ distributed.
To transform $X$ into a SuRF we smooth this discrete data by using the isotropic Gaussian kernel
\begin{equation}\label{eq:smooth_kernel}
	K_f(x,v) = e^{ \frac{ 4\log(2) \Vert x - v \Vert^2}{f^2}}\,,~~~
	f \in \{ 1,2,,\ldots,6 \}\,,
\end{equation}
which we parametrize by its full width at half maximum (FWHM) $f$ as it is common in neuroimaging. The performance
of different LKC estimators is compared on a standardized almost stationary SuRF and a standardized non-stationary SuRF, which we define below.

\begin{figure}[ht]\centering
\begin{tabular}{cccc}
	\raisebox{-0.5\height}[0pt][0pt]{\includegraphics[trim=0 0 0 0,clip,width=1.2in]{\figurepath 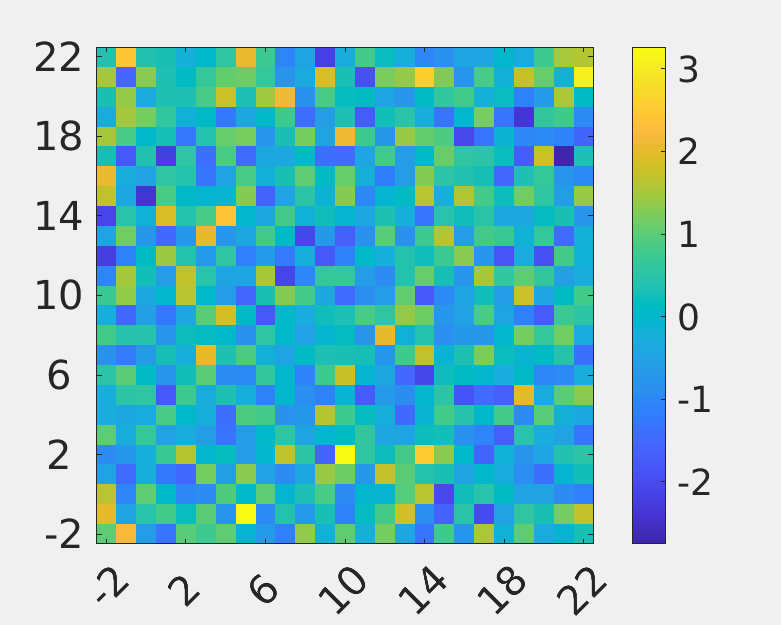}}
	& \raisebox{2.5\height}[0pt][0pt]{$\boldsymbol{\nearrow}$}
	& \includegraphics[trim=0 0 0 0,clip,width=1.2in]{\figurepath 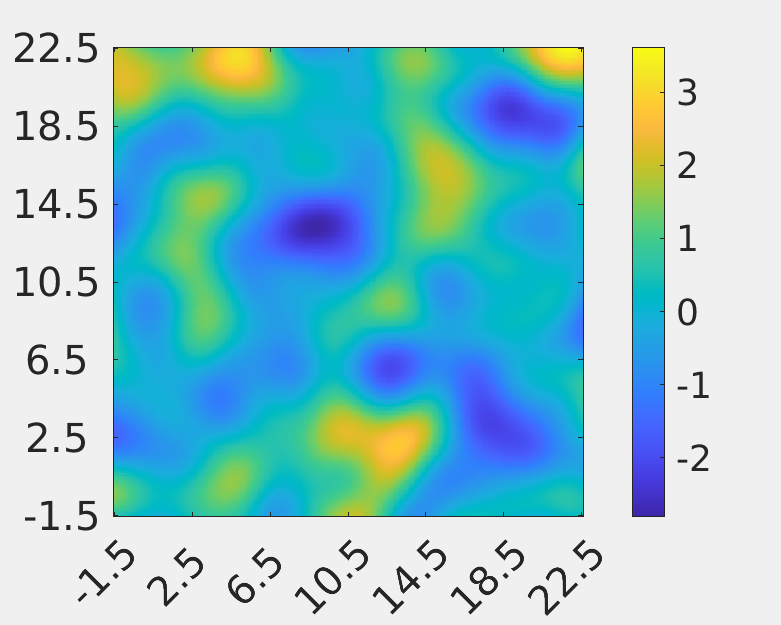}
	& \includegraphics[trim=0 0 0 0,clip,width=1.2in]{\figurepath 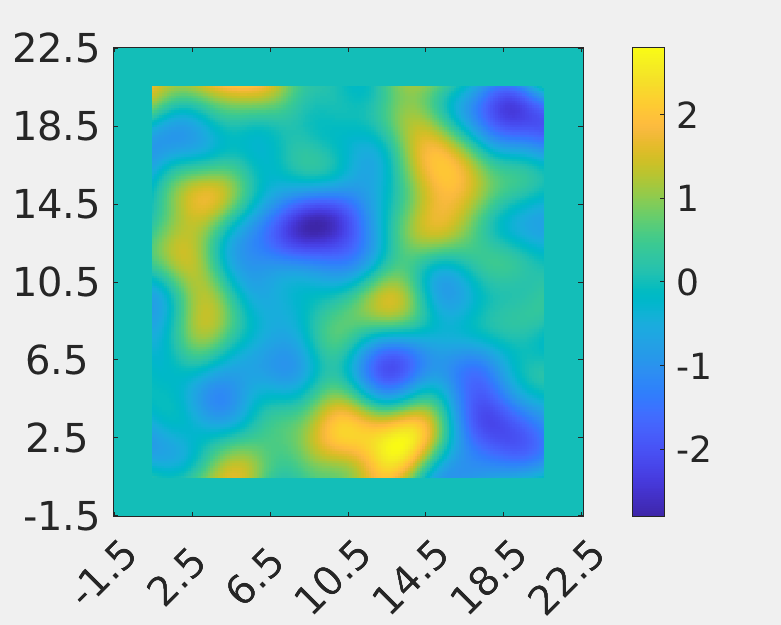}\\
	& \raisebox{6\height}[0pt][0pt]{$\boldsymbol{\searrow}$}
	& \includegraphics[trim=0 0 0 0,clip,width=1.2in]{\figurepath 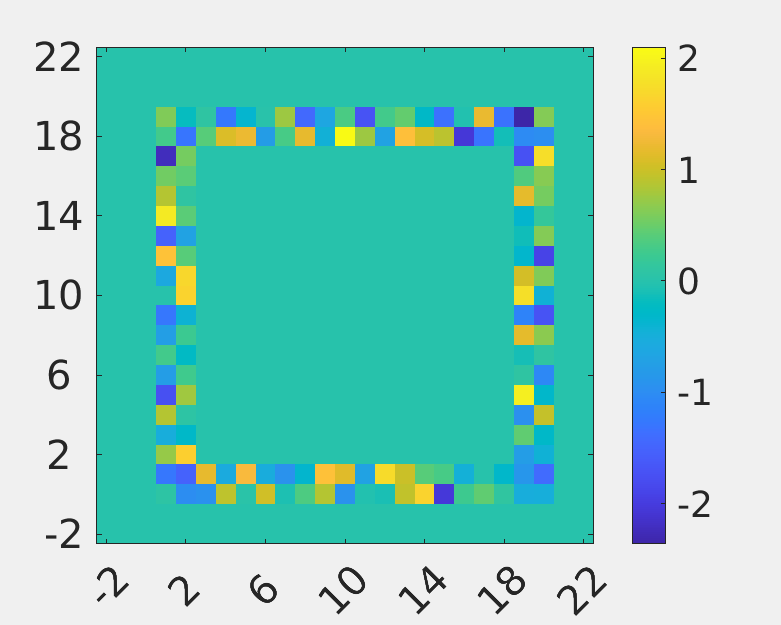}
	& \includegraphics[trim=0 0 0 0,clip,width=1.2in]{\figurepath 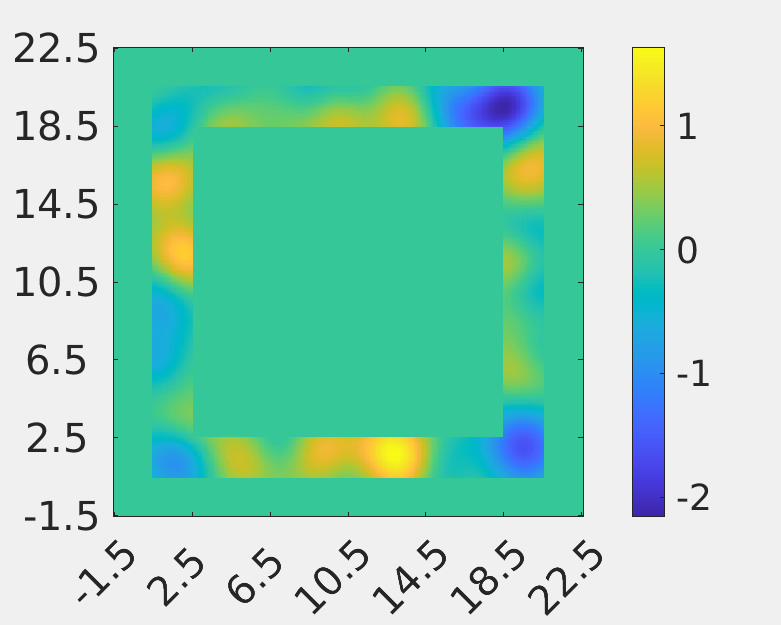}
\end{tabular}
	\caption{Illustration of generating an almost stationary or non-stationary SuRF from
	a white noise field on a grid using $K_3$. The almost stationary field results from a convolution with the kernel $K_3$ before restricting the domain of the field. This circumvents boundary effects	(top row). The  non-stationary field (bottom row) results from first restricting to $\cV_2$ and then convolving with the kernel $K_3$. The edge effects produces a non-stationary field.
	\label{fig:DataGen}}
\end{figure}

The advantage of starting from the field $X$ which is i.i.d. on the lattice is that we can
determine the theoretical LKCs for the normalized almost stationary and the
normalized non-stationary SuRF quickly with high precision on a computer since the
double sum in the Riemannian metric induced by the SuRF reduces to a single sum
due to the lack of correlation, compare Proposition \ref{cor:geometry-prop}
and equation \eqref{eq:InnerProdRepr}.

For the almost stationary simulations, for $ D \in \lbrace 1, 2, 3 \rbrace$, $a \geq 0$, we take $\mathcal{V} = \cV_D^a \subset\mathbb{R}^D$, where 
$
\cV^{a}_1 = [1 - a, 100 + a]   \cap \mathbb{Z}$, 
$\cV^{a}_2 = [1 - a,  20 + a]^2 \cap \mathbb{Z}^2$ and 
$\cV^{a}_3 = [1 - a,  20 + a]^3 \cap \mathbb{Z}^3
$.
The almost stationary SuRF is given by $(\tilde X, X, K_f, \cV^{a}_D )$ with $a = \sqrt{2}f / \sqrt{\log(2)}$ and $\tilde X$ where we choose $\cM_{\cV^{0}_D}$ to be the domain of $ \tilde{X} $. We expand $\cV_D^0$ by $a$ in each direction and restrict $ \tilde{X} $ to the voxel manifold $ \cM_{\cV^{0}_D} $ to remove boundary effects which allows a comparison with LKC estimators for stationary fields. We call this SuRF almost stationary because it is a stationary field on $\cV^{0}_D$, but non-stationary on $\cM_{\cV^{0}_D}$.
However, for $f$ larger than $\approx 2$ the LKCs of this field are almost identical to the
LKCs of a stationary, unit variance Gaussian random field defined on $\cM_{\cV^{0}_D}$ with the covariance function $c(s, s') = K_f(s,s')$ with $s,s' \in \cM_{\cV^{0}_D}$. This can be verified with the RFTtoolbox \cite{RFTtoolbox}
as the theoretical LKCs of such fields can be approximated with high precision, see Appendix \ref{App:LKCapprox}, and
is also supported by our observation in \cite{Davenport2023} that the SuRF LKC estimators yield precise
FWHM estimates even at $f\approx 2.5$ where state-of-the-art methods are biased. The LKCs for stationary,  unit-variance random fields over convex domains are well known \cite{Worsley1996} and smoothing white noise on a grid
with a Gaussian kernel is the typical way that LKC estimators and the FWER have been validated in the past \citep{Nichols2003,Hayasaka2003,Taylor2007b}. As data in practice is non-stationary, we consider also simulations
involving non-stationarity SuRFs. We simulate them by deliberately
not correcting for boundary effects when smoothing the white noise field on the grid, i.e.,
we use $ \mathcal{V}_D = \mathcal{V}_D^0 $ for $ D \in \lbrace1,2,3\rbrace $ with
\begin{equation*}
\begin{split}
	\mathcal{V}_1 = & ([1,100] \cap \mathbb{Z}) \setminus \big\{ 2, 4, 8, 9, 11, 15, 20, 21, 22, 40,\ldots, 45, 60, 62, 64, 65, 98, \ldots, 100\big\}\\	
	\mathcal{V}_2 = &\big\{ x \in \mathbb{R}^2 ~\vert~
							x_1 \in \{ 1,2, 19, 20 \} ~\vee~ x_2 \in \{ 1,2, 19, 20 \}
							 \big\} \cap [1, 20]^2\\
	\mathcal{V}_3 = &\big\{ x \in \mathbb{R}^3 ~\vert~
							x_1 \in \{ 1,2, 19, 20 \} ~\vee~ x_2 \in \{ 1,2, 19, 20 \}~\vee~ x_3 \in \{ 1,2, 19, 20 \}
							 \big\}\cap [1, 20]^3
\end{split}
\end{equation*}
and use $\cM_{\cV_D}$ as the domain of the SuRFs. The boundary effect produces a non-stationary SuRF
as the weighted average is over a variable number of voxels depending on $s\in \cM_{\cV_D}$.
An illustration of the considered sample fields for $ D = 2 $ is illustrated in Figure \ref{fig:DataGen}.

In each considered simulation settings, given a sample size $ N \in \mathbb{N} $, we generate
SuRFs $\tilde X_1,\ldots, \tilde X_N$ of SuRFs from $X_1,\ldots, X_N \sim X$ as described above.

\FloatBarrier
\subsection{Results of the LKC Estimation}\label{scn:simLKC}
\FloatBarrier
For resolutions $r \in \{ 1, 3, 5 \}$ and $D = 2$  we compare the SuRF LKC estimator,
see Section \ref{scn:SuRFLKCestim}, to the HP, bHP, Kiebel and Forman estimator obtained from
samples of SuRFs evaluated on the grids with added resolution $r$.
The results for $ D =1,2 $ are qualitatively similar, see Appendix \ref{appendix:AddSimulations}. In each setting, we run 1000 simulations in which we generate $N\in\{20,50,100\}$ SuRFs, compare Section \ref{scn:setup}, and estimate the LKCs from these fields.
\begin{figure}[h]\centering
\includegraphics[trim=0 0 40 0,clip,width=1.8in]{\figurepathh 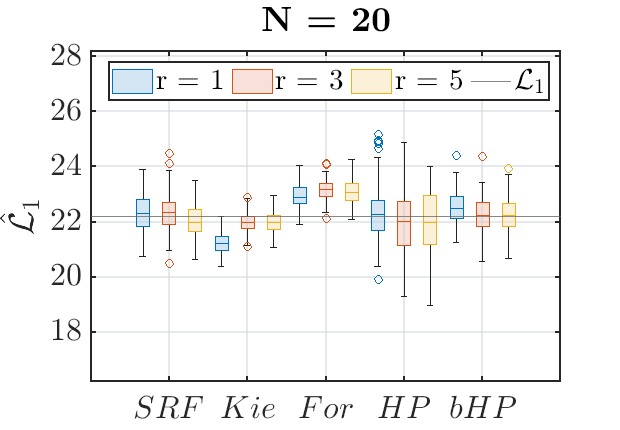}
\includegraphics[trim=0 0 40 0,clip,width=1.8in]{\figurepathh 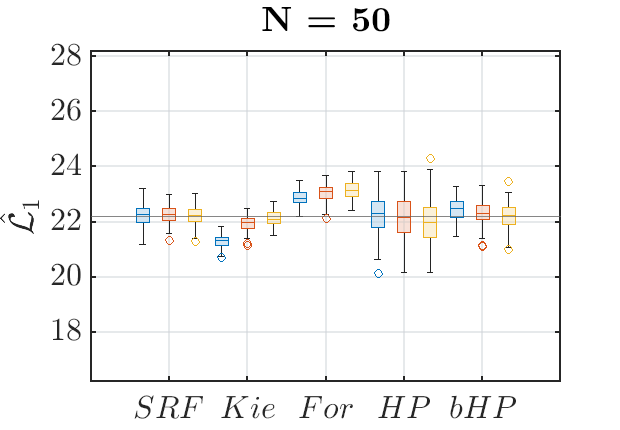}
\includegraphics[trim=0 0 40 0,clip,width=1.8in]{\figurepathh 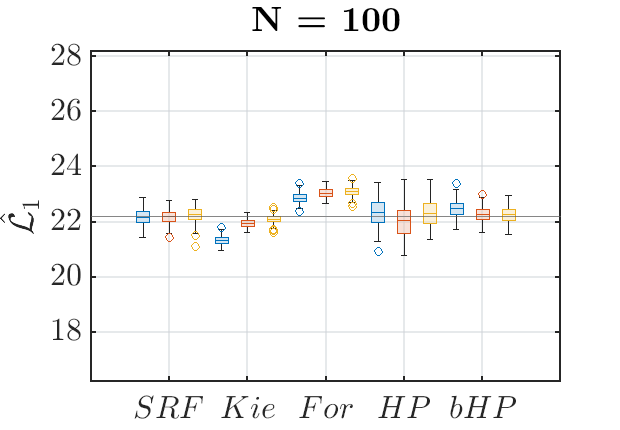}\\
\includegraphics[trim=0 0 40 0,clip,width=1.8in]{\figurepathh 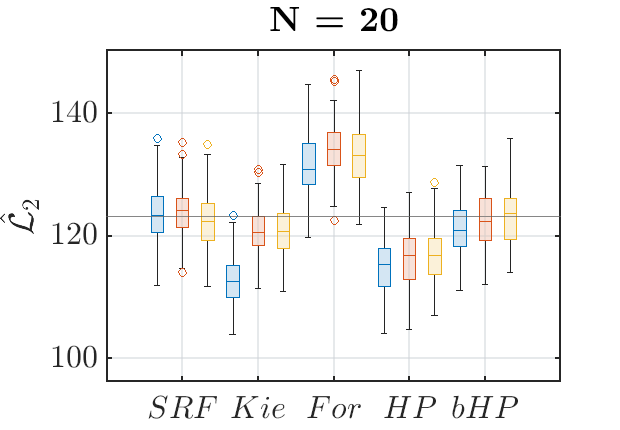}
\includegraphics[trim=0 0 40 0,clip,width=1.8in]{\figurepathh 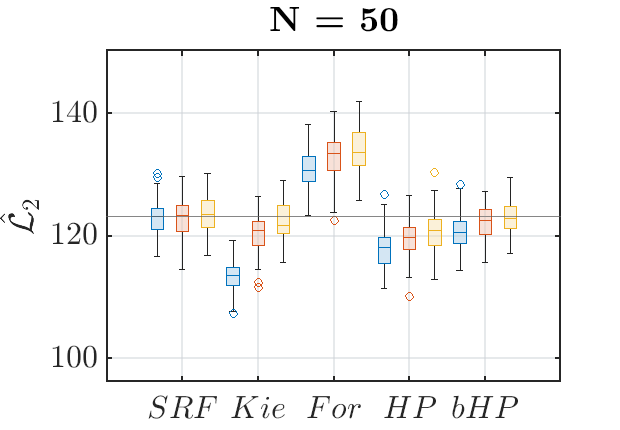}
\includegraphics[trim=0 0 40 0,clip,width=1.8in]{\figurepathh 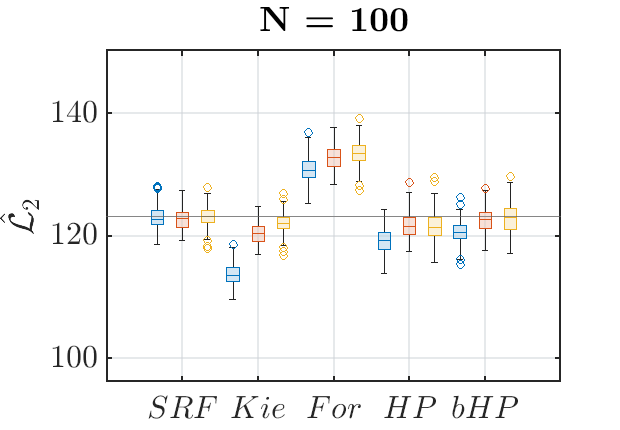}\\
	\caption{ 2D Simulation results for estimation of the LKCs
			  of the almost stationary SuRF described in Section \ref{scn:setup}.
			  The smoothing bandwidth is $f = 3$.
			  \label{fig:D2FixedFWHM_stat}}
\end{figure}

Figure \ref{fig:D2FixedFWHM_stat} shows boxplots of the LKC estimates for $f=3$
and varying sample
sizes $N\in\{20, 50, 100\}$ in the almost stationary setting. Figure \ref{fig:D2FixedFWHM_nonstat} contains the same results for the non-stationary setting. At all resolutions the SuRF estimator seems to be unbiased and has a lower variance than the other estimators. Only the bHP estimator is comparable efficient, however, it has a small bias at resolution $r=1$.
\begin{figure}[ht]\centering
\includegraphics[trim=0 0 40 0,clip,width=1.8in]{\figurepathh 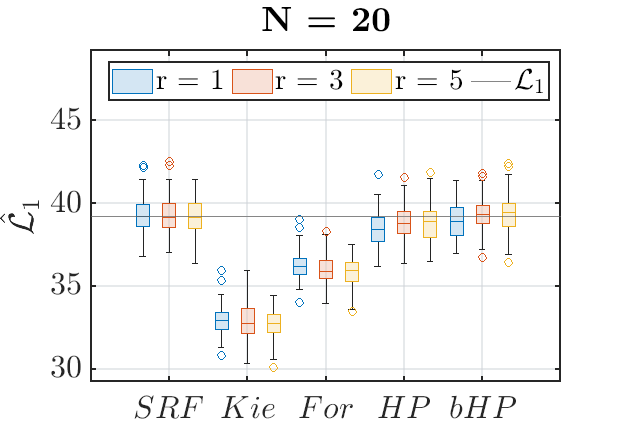}
\includegraphics[trim=0 0 40 0,clip,width=1.8in]{\figurepathh 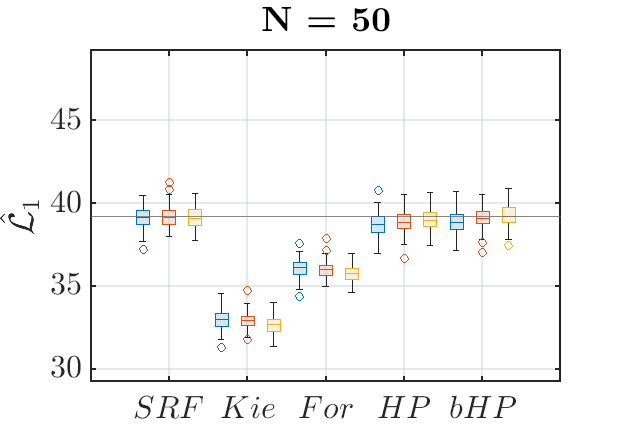}
\includegraphics[trim=0 0 40 0,clip,width=1.8in]{\figurepathh 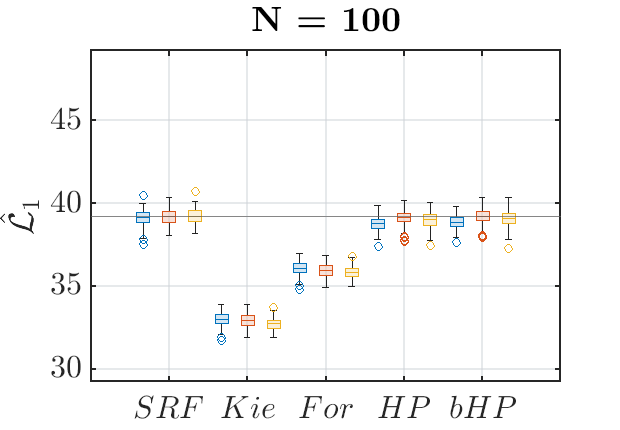}\\
\includegraphics[trim=0 0 40 0,clip,width=1.8in]{\figurepathh 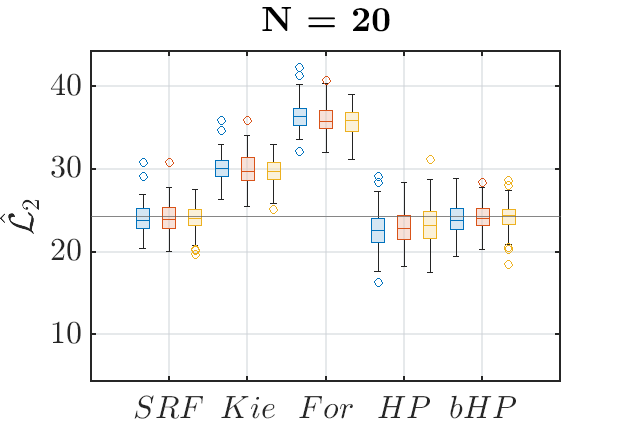}
\includegraphics[trim=0 0 40 0,clip,width=1.8in]{\figurepathh 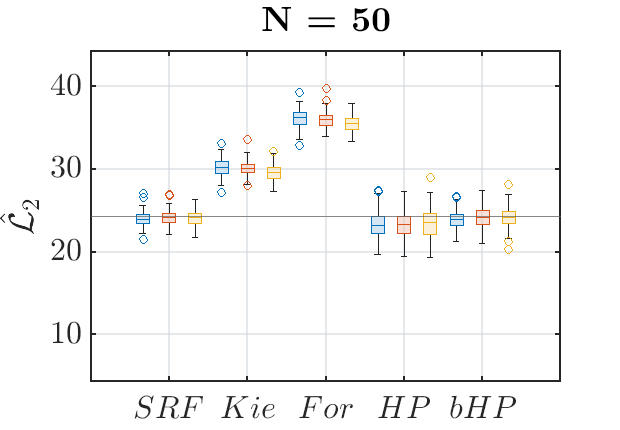}
\includegraphics[trim=0 0 40 0,clip,width=1.8in]{\figurepathh 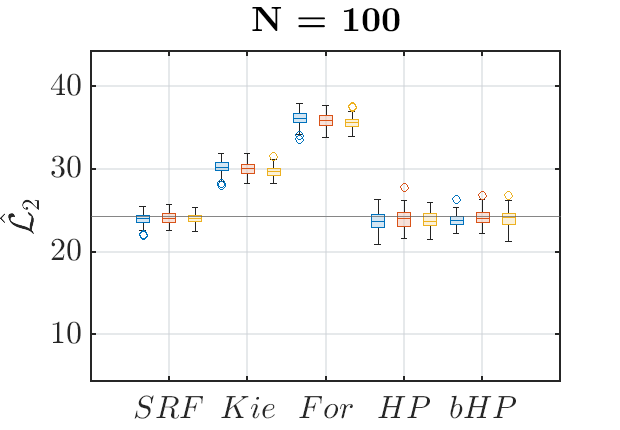}\\
	\caption{ 2D Simulation results for estimation of the LKCs of the
			  non-stationary SuRF described in Section \ref{scn:setup}.
			  The smoothing bandwidth is $f = 3$.
			  \label{fig:D2FixedFWHM_nonstat}}
\end{figure}
The HP estimator has a similar bias and a
much larger variance, compare also \cite{Telschow2020}. The Kiebel and Forman estimators are biased for our non-stationary SuRF example and the almost stationary SuRF for the small value $f=3$ which is
typical in neuroimaging.
The latter is known in the literature, see e.g., \cite{Kiebel1999}.
The dependence of the LKC estimates on the smoothing bandwidth $f$ for $D=2$ can be found in
Figures \ref{fig:D2L1FixedNSUBJ_stat}-\ref{fig:D2L2FixedNSUBJ_nonstat}
in Appendix \ref{appendix:AddSimulations}.
Here only the SuRF estimator correctly estimates the LKCs for all $f\in\{1,\ldots, 6\}$.
Moreover, only in the case $f=1$ a resolution increase of $r\geq 3$ is
necessary to have unbiased estimates. Remarkably, even for FWHM $f = 2$ a resolution increase
of $r=1$ seems to be sufficient for unbiased estimation using the SuRF estimator.
Additionally, the SuRF estimator is
$10$ times faster to compute than its only reliable
competitor the bHP estimator, see Table \ref{Tab:SimTimes}. Here we compared the average computation time
of the SuRF, the bHPE and the Kiebel estimator on the stationary box example with $f = 3$ and $N=100$
at different added resolution.

\begin{table}[b]
\footnotesize
\begin{center}
\begin{tabular}{ | c |c c c | c c c | c c c | }\hline
 & \multicolumn{3}{c}{SuRF} & \multicolumn{3}{|c|}{bHPE} & \multicolumn{3}{c|}{Kiebel}\\
 $r$ & 1 & 3 & 5 & 1 & 3 & 5 & 1 & 3 & 5\\\hline
 $D = 1$ & 0.06  & 0.05  & 0.05   & 7.01  & 13.35 & 16.13 & 4  & 0.03  & 0.04  \\ 
 $D = 2$ & 0.28  & 0.66 & 1.66 & 19.96 & 30.67 & 46.67 & 0.11 & 0.42 & 0.97 \\  
 $D = 3$ & 18.22  &  80.45   &  239.30    & 196.44    & -    & -    & 8.24  & 47.90  & 151.75\\\hline
\end{tabular}
\end{center}
	\caption{Computation time of LKC estimators in the stationary box example ($N=100$, $f= 3$).
	We show averages in seconds of $100$ runs of the estimators.
         The times for resolution increases beyond $r=1$ and $D=3$ of the bHPE is not reported as they are
         very long.}
\label{Tab:SimTimes}
\end{table}

\subsection{Results of the FWER Simulation}\label{scn:simFWER}
We now verify that our SuRF framework yields a non-conservative control of the FWER.
To do so we determine the FWER of the one-sample $ t $-test for a sample of Gaussian
SuRFs generated as in Section \ref{scn:setup} using the methodology
described in Section \ref{scn:FWER}.
In each simulation setting, for $ 1 \leq b \leq B $ and $ N \in \lbrace 20, 50, 100 \rbrace$ we obtain $ t $-fields $ T_{N,b,r} $ of resolution $ r \in \lbrace 0, 1, \infty \rbrace $. Here $ r = 0 $ corresponds to the traditional RFT approach, i.e., $ T_{N,b,0} $ is the test statistic described in Section \ref{scn:FWER} evaluated on the lattice $ \mathcal{V} $ on which the original data is observed. Similarly, the case $ r = 1 $
corresponds to  $ T_{N,b,1} $  being the test statistic evaluated on $ \cM_\cV^{(1)} $, compare \eqref{eq:grid}.
The case $ r = \infty $ means the test statistic $ T_{N,b,\infty} $ is a SuRF defined on $\cM_\cV$.
We evaluate the FWER in this case by using numerical optimization (in particular sequential quadratic programming, \cite{Nocedal2006}),
initialized at the largest peaks of $ T_{N,b,1} $ to find the global maximum of $ T_{N,b,\infty} $ over $\cM_\cV$.
The LKCs are estimated as described in Section \ref{scn:SuRFLKCestim}
using added resolution $r=1$ which is justified by Section \ref{scn:simLKC}.
The threshold $ \hat{u}_{N,b} $ which controls the FWER at a level $\alpha = 0.05$ is obtained
as described in Section \ref{scn:FWER}. For $ r \in \lbrace 0, 1, \infty \rbrace$ and $ \cM_\cV^{(0)}=\mathcal{V} $
and $ \cM_\cV^{(\infty)} = \cM_\cV $, the FWER is then estimated by evaluating 
\begin{equation*}
	\frac{1}{B} \sum_{b = 1}^B 1\left[\sup_{s \in\cM_\cV^{(r)}} T_{N,b,r}(s) > \hat{u}_{N,b} \right]\,,
\end{equation*}
\begin{figure}[t!]
\centering
\begin{turn}{90}
{\fontfamily{lmss}\selectfont
\small{\textbf{\hspace{0.47in}Stationary}}
}
\end{turn}
	\includegraphics[trim=0 0 0 0,clip,width=1.8in]{\figurepathfwer 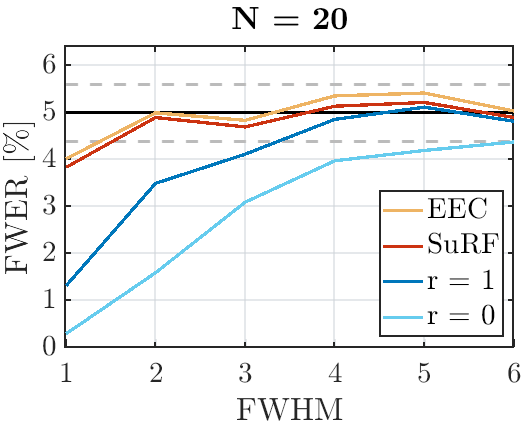}
	\includegraphics[trim=0 0 0 0,clip,width=1.8in]{\figurepathfwer 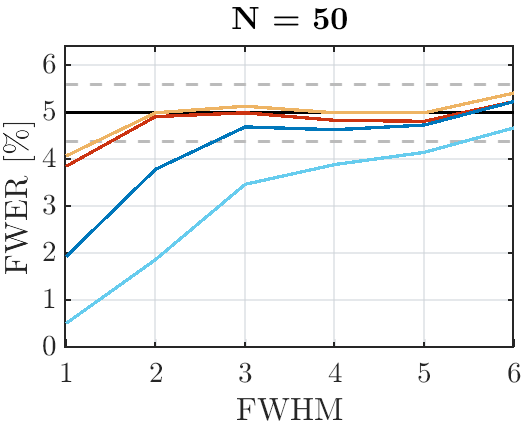}
	\includegraphics[trim=0 0 0 0,clip,width=1.8in]{\figurepathfwer 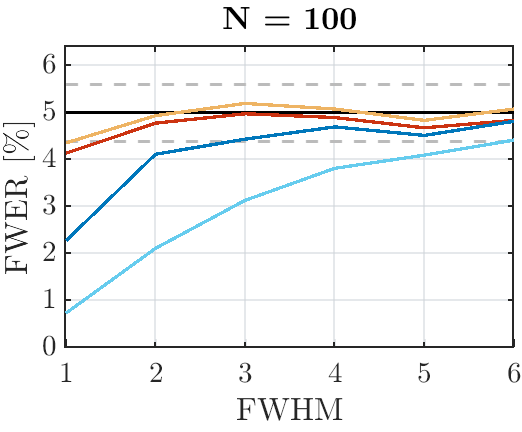}\\
\begin{turn}{90}
{\fontfamily{lmss}\selectfont
\small{\textbf{\hspace{0.3in}Non-Stationary}}
}
\end{turn}
	\includegraphics[trim=0 0 0 0,clip,width=1.8in]{\figurepathfwer 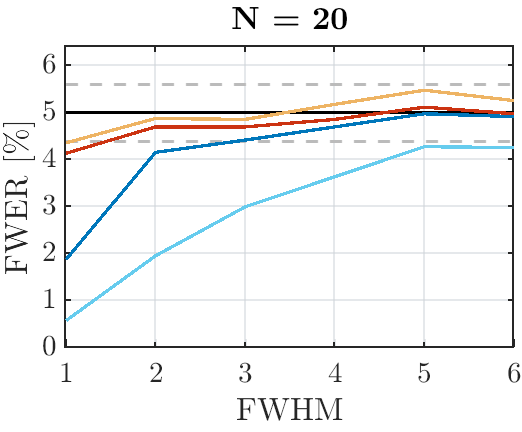}
	\includegraphics[trim=0 0 0 0,clip,width=1.8in]{\figurepathfwer 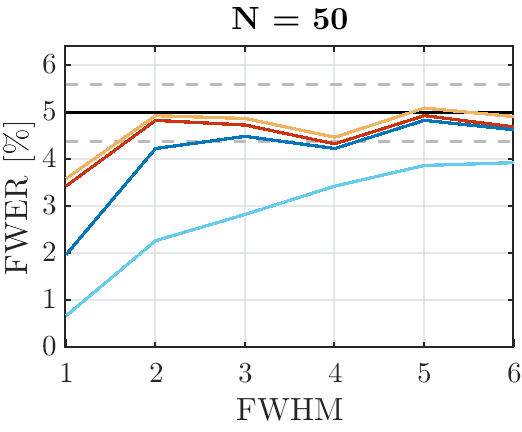}
	\includegraphics[trim=0 0 0 0,clip,width=1.8in]{\figurepathfwer 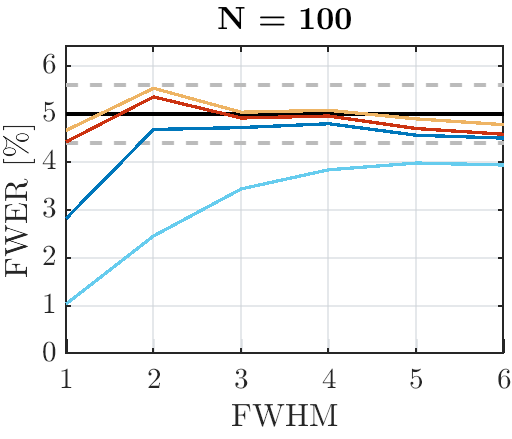}		\caption{FWER results for the almost-stationary (top row) and non-stationary (bottom row) settings for
	$D=2$.
	}\label{fig:FWER2D}
\end{figure}
\begin{figure}[h!]
	\centering
\begin{turn}{90}
{\fontfamily{lmss}\selectfont
\small{\textbf{\hspace{0.47in}Stationary}}
}
\end{turn}
	\includegraphics[trim=0 0 0 0,clip,width=1.8in]{\figurepathfwer 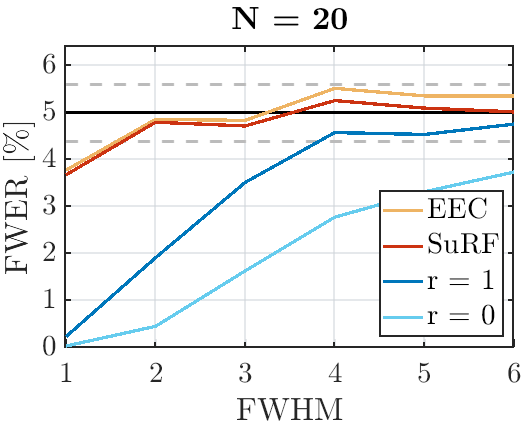}
	\includegraphics[trim=0 0 0 0,clip,width=1.8in]{\figurepathfwer 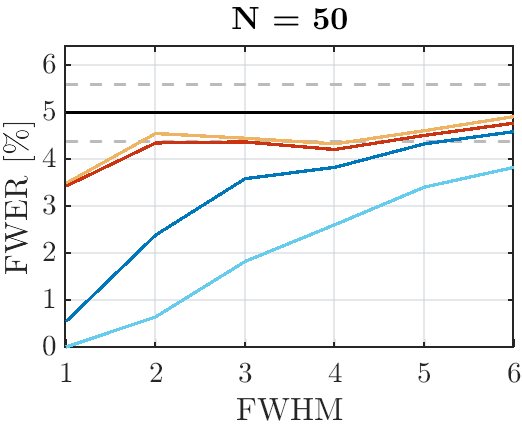}
	\includegraphics[trim=0 0 0 0,clip,width=1.8in]{\figurepathfwer 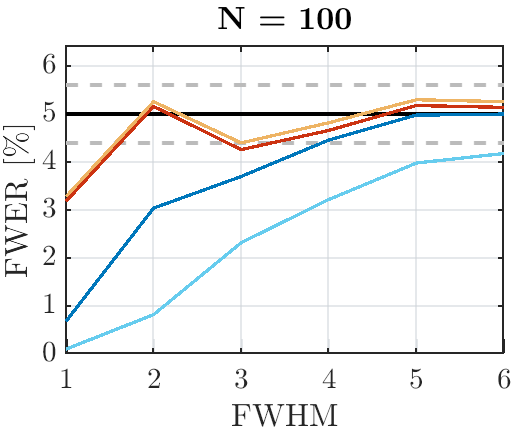}\\
\begin{turn}{90}
{\fontfamily{lmss}\selectfont
\small{\textbf{\hspace{0.3in}Non-Stationary}}
}
\end{turn}
	\includegraphics[trim=0 0 0 0,clip,width=1.8in]{\figurepathfwer 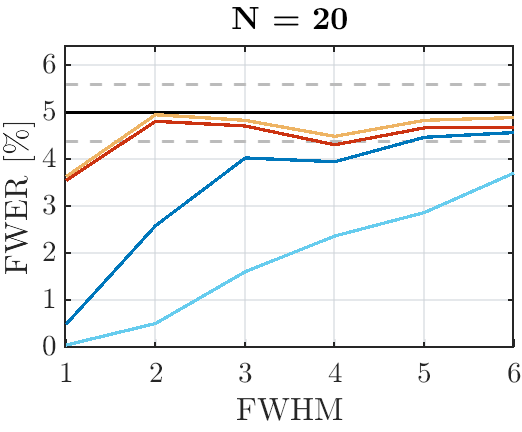}
	\includegraphics[trim=0 0 0 0,clip,width=1.8in]{\figurepathfwer 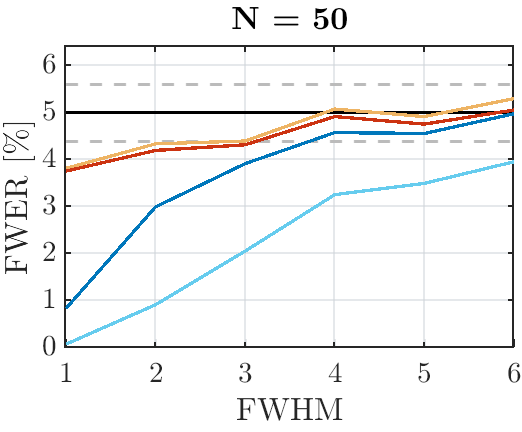}
	\includegraphics[trim=0 0 0 0,clip,width=1.8in]{\figurepathfwer 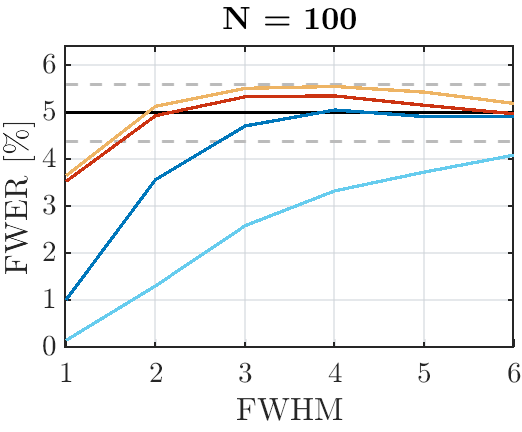}\\
	\caption{FWER results for the almost-stationary (top row) and non-stationary (bottom row) settings
	for $D=3$.}\label{fig:FWER3D}
\end{figure}
In each simulation we also count the number of local maxima $ l_{N,b,r} $ of
$ T_{N,b,r} $ on $ \cM_\cV^{(r)} $ exceeding the level $ \hat{u}_{N,b} $ as it is a good approximation of the EC of the excursion set
for $\alpha\approx 0.05$. The resulting estimate
$
	\widehat{\text{EEC}} = 	\frac{1}{B} \sum_{b = 1}^B l_{N,b,r}
	$
is plotted in yellow in our figures. The fact, that $\widehat{\text{EEC}}$ is close to
the FWER obtained for $r=\infty$, is a good indicator that the EEC heuristic
used in \ref{scn:FWER} to estimate $ \hat{u}_{N,b} $ is valid. 

The results of our simulations for $D=2,3$ are shown in Figures \ref{fig:FWER2D} and \ref{fig:FWER3D} respectively.  They reproduce the well-known observation that the traditional ($ r = 0 $) approach is conservative.
This is more severe at low smoothness levels, but even at FWHM $f = 6$, evaluating the random fields on the original lattice yields a conservative FWER. This effect is slightly more pronounced in 3D.
Moving from the original ($ r = 0 $) lattice to the resolution $ 1 $ lattice already reduces the conservativeness which demonstrates that the main cause of the conservativeness is the mismatch between the discreteness of the data and modeling it as a random field over a manifold. For $r=\infty$ the FWER is close to the nominal level $\alpha =0.05$ and only slightly smaller than the estimated EEC. The less accurate FWER for FWHM $f=1$ could be partially caused
by the biased LKC estimates as we used added resolution $1$, compare Figures \ref{fig:D2L1FixedNSUBJ_stat}-\ref{fig:D2L2FixedNSUBJ_nonstat} in Appendix \ref{appendix:AddSimulations}.

\FloatBarrier
\section{Discussion}\label{scn:Discussion}
In this article we have shown how to resolve the conservativeness of voxelwise inference using RFT
(\cite{Nichols2003,Taylor2007b,Eklund2016}).
Our solution is inspired by the rigorous distinction between data of an experiment and
atoms of a probabilistic model, which helps to identify the main problems of traditional RFT
for Gaussian data:
(i) the test statistic is evaluated only on a discrete set
given by the data although the atoms of RFT are random fields
over a WS  manifold and (ii) continuous
quantities are replaced by discrete approximations.
Here we proposed using kernel smoothers as ferrymen transferring
discrete data into random fields over a WS manifold
(the atoms of RFT) and showed in simulations that this view enables accurate control of the FWER
at a given significance level.
Although our SuRF methodology is tailored to Gaussian related fields,
this assumption can be relaxed. As discussed in \cite{Nardi2008,Telschow:2022SCB} the GKF approximates
the EEC of the excursion sets of the test statistic $\tilde T$ well,
if $\tilde T$ only is asymptotically Gaussian.
Because the SuRF LKC estimates are consistent so long as derivatives of the sample covariance
of the residuals converge uniformly to the derivatives of the covariance function of the limiting
field, compare \cite{Telschow2020}, our SuRF methodology remains asymptotically valid
even under non-Gaussianity.
However, as the sample size required to deal with the non-Gaussianity of fMRI data
is large, we extend in Part 2, i.e., \cite{Davenport2023}, our SuRF based voxelwise inference using RFT to
non-Gaussian data and demonstrate its non-conservativeness using a large resting state validation.

The importance of our work for neuroimaging is twofold.
Firstly, our GKF based method is computationally faster
than resampling based inference such as permutation tests
which control the FWER at level $\alpha$ over the grid  $\cV$
at the cost of a high computational burden.
Secondly, solving the conservativeness that has long caused
power problems for voxelwise inference using RFT is a first step
towards identifying and solving the problems of false positive rates in cluster-size
inference (\cite{Eklund2016}) as it also relies on the GKF
and applies continuous theory to smoothed fields evaluated on the voxel lattice \citep{Friston1994}.

As the estimator for $ \mathcal{L}_1 $ in 3D using Theorem \ref{thm:L1VM}
is difficult to implement, we proposed a local stationary approximation which consists of the first integral in \eqref{eq:L1Surf}.
Even if the fields are highly non-stationary using this approximation does not have a large effect on voxelwise inference as the estimate $\hat u_\alpha$ from the EEC is primarily driven by $\hat{\mathcal{L}}_2$ and $\hat{\mathcal{L}}_3$.
This can be seen for example for real fMRI data using the supplementary
material of \cite{Telschow2020}. Here $u_{0.95} \approx 4.2$ which is typical for neuroimaging data.
Thus,
$
	\rho^{\tilde T}_3(u_\alpha)/\rho^{\tilde T}_1(u_\alpha) \approx 2.6 
$
and
$\rho^{\tilde T}_2(u_\alpha)/\rho^{\tilde T}_1(u_\alpha) \approx 1.6$
for sample size $50$. Moreover, typically $\mathcal{L}_1 \ll \mathcal{L}_3$. To obtain intuition on why this holds, consider the case where the atoms arise as stationary random fields with the square exponential covariance
function with bandwidth parameter $h$ over a convex WS manifold $\cM$. In this case $\cL_3$ is the volume of $\cM$
divided by $h^3$, $\cL_2$ is half the surface area of $\cM$ divided by $h^2$ and $\cL_1$ is twice the diameter of  $\cM$ divided by $h$ \cite[Table 2]{Worsley2004}.\footnote{Note that in Table 2 the FWHM needs to be transformed
into a bandwidth by multiplication with $\sqrt{4\log(2)}/2\pi$
to transform resels --a concept from \cite{Worsley1996} designed
for isotropic, stationary processes-- into LKCs, compare our discussion in the Supplementary of Part 2.}
That the local stationary approximation of $\cL_1$ does not influence the false positive rate
can additionally be seen from Figure \ref{fig:FWER3D} and our simulations in Part 2.

The careful distinction between data and atoms and the practical problems
$(i)$ and $(ii)$ are relevant in a number of further applications. For example, any functional data
method involving smoothing via basis functions or kernel smoothers, and for which
the coverage of simultaneous confidence bands or the family-wise error rate of tests is reported,
can benefit from our SuRF reasoning.
A concrete example is \cite{Davenport2022confidence}.
Here confidence regions for peaks of the signal of random fields over bounded open domains
of $\mR^D$ are developed and convolution fields are used to localize peaks of activation in fMRI and MEG.
In fact, MEG is a natural domain in which our SuRF reasoning applies,
because the power spectrum is a convolution field,
see the Supplementary material of \cite{Davenport2022confidence} for details.
Another potential application of the SuRF reasoning are coverage probability
excursion (CoPE) sets.
CoPE sets provide confidence sets for the excursion above a value
$c\in\mathbb{R}$ of a real-valued target function defined over
a domain in $\mathbb{R}^D$, $D>0$, from noisy data, \citep{Sommerfeld2018,Bowring:2019,Bowring:2021,Maullin:2023spatial}.
These papers found in simulations that the empirical coverage of CoPE sets is larger than the specified coverage, but converges to the correct coverage probability
as the domain is sufficiently densely sampled. Theoretically, this has been explained in Theorem 1.b) from \citep{Sommerfeld2018} and is illustrated in Figure 3 of \cite{Bowring:2019}. SuRFs provide a practical means of
overcoming this overcoverage.

%% file: acknowledgements.tex
F.T. is funded by the Deutsche Forschungsgemeinschaft (DFG) under Excellence Strategy The
Berlin Mathematics Research Center MATH+ (EXC-2046/1, project ID:390685689).
F.T. and S.D. were partially supported by NIH grant R01EB026859. S.D. was partially supported by NIH grant R01MH128923. We thank Armin Schwartzman for generous funding and
discussions on this topic -- especially for suggesting the simplified notation \eqref{eq:InnerProdRepr},
Henrik Schumacher for helpful discussions on Riemannian geometry and spotting that the term involving
the Riemannian curvature tensor in \eqref{eq:L1general} is not zero
and Thomas E. Nichols for suggesting the term \textit{Super Resolution field} (SuRF)
which fits well the institution F.T. was employed at when he started this topic.

%% file: appendix.tex
\appendix
\section{Additional Figures, Simulation Results and Tables}\label{appendix:AddSimulations}
\begin{figure}[ht]
\begin{center}
	\includegraphics[trim=0 0 0 0,clip,width=2.5in]{\figurepath 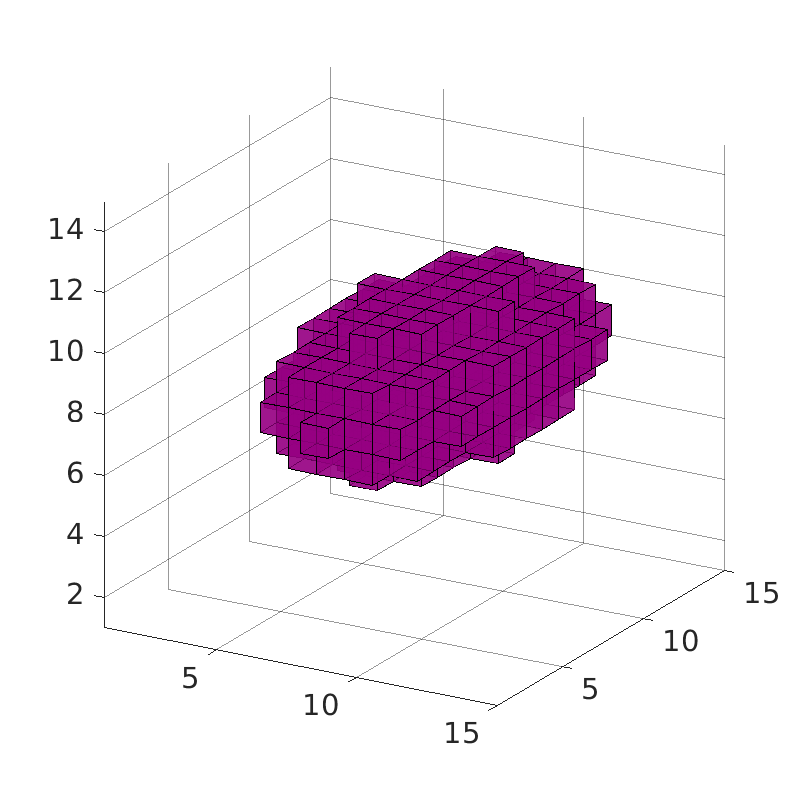}
	\includegraphics[trim=0 0 0 0,clip,width=2.5in]{\figurepath 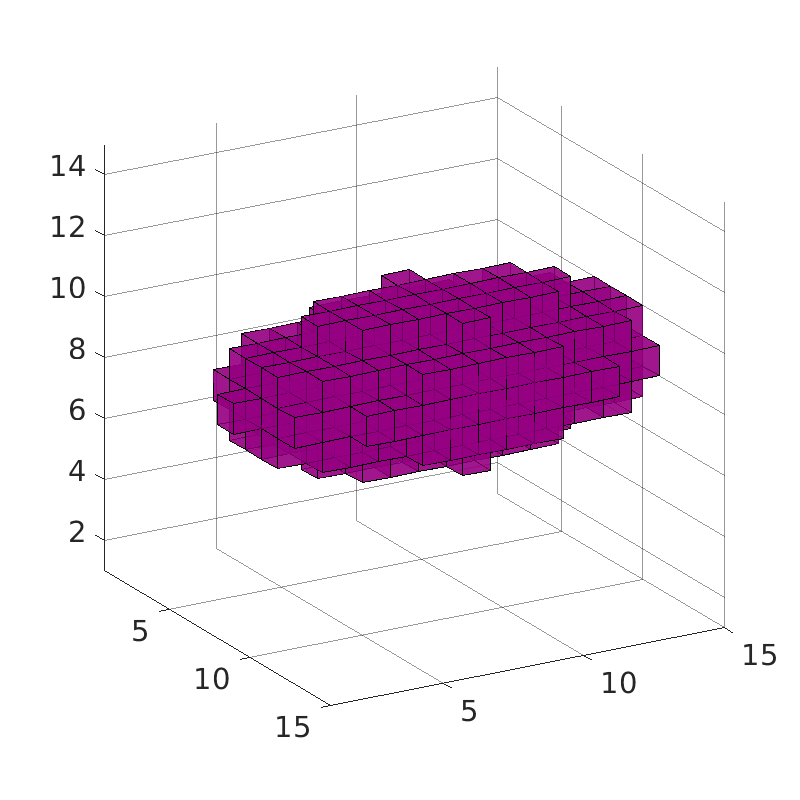}
	\caption{An example of a voxel manifold viewed from different angles.}
	\label{fig:voxmf}
\end{center}
\end{figure}

\FloatBarrier
\subsection{LKC estimation for $D=1$}
\FloatBarrier
\begin{figure}[h!]
\centering
\begin{turn}{90}
{\fontfamily{lmss}\selectfont
\small{\textbf{\hspace{0.4in}Stationary}}
}
\end{turn}
\includegraphics[trim=0 0 40 0,clip,width=1.9in]{\figurepathh 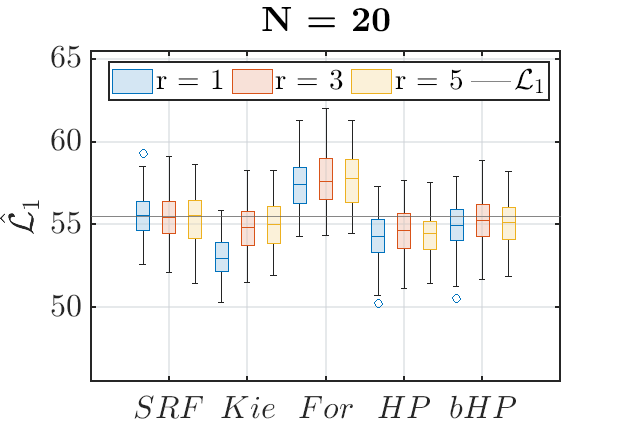}
\includegraphics[trim=0 0 40 0,clip,width=1.9in]{\figurepathh 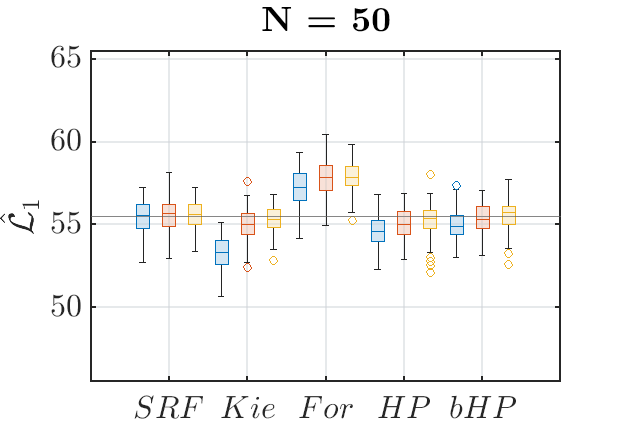}
\includegraphics[trim=0 0 40 0,clip,width=1.9in]{\figurepathh 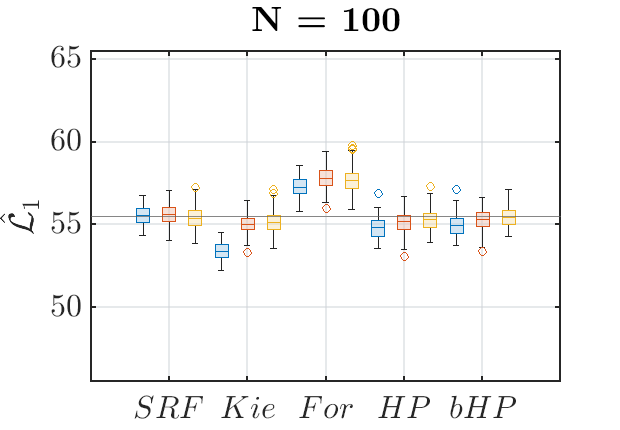}\\
\begin{turn}{90}
{\fontfamily{lmss}\selectfont
\small{\textbf{\hspace{0.2in}Non-Stationary}}
}
\end{turn}
\includegraphics[trim=0 0 40 0,clip,width=1.9in]{\figurepathh 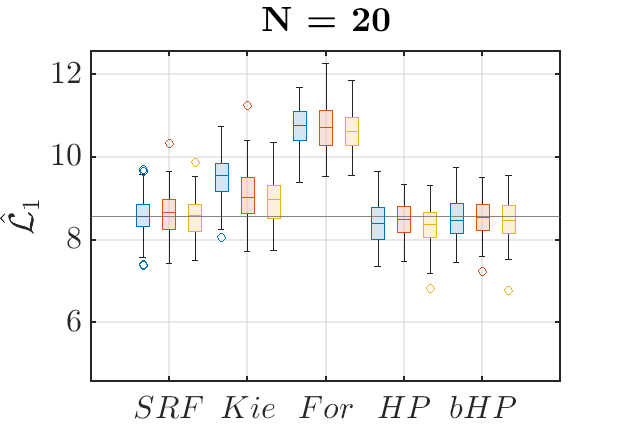}
\includegraphics[trim=0 0 40 0,clip,width=1.9in]{\figurepathh 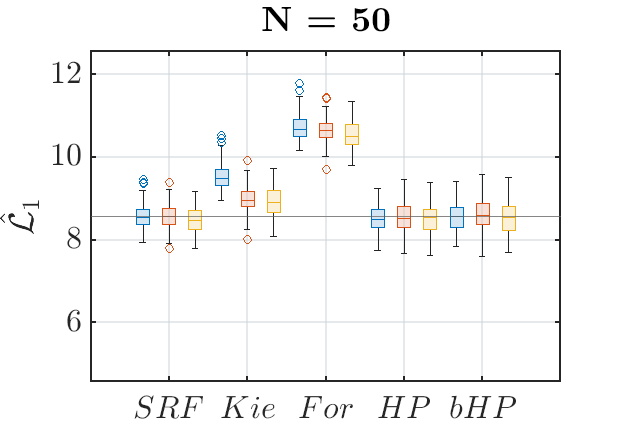}
\includegraphics[trim=0 0 40 0,clip,width=1.9in]{\figurepathh 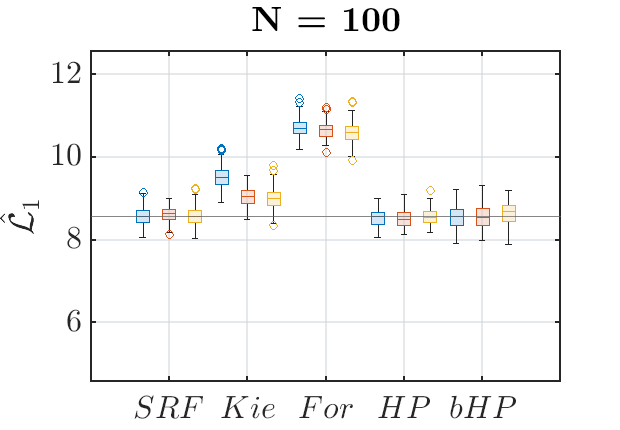}\\
	\caption{1D Simulation results of estimation of the LKCs of the two SuRFs
			 derived from the stationary box example
			 and the non-stationary sphere example. The FWHM is  $f=3$.
			  \label{fig:D1L1FixedFWHM}}
\end{figure}

\begin{figure}[h!]
\centering
\includegraphics[trim=0 0 40 0,clip,width=1.9in]{\figurepathh 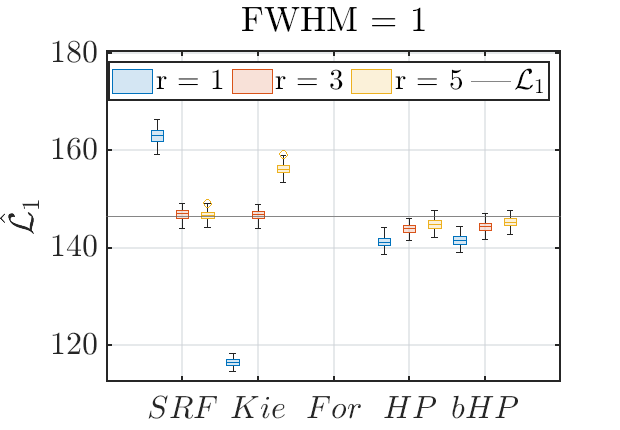}
\includegraphics[trim=0 0 40 0,clip,width=1.9in]{\figurepathh 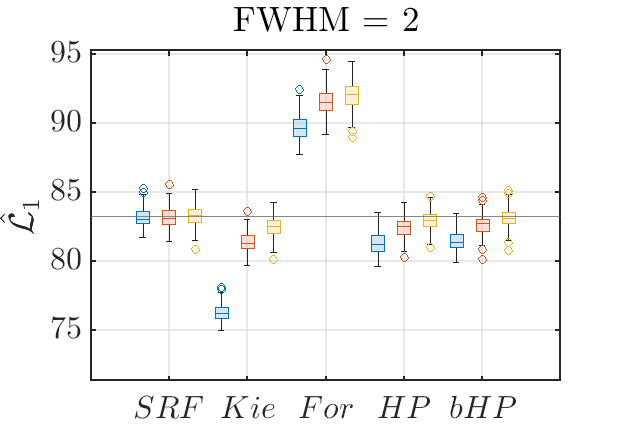}
\includegraphics[trim=0 0 40 0,clip,width=1.9in]{\figurepathh 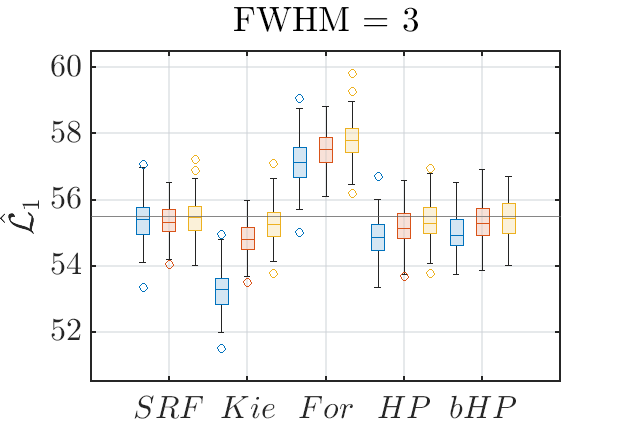}\\
\includegraphics[trim=0 0 40 0,clip,width=1.9in]{\figurepathh 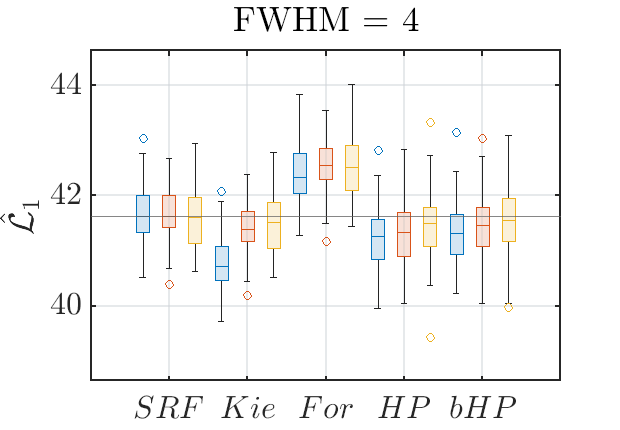}
\includegraphics[trim=0 0 40 0,clip,width=1.9in]{\figurepathh 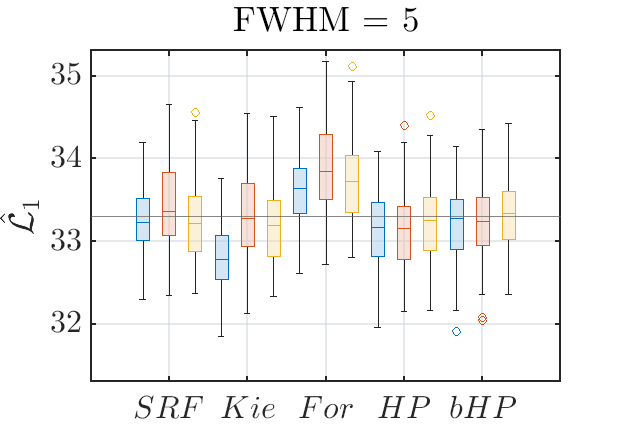}
\includegraphics[trim=0 0 40 0,clip,width=1.9in]{\figurepathh 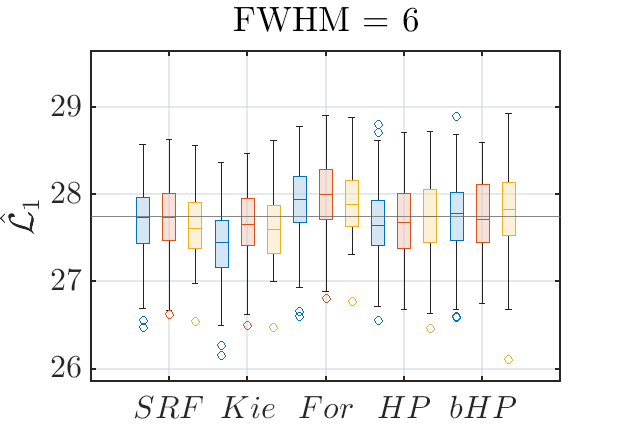}\\
	\caption{ 1D Simulation results of estimation of the LKCs of SuRFs derived from the stationary box example.
			  The results show the dependence of the LKC estimation on the FWHM
			  used in the smoothing kernel for sample size $N=100$.
	          \label{fig:D1L1FixedNSUBJ_stat}}
\end{figure}

\begin{figure}[h!]
\centering
\includegraphics[trim=0 0 40 0,clip,width=1.9in]{\figurepathh 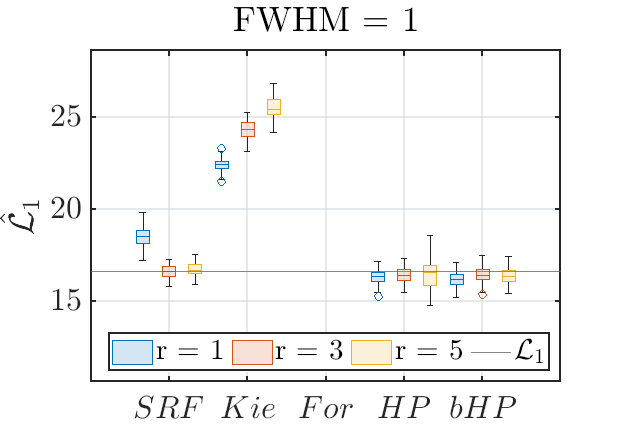}
\includegraphics[trim=0 0 40 0,clip,width=1.9in]{\figurepathh 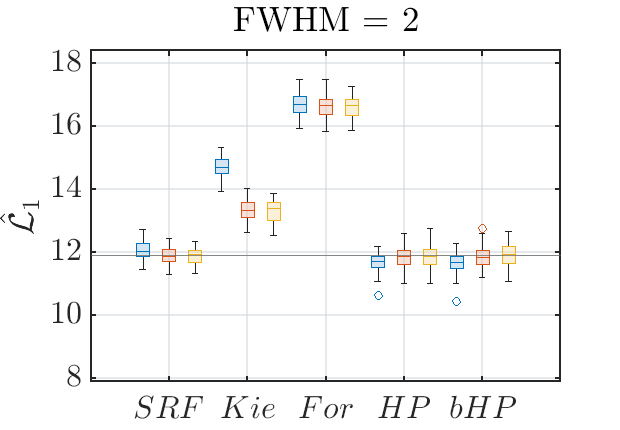}
\includegraphics[trim=0 0 40 0,clip,width=1.9in]{\figurepathh 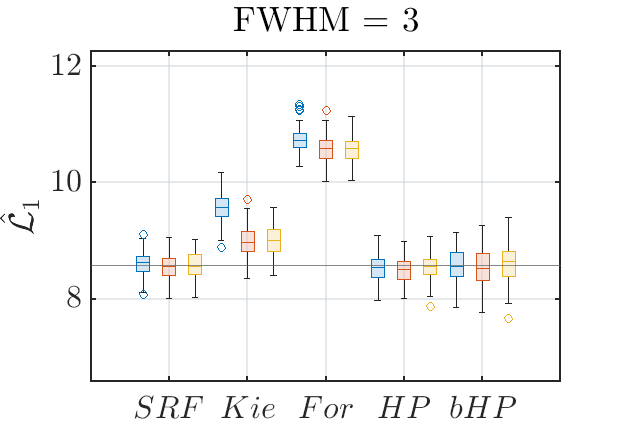}\\
\includegraphics[trim=0 0 40 0,clip,width=1.9in]{\figurepathh 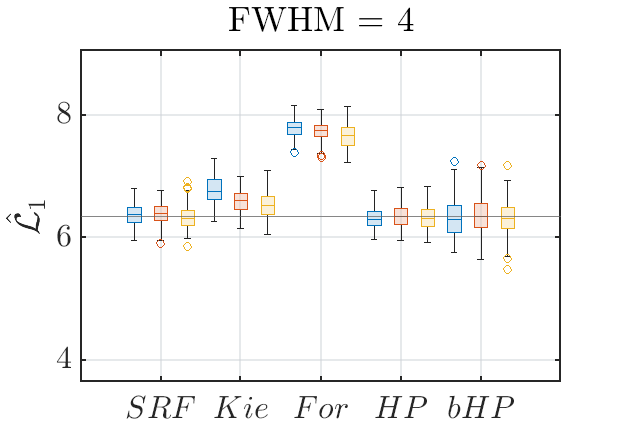}
\includegraphics[trim=0 0 40 0,clip,width=1.9in]{\figurepathh 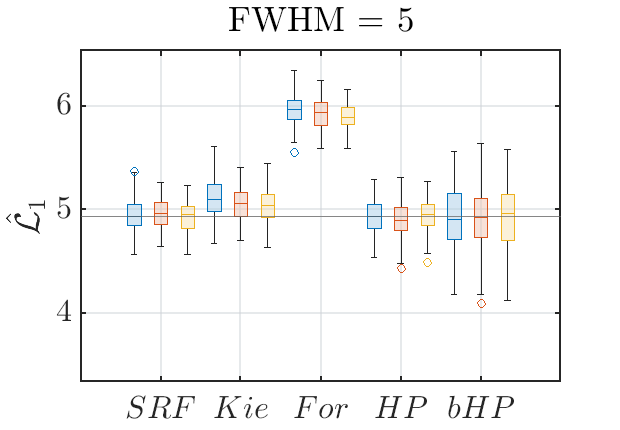}
\includegraphics[trim=0 0 40 0,clip,width=1.9in]{\figurepathh 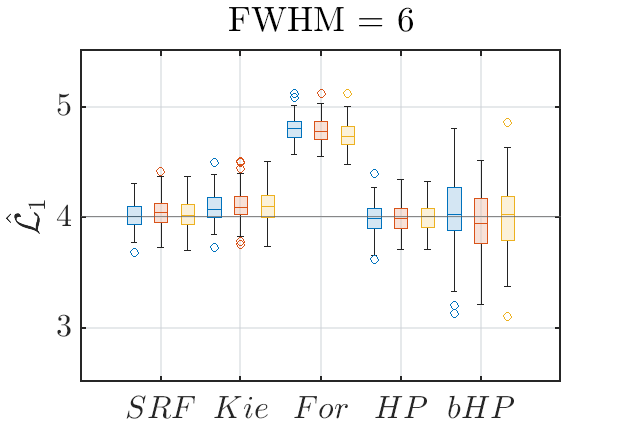}\\
	\caption{ 1D Simulation results of estimation of the LKCs of SuRFs derived from the non-stationary sphere example.
			  The results show the dependence of the LKC estimation on the FWHM
			  used in the smoothing kernel for sample size $N=100$.
	          \label{fig:D1L1FixedNSUBJ_nonstat}}
\end{figure}
\FloatBarrier
\clearpage

\subsection{LKC estimation for $D=2$}
\FloatBarrier

\begin{figure}[h!]
\centering
\includegraphics[trim=0 0 40 0,clip,width=1.9in]{\figurepathh 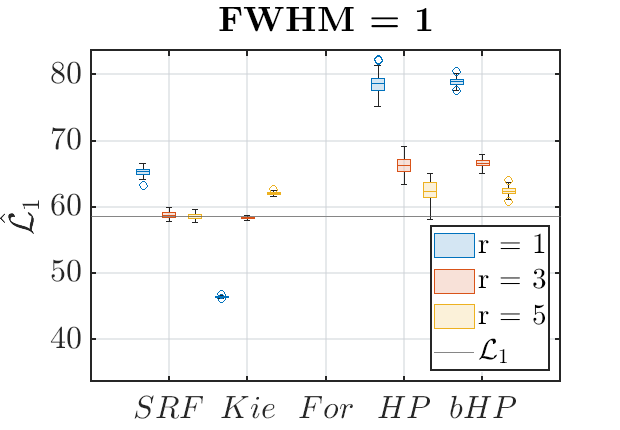}
\includegraphics[trim=0 0 40 0,clip,width=1.9in]{\figurepathh 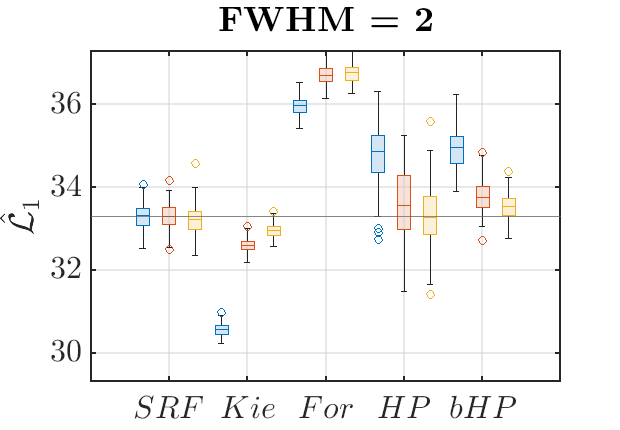}
\includegraphics[trim=0 0 40 0,clip,width=1.9in]{\figurepathh 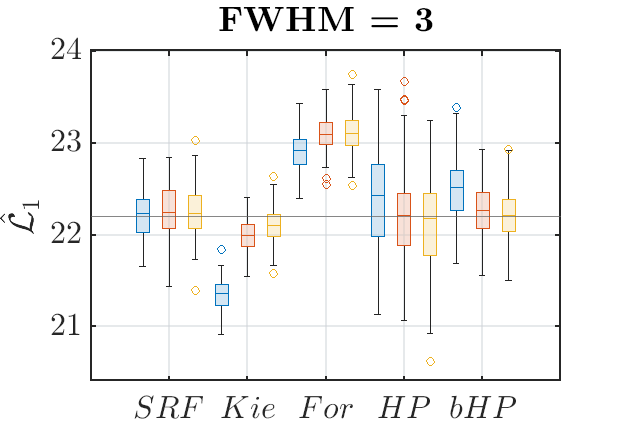}\\
\includegraphics[trim=0 0 40 0,clip,width=1.9in]{\figurepathh 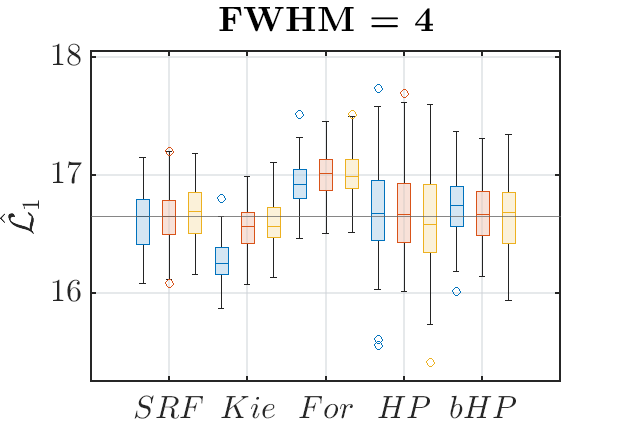}
\includegraphics[trim=0 0 40 0,clip,width=1.9in]{\figurepathh 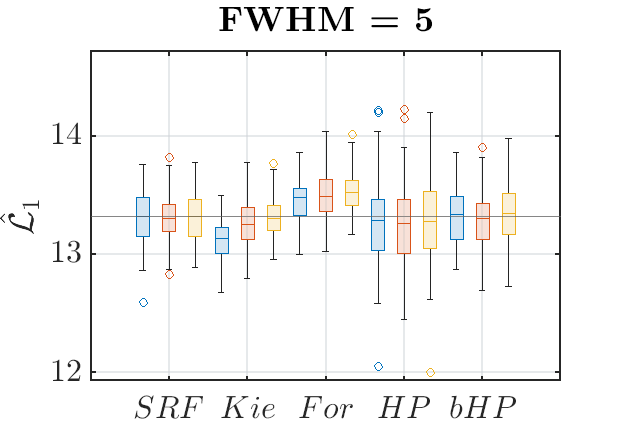}
\includegraphics[trim=0 0 40 0,clip,width=1.9in]{\figurepathh 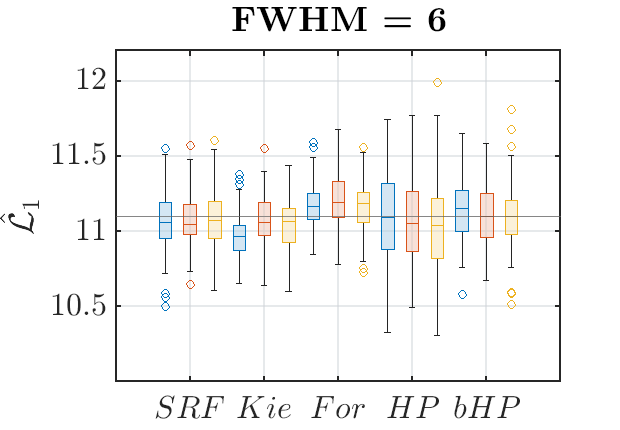}\\
	\caption{ 2D Simulation results for estimation of the LKCs of the
			  almost stationary SuRF described in Section \ref{scn:setup}.
			  The figures show the dependence of the LKC estimation on
			  the smoothing bandwidth $f$ for sample size $N=100$.
	          \label{fig:D2L1FixedNSUBJ_stat}}
\end{figure}

\begin{figure}[h!]
\centering
\includegraphics[trim=0 0 40 0,clip,width=1.9in]{\figurepathh 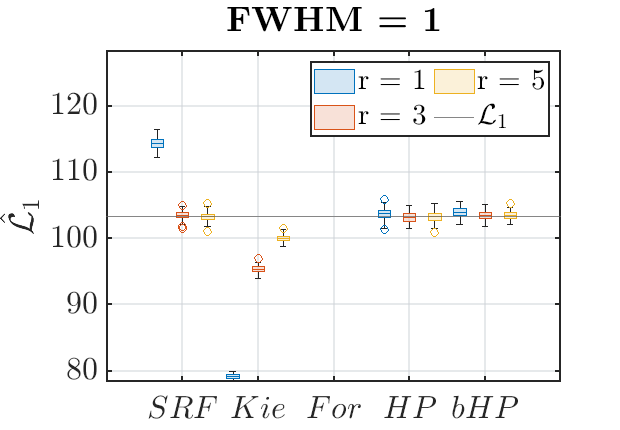}
\includegraphics[trim=0 0 40 0,clip,width=1.9in]{\figurepathh 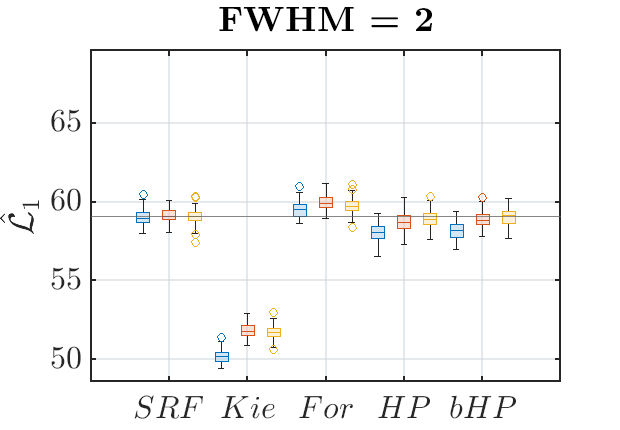}
\includegraphics[trim=0 0 40 0,clip,width=1.9in]{\figurepathh 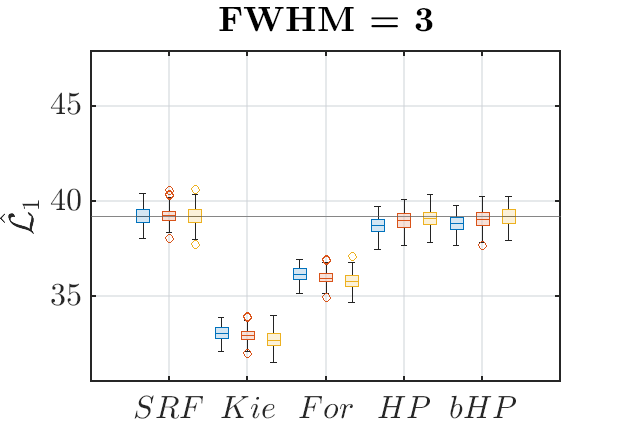}\\
\includegraphics[trim=0 0 40 0,clip,width=1.9in]{\figurepathh 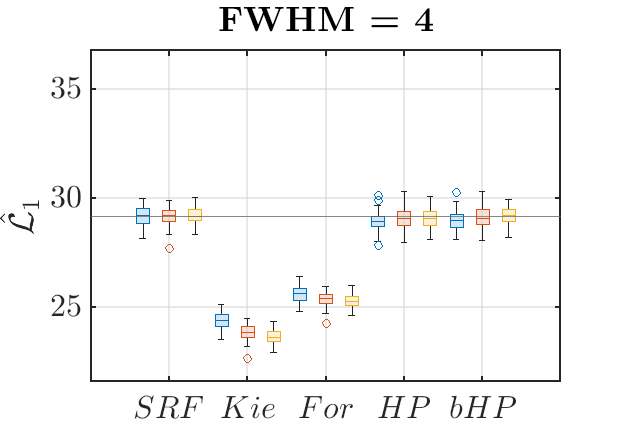}
\includegraphics[trim=0 0 40 0,clip,width=1.9in]{\figurepathh 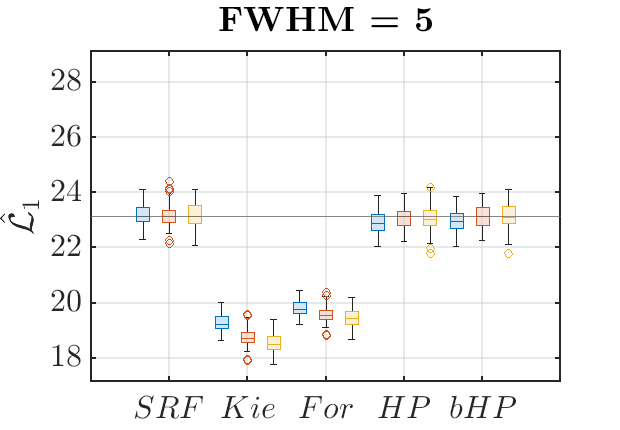}
\includegraphics[trim=0 0 40 0,clip,width=1.9in]{\figurepathh 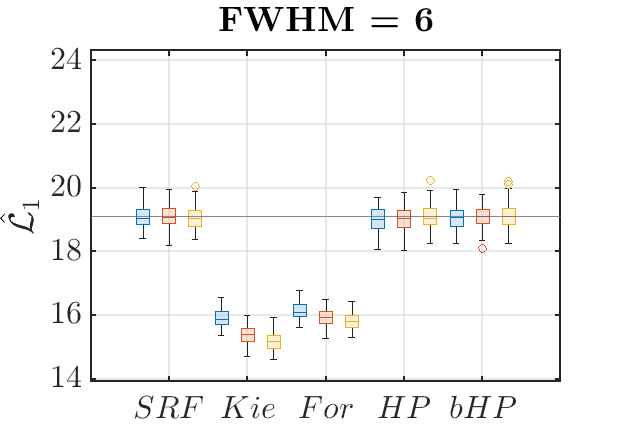}\\
	\caption{ 2D Simulation results of estimation of the LKCs of SuRFs derived from the non-stationary sphere example.
			  The results show the dependence of the LKC estimation on the FWHM
			  used in the smoothing kernel for sample size $N=100$.
	          \label{fig:D2L1FixedNSUBJ_nonstat}}
\end{figure}

\begin{figure}[ht]
\centering
\includegraphics[trim=0 0 40 0,clip,width=1.9in]{\figurepathh 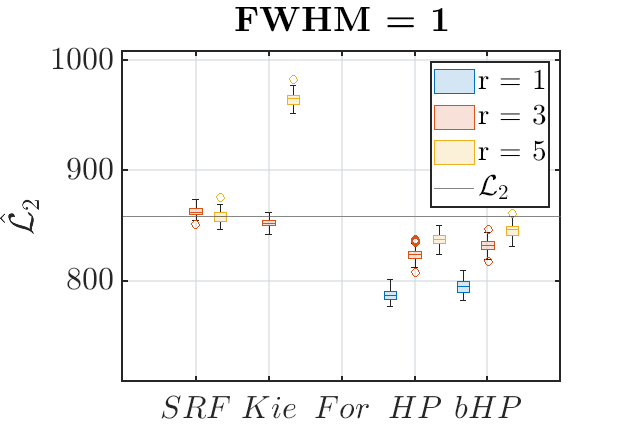}
\includegraphics[trim=0 0 40 0,clip,width=1.9in]{\figurepathh 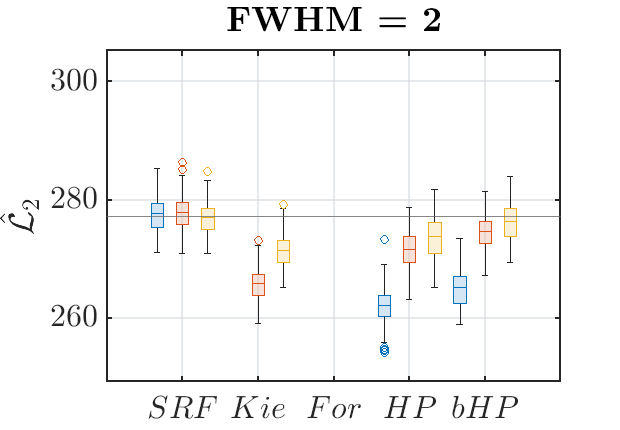}
\includegraphics[trim=0 0 40 0,clip,width=1.9in]{\figurepathh 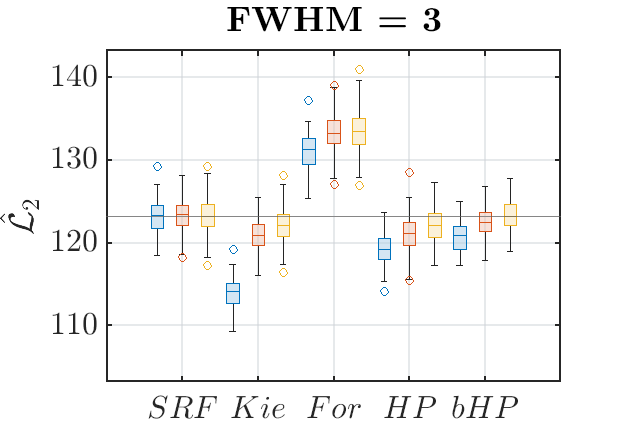}\\
\includegraphics[trim=0 0 40 0,clip,width=1.9in]{\figurepathh 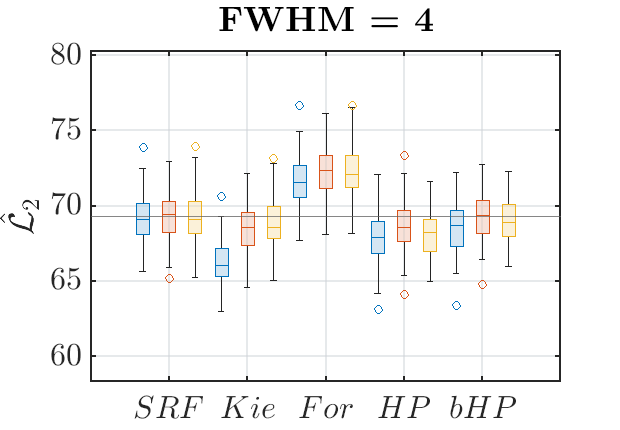}
\includegraphics[trim=0 0 40 0,clip,width=1.9in]{\figurepathh 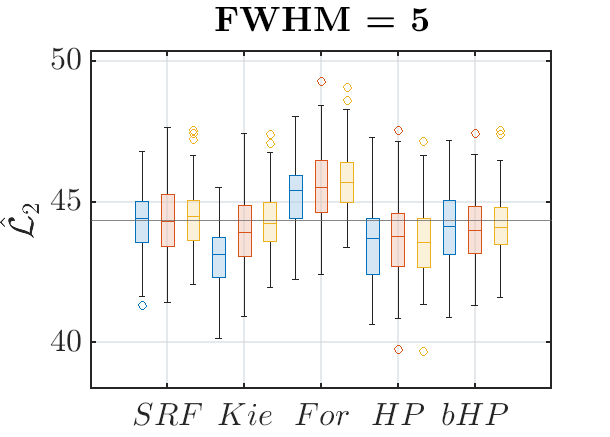}
\includegraphics[trim=0 0 40 0,clip,width=1.9in]{\figurepathh 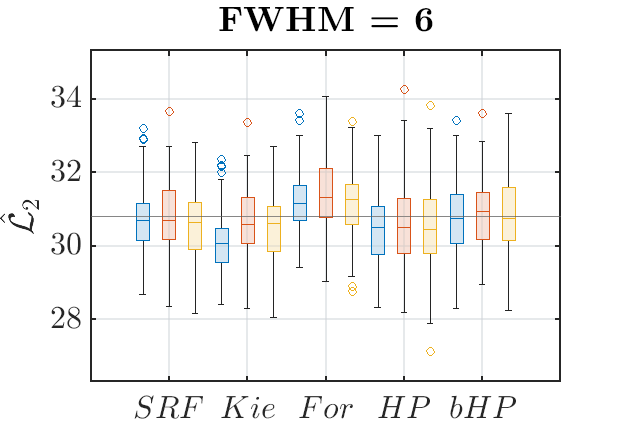}\\
	\caption{ 2D Simulation results for estimation of the LKCs of SuRFs derived from the stationary box example.
			  The results show the dependence of the LKC estimation on the FWHM used in the smoothing kernel for sample size $N=100$.
	          \label{fig:D2L2FixedNSUBJ_stat}}
\end{figure}

\begin{figure}[t!]
\centering
\includegraphics[trim=0 0 40 0,clip,width=1.9in]{\figurepathh 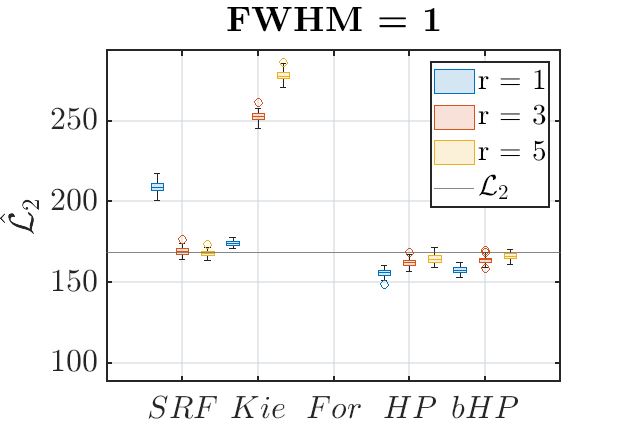}
\includegraphics[trim=0 0 40 0,clip,width=1.9in]{\figurepathh 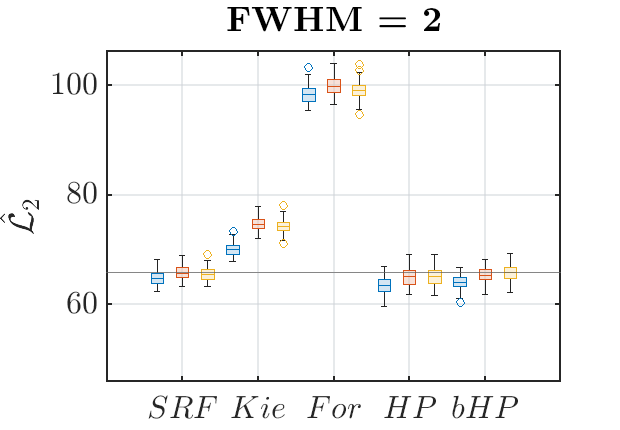}
\includegraphics[trim=0 0 40 0,clip,width=1.9in]{\figurepathh 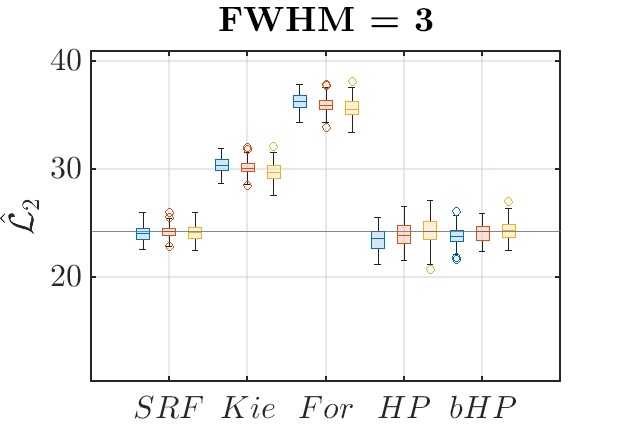}\\
\includegraphics[trim=0 0 40 0,clip,width=1.9in]{\figurepathh 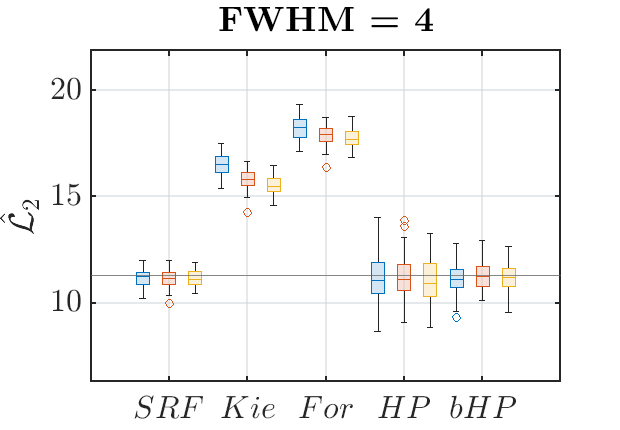}
\includegraphics[trim=0 0 40 0,clip,width=1.9in]{\figurepathh 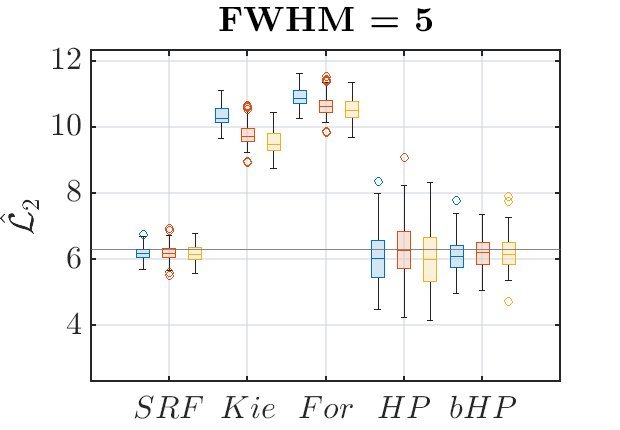}
\includegraphics[trim=0 0 40 0,clip,width=1.9in]{\figurepathh 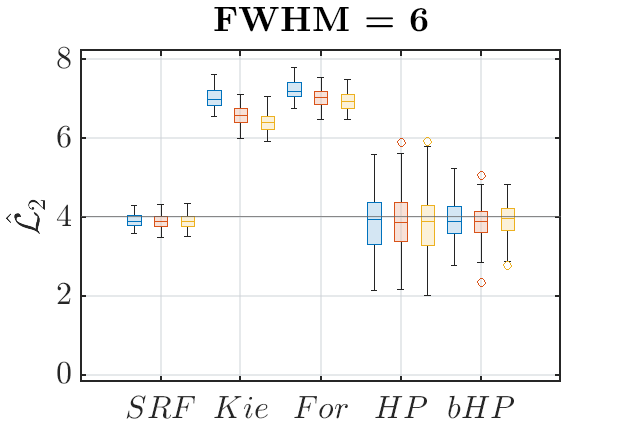}\\
	\caption{ 2D Simulation results for the estimation of the LKCs of SuRFs in the non-stationary setting.
	The results show the dependence of the LKC estimation on the FWHM used in the smoothing kernel
	for sample size $N=100$.
	\label{fig:D2L2FixedNSUBJ_nonstat}}
\end{figure}

\clearpage

\FloatBarrier
\subsection{LKC estimation for $D=3$}
\FloatBarrier

\begin{figure}[h!]
\centering
\includegraphics[trim=0 0 40 0,clip,width=1.6in]{\figurepathh 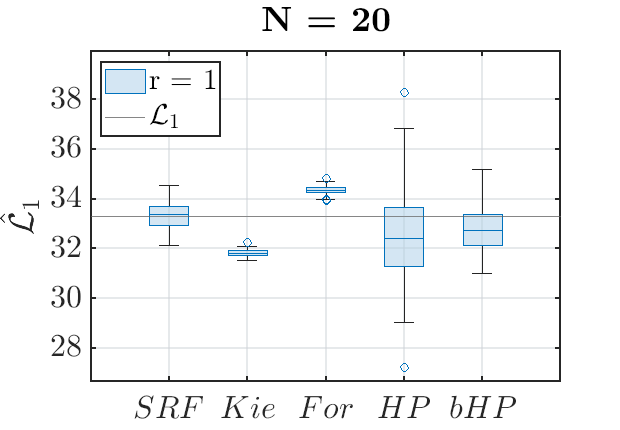}
\includegraphics[trim=0 0 40 0,clip,width=1.6in]{\figurepathh 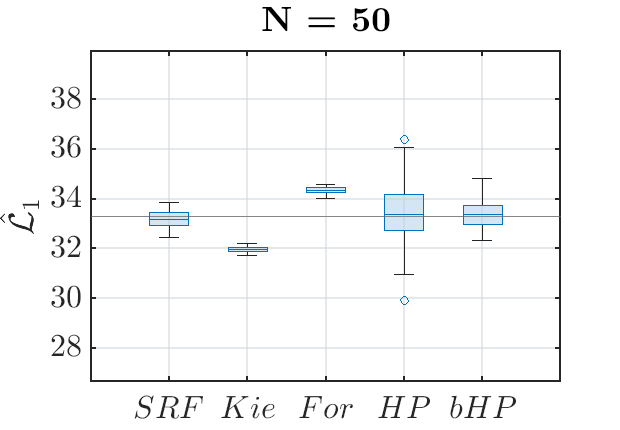}
\includegraphics[trim=0 0 40 0,clip,width=1.6in]{\figurepathh 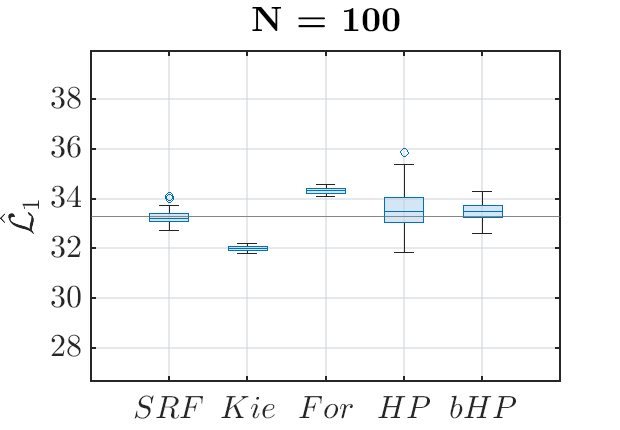}\\
\includegraphics[trim=0 0 40 0,clip,width=1.6in]{\figurepathh 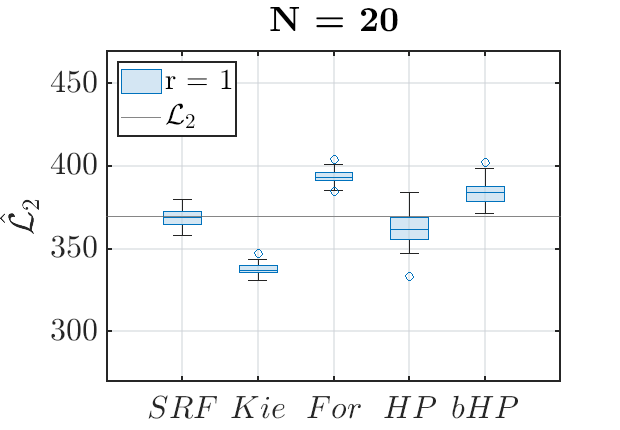}
\includegraphics[trim=0 0 40 0,clip,width=1.6in]{\figurepathh 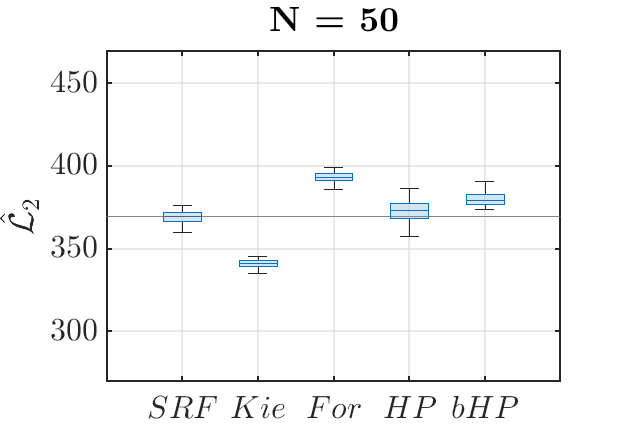}
\includegraphics[trim=0 0 40 0,clip,width=1.6in]{\figurepathh 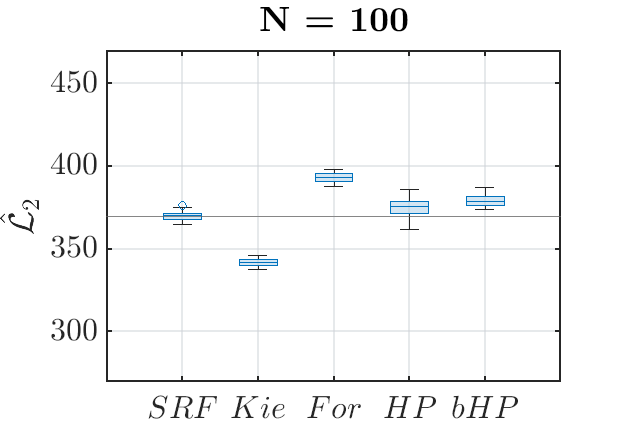}\\
\includegraphics[trim=0 0 40 0,clip,width=1.6in]{\figurepathh 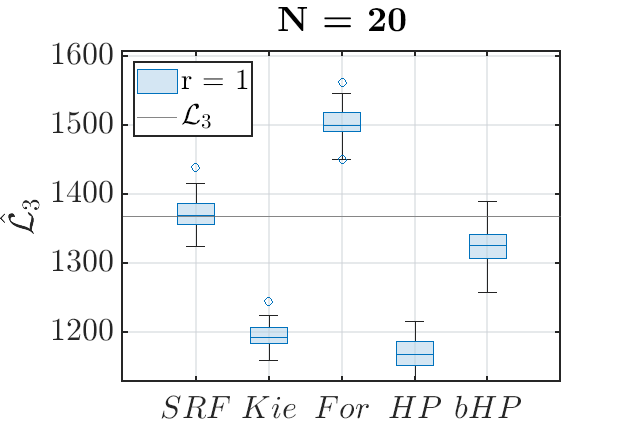}
\includegraphics[trim=0 0 40 0,clip,width=1.6in]{\figurepathh 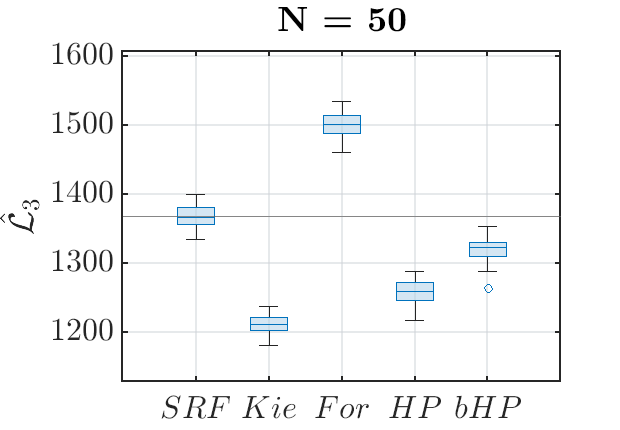}
\includegraphics[trim=0 0 40 0,clip,width=1.6in]{\figurepathh 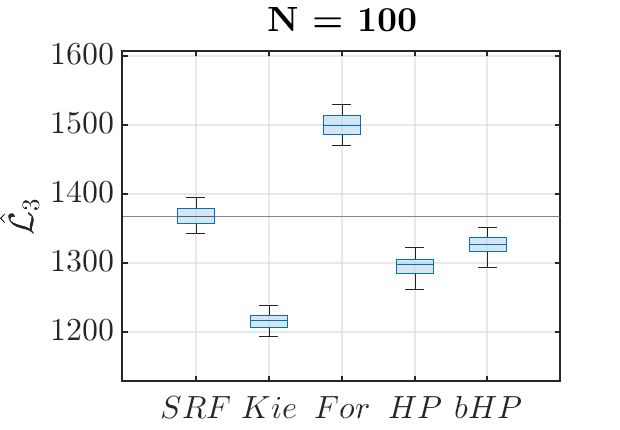}
	\caption{ 3D Simulation results of estimation of the LKCs of the SuRFs
			  derived from the stationary box example. The FWHM is  $f=3$.
			  \label{fig:D3_stationary_FixedFWHM}}
\end{figure}

\begin{figure}[h!]
\centering
\includegraphics[trim=0 0 40 0,clip,width=1.6in]{\figurepathh 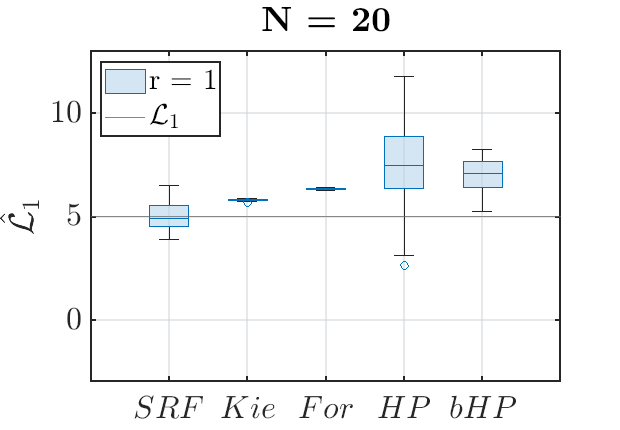}
\includegraphics[trim=0 0 40 0,clip,width=1.6in]{\figurepathh 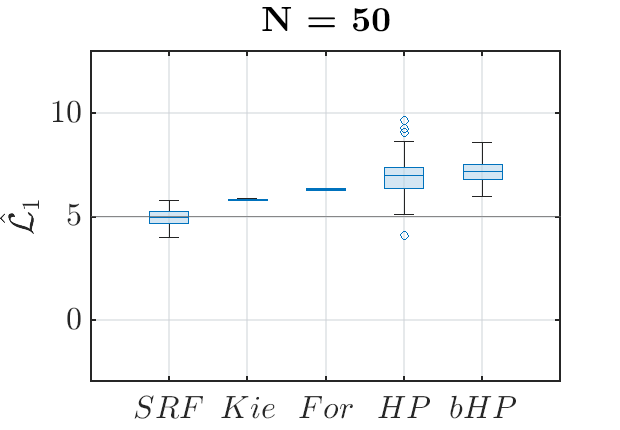}
\includegraphics[trim=0 0 40 0,clip,width=1.6in]{\figurepathh 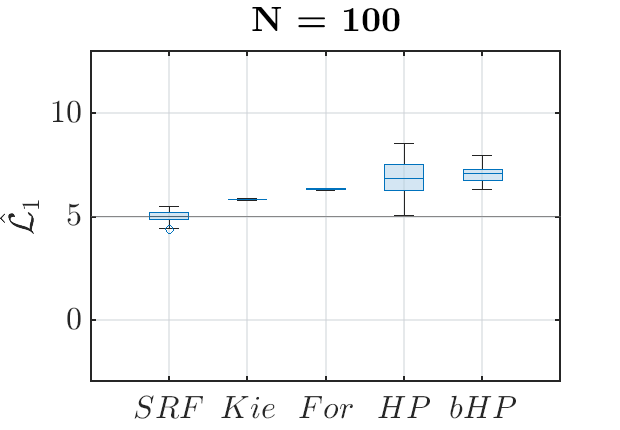}\\
\includegraphics[trim=0 0 40 0,clip,width=1.6in]{\figurepathh 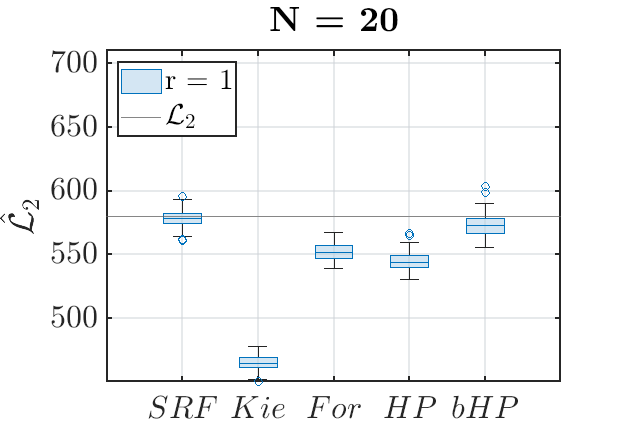}
\includegraphics[trim=0 0 40 0,clip,width=1.6in]{\figurepathh 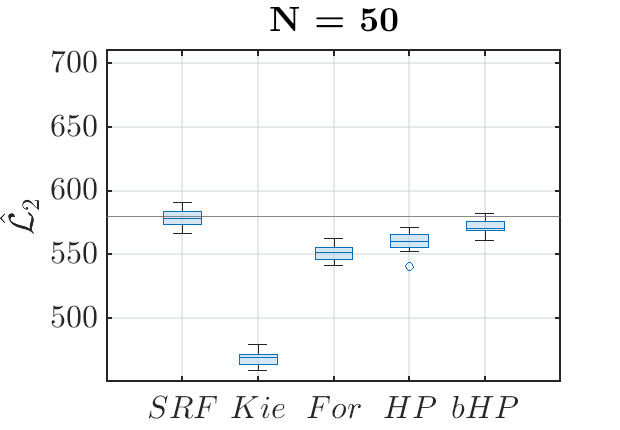}
\includegraphics[trim=0 0 40 0,clip,width=1.6in]{\figurepathh 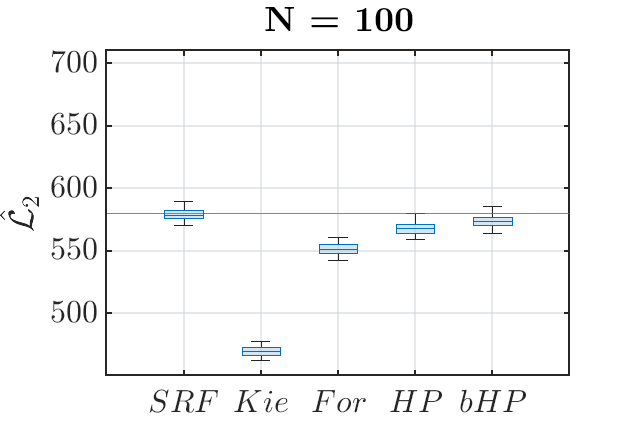}\\
\includegraphics[trim=0 0 40 0,clip,width=1.6in]{\figurepathh 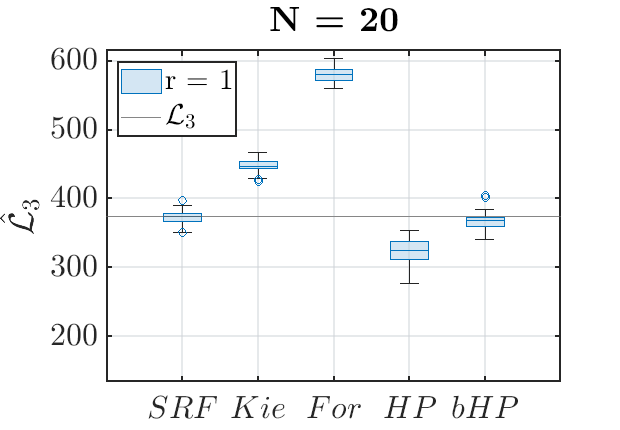}
\includegraphics[trim=0 0 40 0,clip,width=1.6in]{\figurepathh 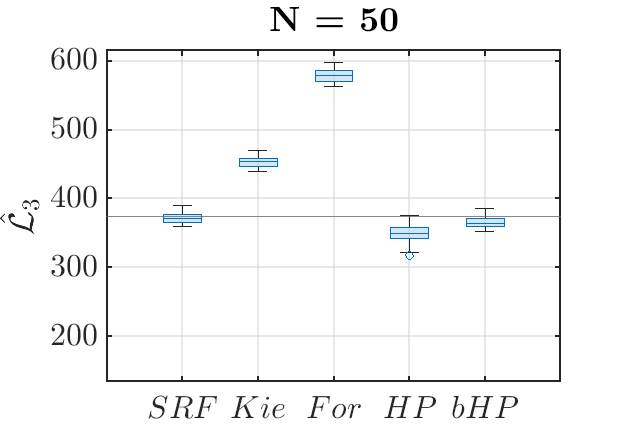}
\includegraphics[trim=0 0 40 0,clip,width=1.6in]{\figurepathh 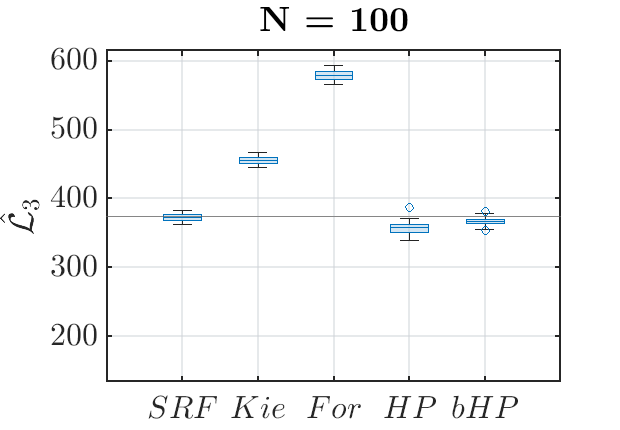}
	\caption{ 3D Simulation results of estimation of
			  the LKCs of the two SuRFs derived from the non-stationary
			  sphere example. The FWHM is  $f=3$. Note that the theoretical
			  value for $\mathcal{L}_1$ is the theoretical value for the locally
			  stationary $\mathcal{L}_1$. The true value for $\mathcal{L}_1$
			  is currently infeasible to obtain.
			  \label{fig:D3_nonstationary_FixedFWHM}}
\end{figure}

\begin{figure}[h!]
\centering
\includegraphics[trim=0 0 40 0,clip,width=1.6in]{\figurepathh 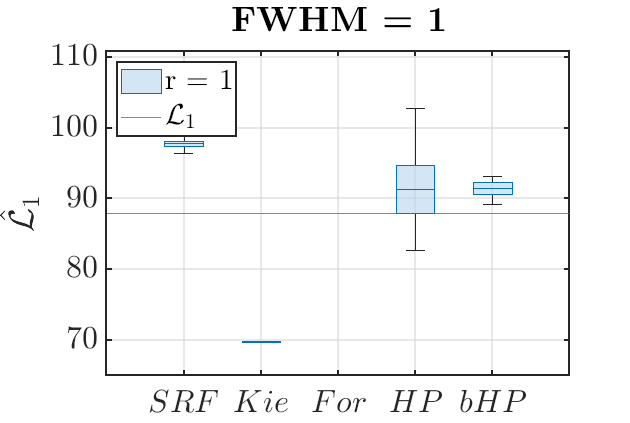}
\includegraphics[trim=0 0 40 0,clip,width=1.6in]{\figurepathh 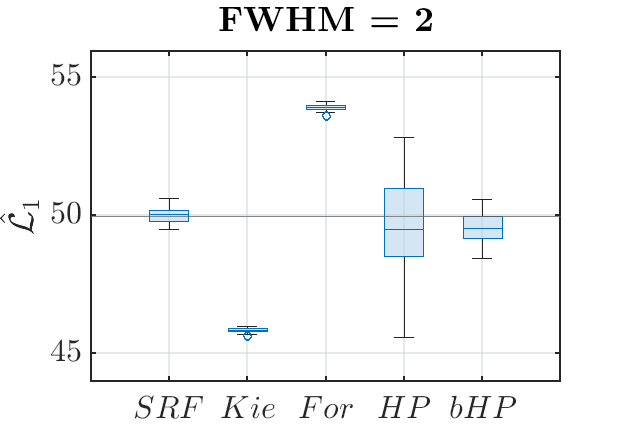}
\includegraphics[trim=0 0 40 0,clip,width=1.6in]{\figurepathh 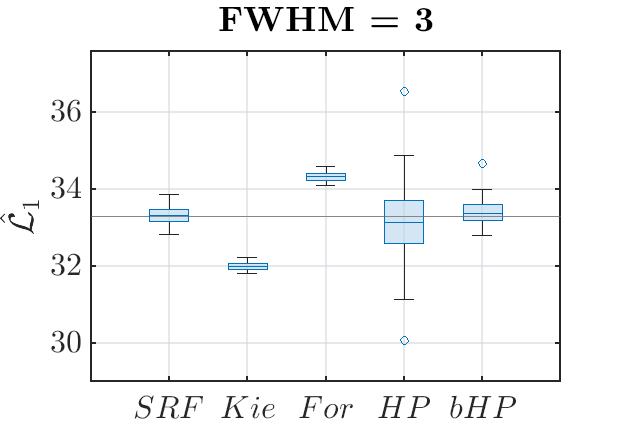}\\
\includegraphics[trim=0 0 40 0,clip,width=1.6in]{\figurepathh 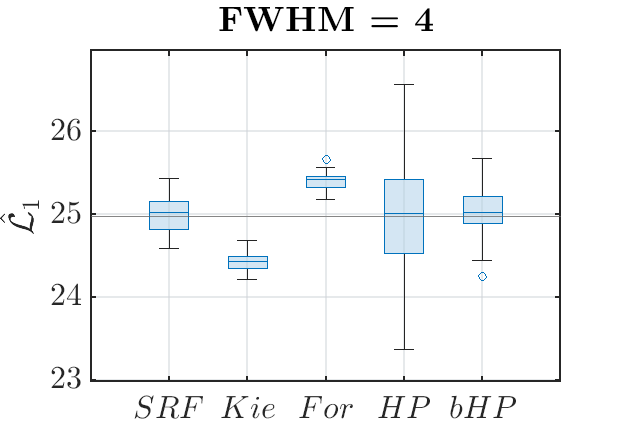}
\includegraphics[trim=0 0 40 0,clip,width=1.6in]{\figurepathh 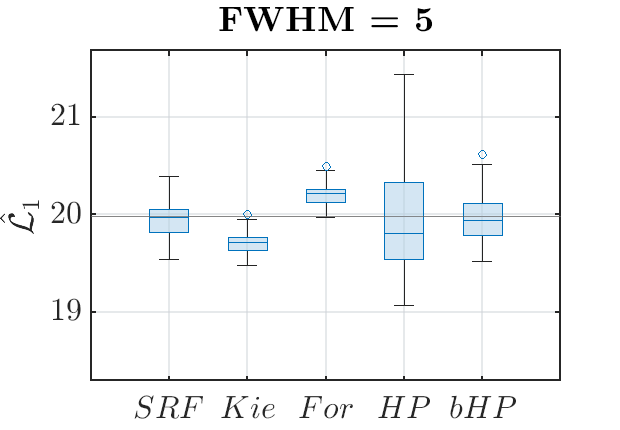}
\includegraphics[trim=0 0 40 0,clip,width=1.6in]{\figurepathh 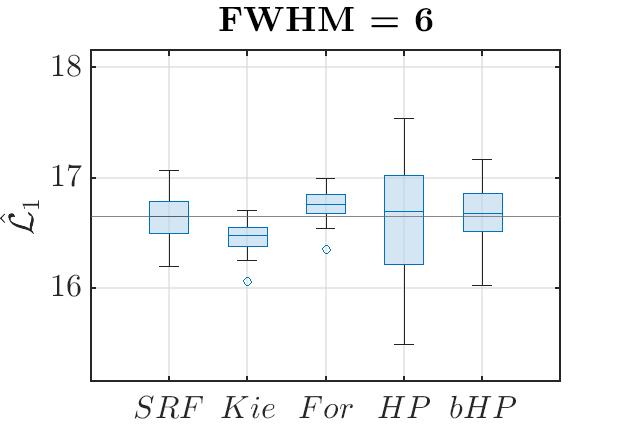}\\
	\caption{ 3D Simulation results of estimation of the LKCs of SuRFs derived from
			  the almost stationary box example.
			  The results show the dependence of the LKC estimation on the FWHM used
			  in the smoothing kernel for sample size $N=100$.
	          \label{fig:D3L1FixedNSUBJ_stat}}
\end{figure}

\begin{figure}[h!]
\centering
\includegraphics[trim=0 0 40 0,clip,width=1.6in]{\figurepathh 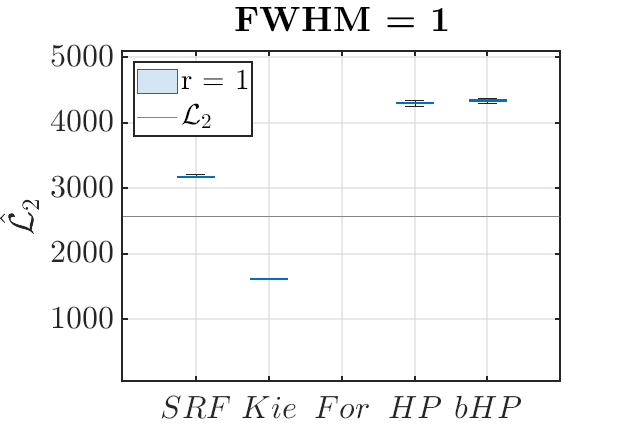}
\includegraphics[trim=0 0 40 0,clip,width=1.6in]{\figurepathh 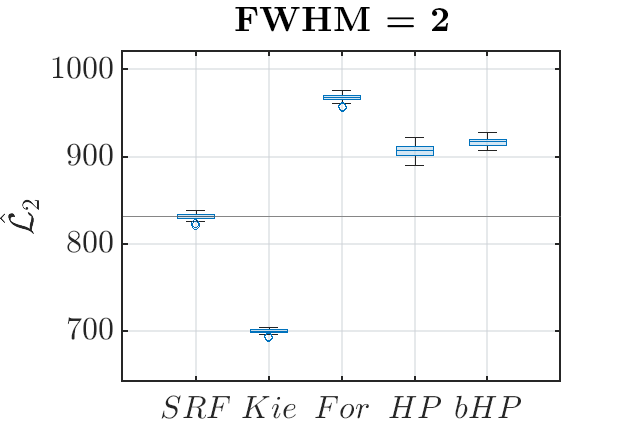}
\includegraphics[trim=0 0 40 0,clip,width=1.6in]{\figurepathh 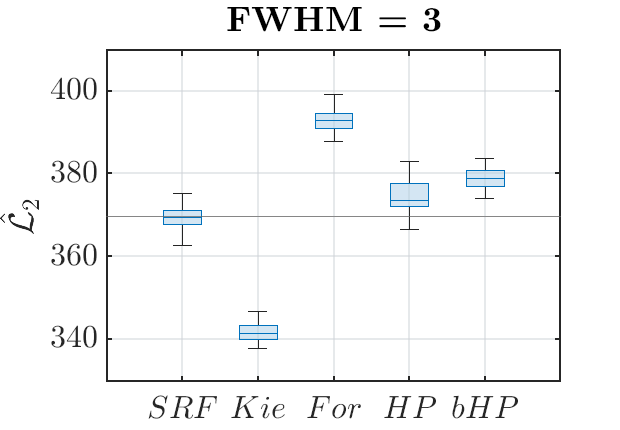}\\
\includegraphics[trim=0 0 40 0,clip,width=1.6in]{\figurepathh 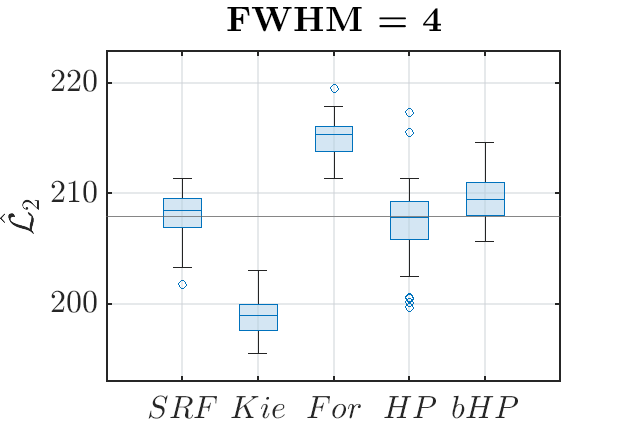}
\includegraphics[trim=0 0 40 0,clip,width=1.6in]{\figurepathh 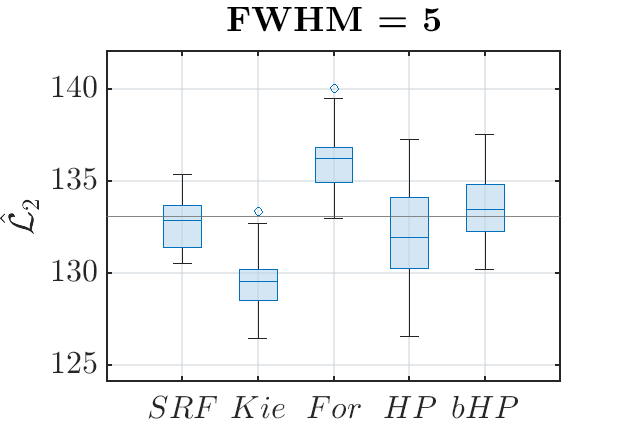}
\includegraphics[trim=0 0 40 0,clip,width=1.6in]{\figurepathh 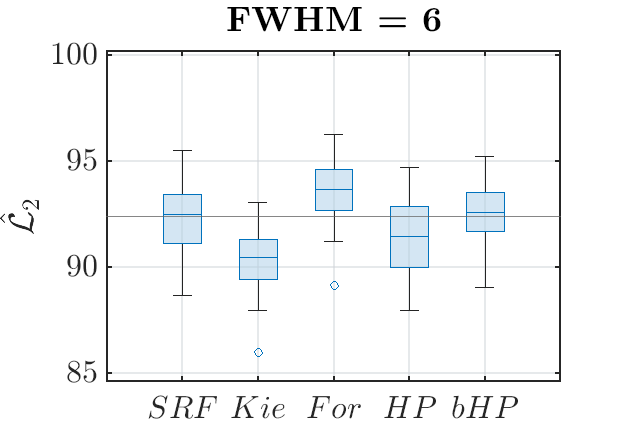}\\
	\caption{ 3D Simulation results of estimation of the LKCs of SuRFs derived from the almost
	stationary box example.
	The results show the dependence of the LKC estimation on the FWHM used in the smoothing kernel
	for sample size $N=100$.
	          \label{fig:D3L2FixedNSUBJ_stat}}
\end{figure}

\begin{figure}[h!]
\centering
\includegraphics[trim=0 0 40 0,clip,width=1.6in]{\figurepathh 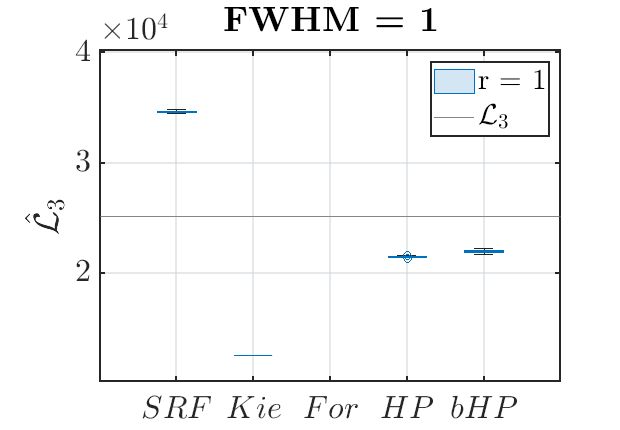}
\includegraphics[trim=0 0 40 0,clip,width=1.6in]{\figurepathh 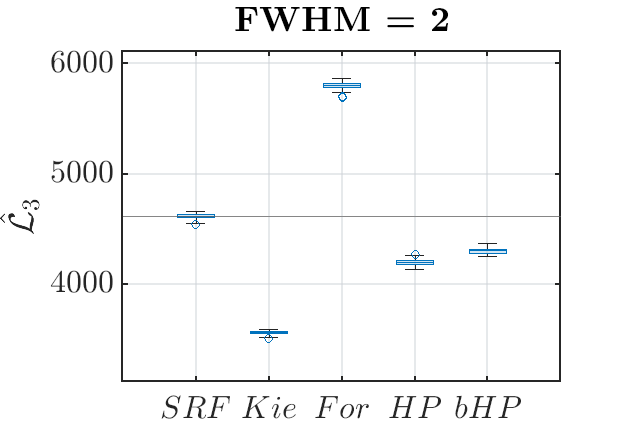}
\includegraphics[trim=0 0 40 0,clip,width=1.6in]{\figurepathh 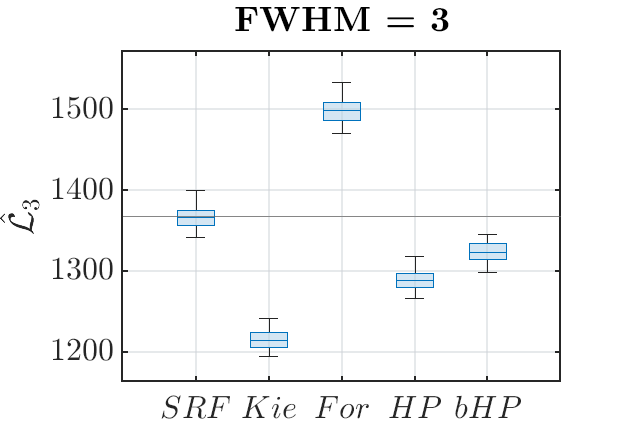}\\
\includegraphics[trim=0 0 40 0,clip,width=1.6in]{\figurepathh 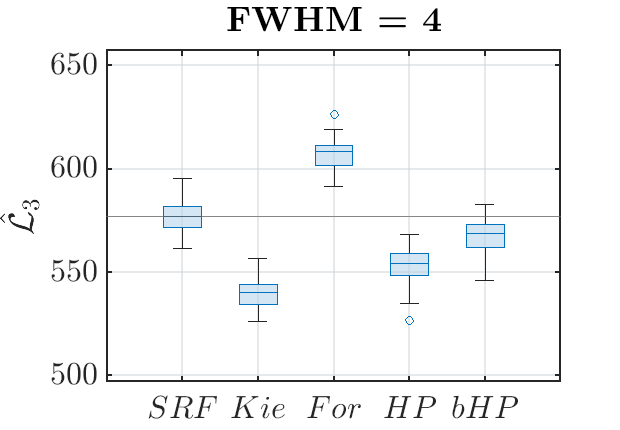}
\includegraphics[trim=0 0 40 0,clip,width=1.6in]{\figurepathh 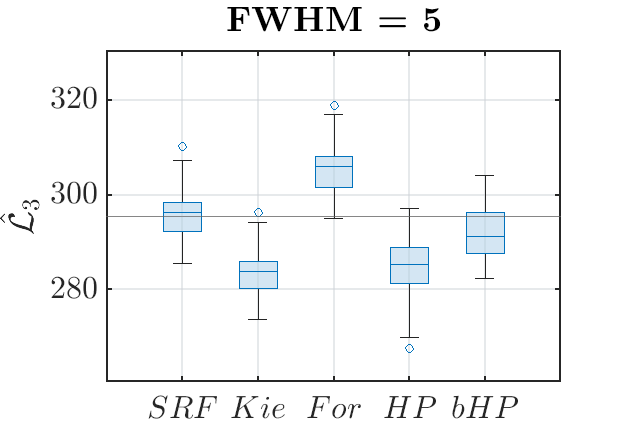}
\includegraphics[trim=0 0 40 0,clip,width=1.6in]{\figurepathh 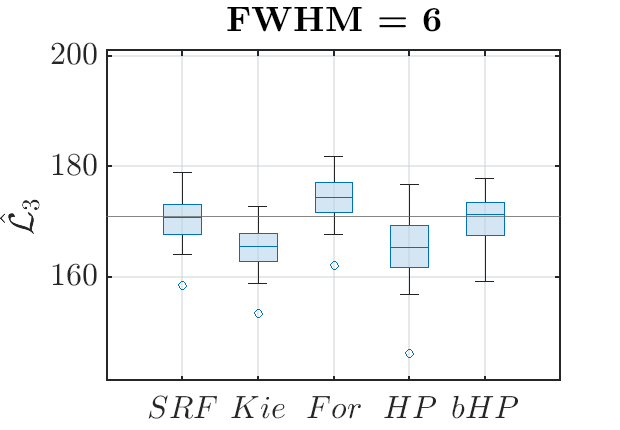}\\
	\caption{ 3D Simulation results of estimation of the LKCs of SuRFs derived
			 from the almost-stationary sphere example.
			  The results show the dependence of the LKC estimation on the
			  FWHM used in the smoothing kernel for sample size $N=100$.
	          \label{fig:D3L3FixedNSUBJ_stat}}
\end{figure}

\begin{figure}[h!]
\centering
\includegraphics[trim=0 0 40 0,clip,width=1.6in]{\figurepathh 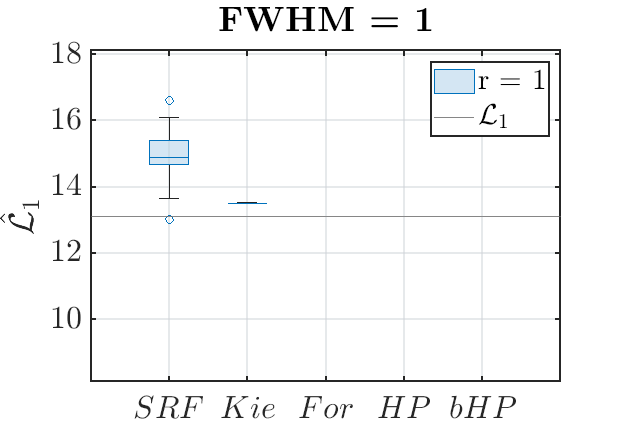}
\includegraphics[trim=0 0 40 0,clip,width=1.6in]{\figurepathh 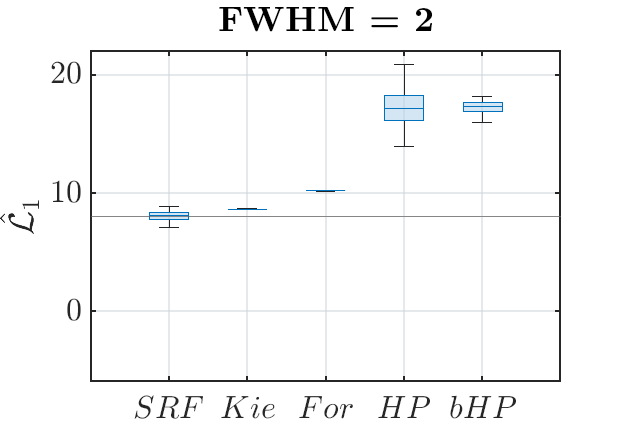}
\includegraphics[trim=0 0 40 0,clip,width=1.6in]{\figurepathh 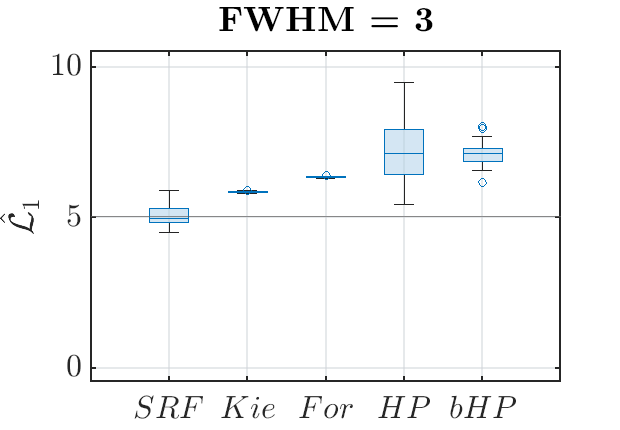}\\
\includegraphics[trim=0 0 40 0,clip,width=1.6in]{\figurepathh 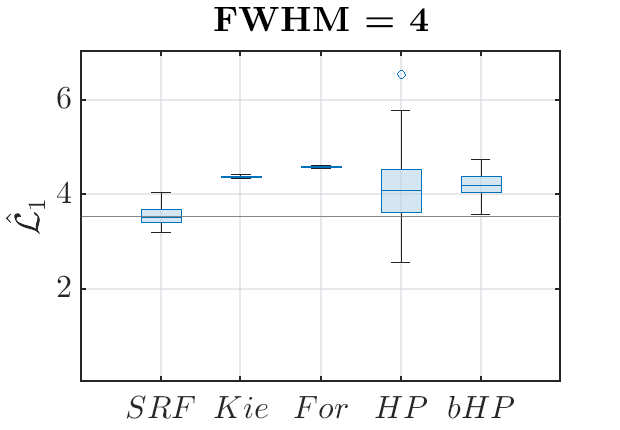}
\includegraphics[trim=0 0 40 0,clip,width=1.6in]{\figurepathh 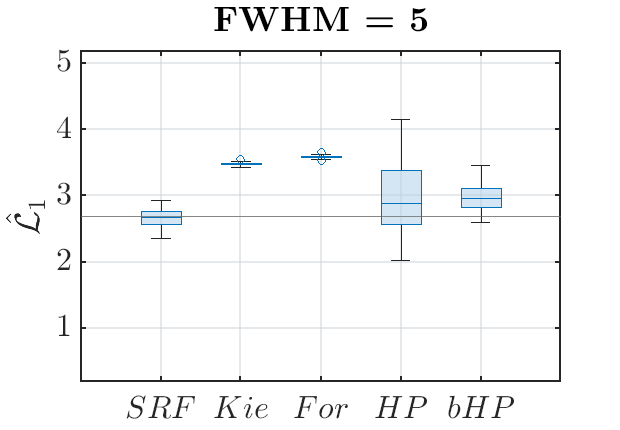}
\includegraphics[trim=0 0 40 0,clip,width=1.6in]{\figurepathh 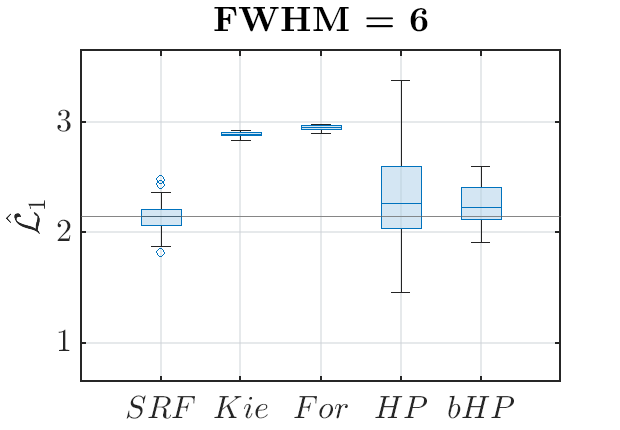}\\
	\caption{ 3D Simulation results of estimation of the LKCs of
			  SuRFs derived from the non-stationary sphere example.
			  The results show the dependence of the LKC estimation
			  on the FWHM used in the smoothing kernel for sample size
			  $N=100$. Note that the theoretical value for $\mathcal{L}_1$ is the
			  theoretical value for the locally stationary $\mathcal{L}_1$. The true value for
			  $\mathcal{L}_1$ is currently infeasible to obtain.
	          \label{fig:D3L1FixedNSUBJ_nonstat}}
\end{figure}

\begin{figure}[h!]
\centering
\includegraphics[trim=0 0 40 0,clip,width=1.6in]{\figurepathh 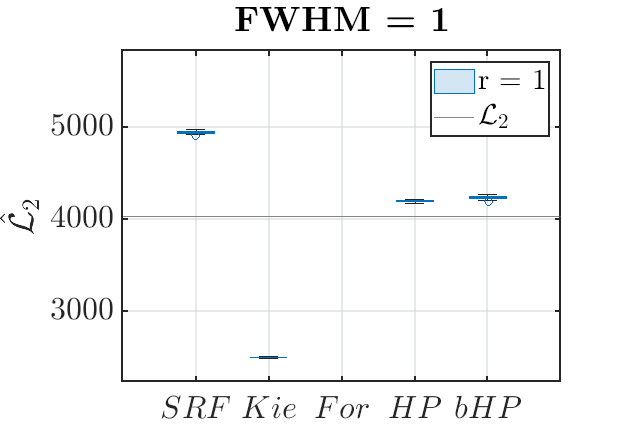}
\includegraphics[trim=0 0 40 0,clip,width=1.6in]{\figurepathh 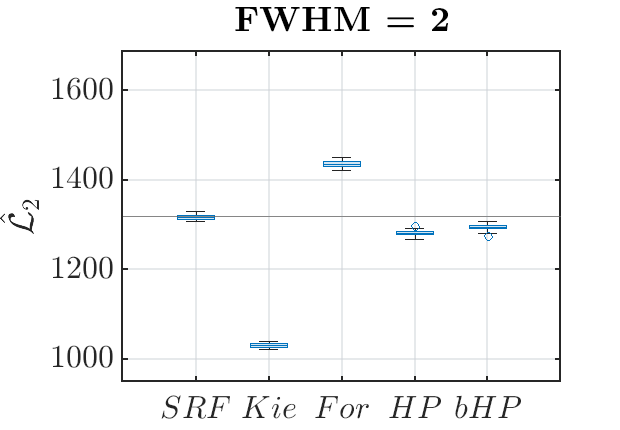}
\includegraphics[trim=0 0 40 0,clip,width=1.6in]{\figurepathh 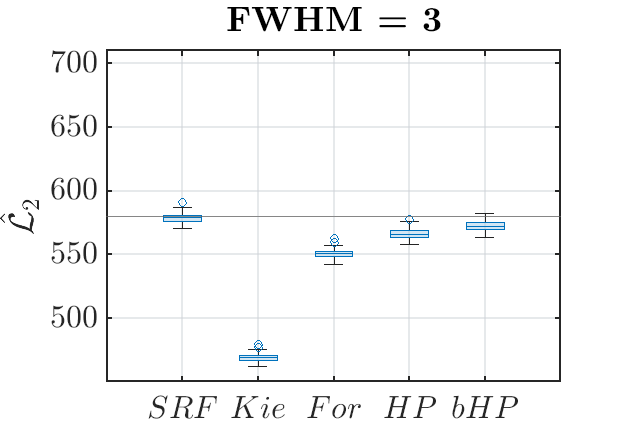}\\
\includegraphics[trim=0 0 40 0,clip,width=1.6in]{\figurepathh 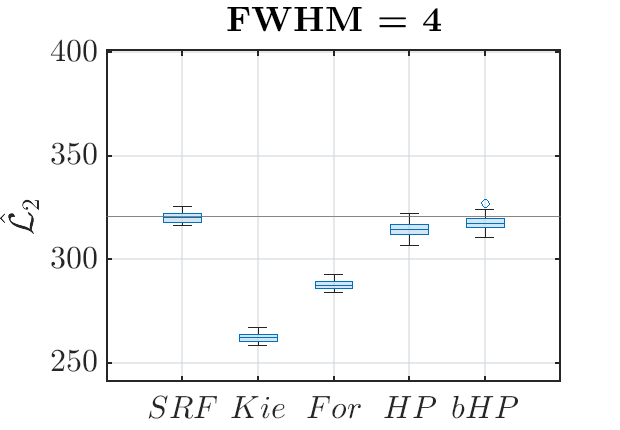}
\includegraphics[trim=0 0 40 0,clip,width=1.6in]{\figurepathh 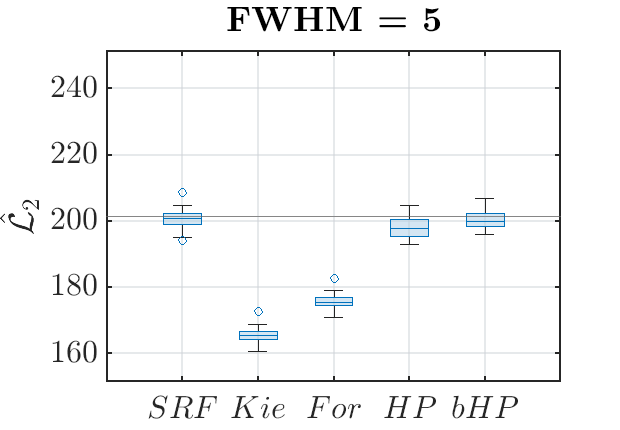}
\includegraphics[trim=0 0 40 0,clip,width=1.6in]{\figurepathh 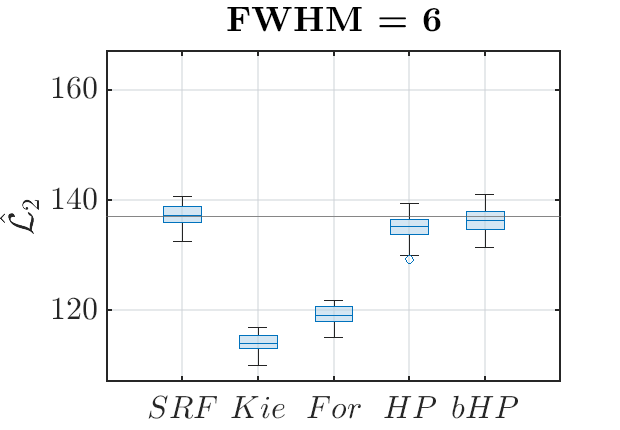}\\
	\caption{ 3D Simulation results of estimation of the LKCs of SuRFs
			  derived from the non-stationary sphere example.
			  The results show the dependence of the LKC estimation on
			  the FWHM used in the smoothing kernel for sample size $N=100$.
	          \label{fig:D3L2FixedNSUBJ_nonstat}}
\end{figure}
\newpage
\FloatBarrier

\begin{figure}[t!]
\begin{center}
\includegraphics[trim=0 0 40 0,clip,width=1.6in]{\figurepathh 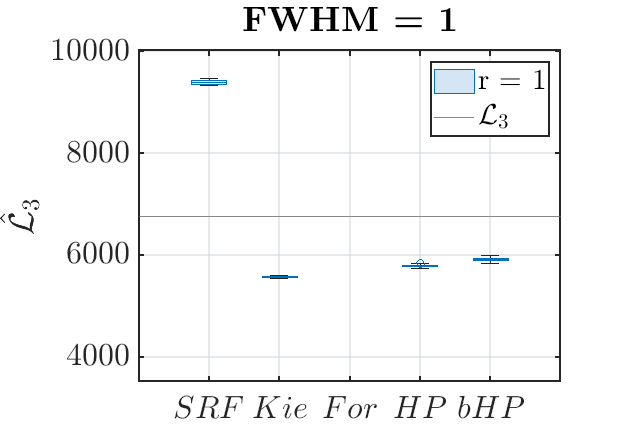}
\includegraphics[trim=0 0 40 0,clip,width=1.6in]{\figurepathh 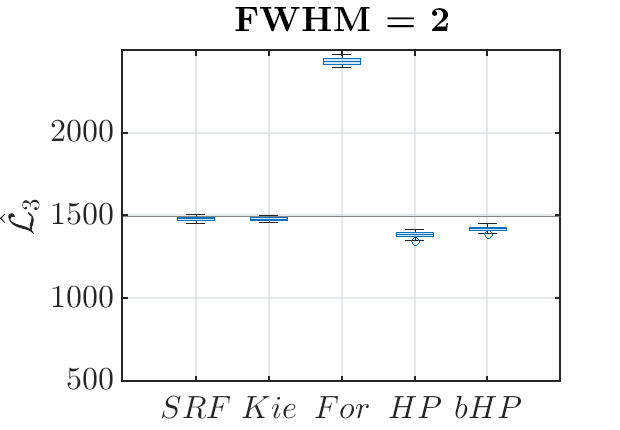}
\includegraphics[trim=0 0 40 0,clip,width=1.6in]{\figurepathh 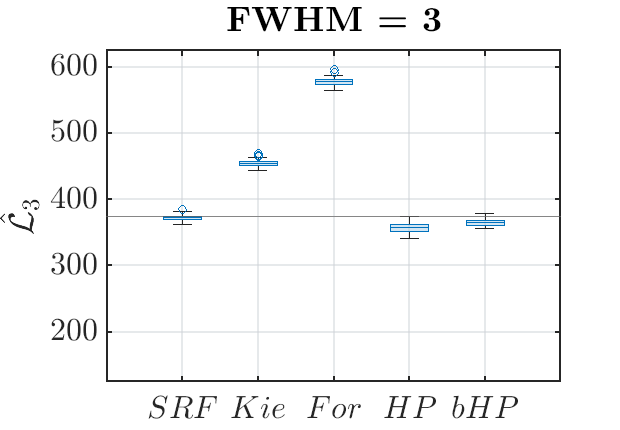}\\
\includegraphics[trim=0 0 40 0,clip,width=1.6in]{\figurepathh 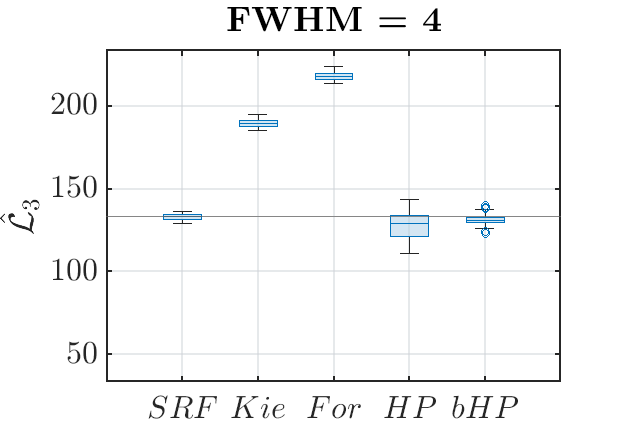}
\includegraphics[trim=0 0 40 0,clip,width=1.6in]{\figurepathh 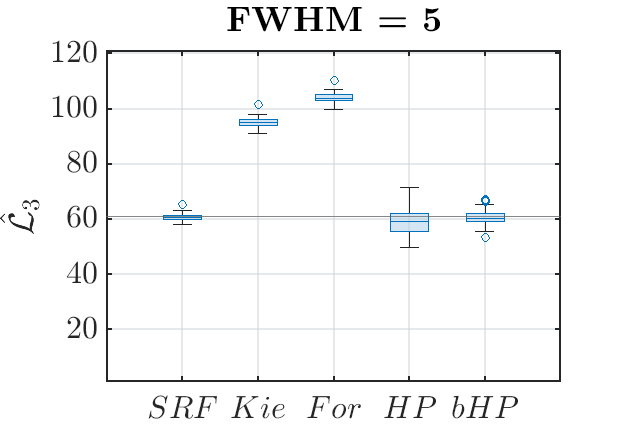}
\includegraphics[trim=0 0 40 0,clip,width=1.6in]{\figurepathh 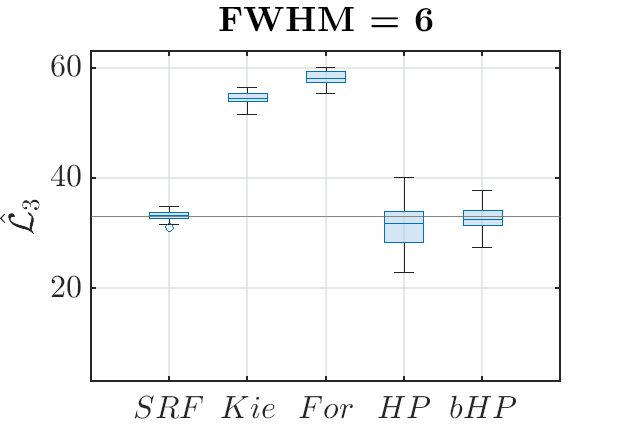}\\
	\caption{ 3D Simulation results of estimation of the LKCs
	          of SuRFs derived from the non-stationary sphere example.
			  The results show the dependence of the LKC estimation on
			  the FWHM used in the smoothing kernel for sample size $N=100$.
	          \label{fig:D3L3FixedNSUBJ_nonstat}}

\end{center}
\end{figure}

\subsection{Comparison of the LKCs of the Almost Stationary Box Example and its Stationary Counterpart}\label{App:LKCapprox}

Here we give an example of approximating the theoretical LKCs
(up to an arbitrary resolution increase to approximate the integrals) for the special case of a normalized SuRF which is derived from a random vector $(X(v):v \in \mathcal{V})$
satisfying $X(v_1)$ being independent of $X(v_2)$ for all $v_1\neq v_2 \in \cV$.
Because of the mutual independence of the entries of $X$ in this case we obtain
\begin{equation*}
        \Cov{ \tilde X(x)}{ \tilde X(x') } = 
    			  \sum_{v\in\cV} K( x, v )K( x', v )\,,~~~x,x' \in \cM_\cV\,.
\end{equation*}
and thus the Riemannian metric \eqref{eq:riem_SuRF} induced by the normalized SuRF of $X$ is
\begin{equation}\label{eq:riem_SuRFest_independent}
\begin{split}
        \Lambda_{dd'}(x) =
        	&\frac{ \sum_{v\in\cV} \partial_d K( x, v ) \partial_{d'} K( x, v ) }
        		 { \sum_{v\in\cV} K( x, v )^2 }\\
       	  &\quad- \frac{ 	\sum_{v\in\cV} K( x, v ) \partial_{d} K( x, v )
       	  		   			\sum_{v\in\cV} K( x, v ) \partial_{d'} K( x, v ) }
       	  		 { \left( \sum_{v\in\cV} K( x, v )^2 \right)^2 }\,.
\end{split}
\end{equation}
Replacing $\hat{\boldsymbol{\Lambda}}$ in \eqref{eq:LKCestim_VM} by $\boldsymbol{\Lambda}(x) \in \mathbb{R}^{D \times D}$
with $(d,d')$th entry given by \eqref{eq:riem_SuRFest_independent}, we can approximate the
LKCs of the normalized SuRF of $X$ up to arbitrary precision on a computer.
We implemented this in the function \emph{LKC\_wncfield\_theory()}
from the RFTtoolbox \cite{RFTtoolbox}.

We can use this function to demonstrate that the LKCs from the almost
stationary box example derived from the Gaussian Kernel \eqref{eq:smooth_kernel} for $f>0$,
are essentially the same as the LKCs of the zero-mean Gaussian field 
with covariance function
$\mathfrak{c}_f(x,y) = \exp\big( -4\log(2)\vert x-y \vert^2 / f^2 \big)$
provided that $f$ is larger than $\approx 2.5$. We denote the correlation function of the
almost stationary box example with $\tilde{\mathfrak{c}}_f(x,y)$.

We compare the resulting LKCs to the LKCs of the random field with covariance function $\mathfrak{c}_f(x,y)$ over a convex
domain $S\subset\mathbb{R}^D$, $D\in\{1,2,3\}$, which are given by
\begin{equation*}
	\begin{split}
	D=1:&\quad \cL_1 = \frac{\vol(S)}{\sqrt{4\log(2)}f}\\
	D=2:&\quad \cL_1 = \frac{\vol(\partial S)}{2\sqrt{4\log(2)}f}\,,\quad
				\cL_2 = \frac{\vol(S)}{4\log(2)f^2}\\
	D=3:&\quad \cL_1 = \frac{2{\rm Diameter}(S)}{\sqrt{4\log(2)}f}\,,\quad
				\cL_2 = \frac{\vol(\partial S)}{2\cdot 4\log(2) f^2}\,,\quad  \cL_3 = \frac{\vol(S)}{f^3 \big(4\log(2)\big)^{3/2}}\,,
	\end{split}
\end{equation*}
compare for example \cite{Telschow2020} and in particular Table 2 from \cite{Worsley2004} for $D=3$.
To approximate the theoretical LKCs of the almost stationary field we use an added resolution of $11$ for $D\in \{1,2\}$
and for the sake of computation time only an added resolution of $7$ for $D=3$. The results
are presented in Tables \ref{Tab:LKCapproxD1}-\ref{Tab:LKCapproxD3}.

\newpage
\begin{table}\setlength{\tabcolsep}{4pt}
\begin{center}
\begin{tabular}{ | c | c | c c c c c c c | }\hline
 $D=1$ & $f$ & 1 & 1.5 & 2 & 2.5 & 3 & 3.5 & 4 \\\hline
 \multirow{ 2}{*}{$\cL_1$} & $\tilde{\mathfrak{c}}_f$ & 146.52 &  110.41 &   83.25 &   66.60 & 55.50 & 47.57 &   41.63  \\ 
  & $\mathfrak{c}_f$ & 166.51 & 111.01 &  83.26 & 66.60 & 55.50 & 47.57 & 41.63  \\ 
 \hline
\end{tabular}
\end{center}
	\caption{Comparison of the theoretical LKCs for different smoothing bandwidths
	between the almost stationary box example for $D=1$, which has
	the covariance function $\tilde{\mathfrak{c}}_f$, and the zero-mean stationary
	Gaussian field	having covariance $\tilde{\mathfrak{c}}_f$ over the domains specified in
	Section \ref{scn:setup} of the main manuscript.}
\label{Tab:LKCapproxD1}
\end{table}

\begin{table}\setlength{\tabcolsep}{4pt}
\begin{center}
\begin{tabular}{ | c | c | c c c c c c c | }\hline
 $D=2$ & $f$ & 1 & 1.5 & 2 & 2.5 & 3 & 3.5 & 4 \\\hline
 \multirow{ 2}{*}{$\cL_1$} & $\tilde{\mathfrak{c}}_f$ & 58.61 &  44.16 & 33.30 & 26.64 & 22.20 & 19.03 & 16.65 \\ 
  & $\mathfrak{c}_f$ & 66.60 & 44.40 & 33.30 & 26.64 & 22.20 & 19.03 & 16.65\\ 
 \hline
 \multirow{ 2}{*}{$\cL_2$} & $\tilde{\mathfrak{c}}_f$ & 858.72 & 487.59 & 277.24 & 177.45 & 123.23 &  90.53 & 69.31\\
  & $\mathfrak{c}_f$ & 1109.00  & 492.90 & 277.26 & 177.45 & 123.23 & 90.53 & 69.31\\  
 \hline
\end{tabular}
\end{center}
	\caption{Comparison of the theoretical LKCs for different smoothing bandwidths
	between the almost stationary box example for $D=2$, which has
	the covariance function $\tilde{\mathfrak{c}}_f$, and the zero-mean stationary
	Gaussian field	having covariance $\tilde{\mathfrak{c}}_f$ over the domains specified in
	Section \ref{scn:setup} of the main manuscript.}
\label{Tab:LKCapproxD2}
\end{table}

\begin{table}\setlength{\tabcolsep}{4pt}
\begin{center}
\begin{tabular}{ | c | c | c c c c c c c | }\hline
 $D=3$ & $f$ & 1 & 1.5 & 2 & 2.5 & 3 & 3.5 & 4 \\\hline
 \multirow{ 2}{*}{$\cL_1$} & $\tilde{\mathfrak{c}}_f$ & 87.91 & 66.24 & 49.95 &  39.96 & 33.30 & 28.54 & 24.98  \\ 
  & $\mathfrak{c}_f$ & 99.91 & 66.60 & 49.95 & 39.96 & 33.30 & 28.54 & 24.98  \\ 
 \hline
 \multirow{ 2}{*}{$\cL_2$} & $\tilde{\mathfrak{c}}_f$ & 2576.13 &  1462.77 & 831.72 & 532.34 & 369.68 & 271.60 & 207.94 \\
  & $\mathfrak{c}_f$ & 3327.11 & 1478.71 & 831.78 & 532.34 & 369.68 & 271.6 & 207.94 \\  
 \hline
 \multirow{ 2}{*}{$\cL_3$} & $\tilde{\mathfrak{c}}_f$ & 25163.37 & 10766.66 & 4616.20 & 2363.73 & 1367.90 & 861.42 & 577.08
\\
  & $\mathfrak{c}_f$ & 36933.30 & 10943.20 & 4616.66 & 2363.73 & 1367.90 & 861.42 & 577.08
\\\hline
\end{tabular}
\end{center}
	\caption{Comparison of the theoretical LKCs for different smoothing bandwidths
	between the almost stationary box example for $D=3$, which has
	the covariance function $\tilde{\mathfrak{c}}_f$, and the zero-mean stationary
	Gaussian field	having covariance $\tilde{\mathfrak{c}}_f$ over the domains specified in
	Section \ref{scn:setup} of the main manuscript.}
\label{Tab:LKCapproxD3}
\end{table}

\newpage
\FloatBarrier

\section{LKCs induced by a normalized field}\label{appendix:LKCS}

\subsection{Induced Riemannian metric of a normalized random field}\label{app:InducedRiemann}

The most important quantity for the Gaussian Kinematic formula is the Riemannian metric induced
by a random field. It is the backbone of the GKF for Gaussian related fields developed in
\cite{Taylor2006}. 

In this section $f$ denotes a zero-mean random field with almost surely continuously,
differentiable sample paths over the domain $\overline{\cM}$
and we call the random field $f / \sqrt{ \Var[f] }$ the normalized field
derived from $f$.
Recall that a  vector field $V \in \cT \overline{\cM}$ can be interpreted as a
first order differential operator, i.e., for all $h\in C^1(\overline{\cM})$ the
expression $Vh: \overline{\cM} \rightarrow \mathbb{R}$,
$s \mapsto (Vh)(s)$ defines a
function in $C(\overline{\cM})$. This can be made precise by taking a local
chart $\varphi: \overline{\cM} \supseteq \mathcal{U} \rightarrow \varphi(\mathcal{U}) = \cW \subset\mR^D$ of
$\overline{\cM}$ with inverse $\psi = \varphi^{-1}$. The vector field $V$
in local coordinates $(\cU, \varphi)$ can be written as
\begin{equation}
	V = \sum_{d=1}^D V_d \partial_d\,,~~V_1,\ldots, V_D\in C( \cW )
\end{equation}
where $\partial_d$ is the vector field on $\cW$ defined
by
\begin{equation}
	\partial_d(h\circ\psi)(x) = \frac{\partial h\circ\psi }{\partial x_d}(x) = \partial_{d}^x h(\psi (x))\,,~~ d=1,...,D\,
\end{equation}
where $h \in C( \mathcal{U} )$ and $x\in\varphi(\mathcal{U})$.
\begin{definition}
    Let $V,W\in \cT \overline{\cM}$ be differentiable vector fields and $f$ unit-variance random field over a manifold $\overline{\cM}$ with almost
    surely differentiable sample paths. Then
	\begin{equation}\label{eq:InducedMetricI}
		\bar \Lambda_s(V,W) = \Cov{ V\!f(s)}{ W\!f(s) }\,,~s\in \overline{\cM}\,,
	\end{equation}
 	is called the \emph{induced Riemannian metric of $f$} on $\overline{\cM}$.
 	In local coordinates at a point $z\in\cW$ it is represented by
	\begin{equation}
 	\begin{split}
 		\bar \Lambda_z(V,W)
 				 &= \sum_{d, d'=1}^D V_d W_{d'} \partial_{d}^x\partial_{d'}^y\Cov{ f\big(\psi(x)\big)}{ f\big(\psi(y)\big) }  \Big\vert_{(x,y)=(z,z)}\\
 		 				 &= \sum_{d, d'=1}^D V_d W_{d'} \Cov{ \partial_{d}^x f\big(\psi(x)\big) \big\vert_{x=z}}{ \partial_{d'}^y f\big(\psi(y)\big) \big\vert_{y=z} }
 		\,.
 	\end{split}
 	\end{equation}
 	Here $V_1,\ldots,V_D, W_1,\ldots, W_D \in C(\cW)$ are the coordinate coefficients representing the vector
 	fields $V$ and $W$.
\end{definition}
\begin{remark}
	Assumption \textbf{(G2)} from the main manuscript ensures that this is a
	Riemannian metric on $\overline{\cM}$.
\end{remark}

Since the vector fields $\partial_{d}$, $d=1,\ldots, D$,  form a basis of $\mathcal{T}\mathcal{U}$ the Riemannian metric induced by
the random field $f$ can be written in local coordinates as the $D\times D$ matrix having components
\begin{equation}
	 \bar \Lambda_{dd'}( z ) = \Cov{ \partial_{d}^x f\big(\psi(x)\big)\big
\vert_{x=z} }{ \partial_{d'}^y f\big(\psi(y)\big)\big\vert_{y=z} }
\end{equation}
For simplicity in what follows, we establish the following alternative notations suppressing the dependencies on $f$:
\begin{equation}\label{eq:gen_inner}
     		\Cov{ \partial_{d}^x\!f(x) }{ f(y) } = \ska{ \partial_{d}^x }{ 1_y }\,,
      ~ ~ ~ \Cov{ \partial_{d}^x\!f(x) }{ \partial_{d'}^y\!f(y) } = \ska{ \partial_{d}^x }{ \partial_{d'}^y }\,.
\end{equation}
Similarly, $\Vert 1_x \Vert^2 = \Var\big[ f(x) \big]$ and $\Vert \partial_{d}^x \Vert^2 = \Var\big[ \partial_{d}^xf(x) \big]$.
 \begin{theorem}\label{thm:RiemannianMetric}
    The Riemannian metric on $\overline{\cM}$ induced by a normalized random field is given in local coordinates by
    \begin{equation}
        \bar \Lambda_{dd'}( x ) = \frac{ \ska{ \partial_{d}^x }{ \partial_{d'}^x } }{ \Vert 1_x \Vert }
       - \frac{  \ska{ \partial_{d}^x }{ 1_x } \ska{ \partial_{d'}^x }{ 1_x } }{ \Vert 1_x \Vert^2 }\,.
    \end{equation}
 \end{theorem}
 
The computation of Lipshitz-Killing curvatures requires the shape operator and the
Riemannian curvature. Therefore, the next theorem computes the Christoffel symbols
of the first kind, which can be used to express these quantities in local coordinates.
This is because most fundamental geometric quantities such as the
covariant derivative and the Riemannian curvature are functions of the Christoffel
symbols and their derivatives.
\begin{theorem}\label{thm:Christoffel}
    The Christoffel symbols $\bar\Gamma_{kdd'}$, $k,d,d'\in\{1,\ldots,D\}$ of the first kind
    of the induced Riemannian metric by
    a normalized field are given by
    \begin{equation}\label{eq:ChristoffelI}
        \begin{aligned}
                \bar\Gamma_{kdd'}(x)
                &=   \frac{ \ska{ \partial_{k}^x \partial_{d}^x }{ \partial_{d'}^x } }{ \Vert 1_x \Vert  }
                   - \frac{ \ska{ \partial_{k}^x \partial_{d}^x }{ 1_x } \ska{ \partial_{d'}^x }{ 1_x } }{ \Vert 1_x \Vert^2 }
                   - \frac{ \ska{ \partial_{k}^x }{ \partial_{d'}^x } \ska{ \partial_{d}^x }{ 1_x } }{ \Vert 1_x \Vert^2 }
                   \\
                &\quad - \frac{ \ska{ \partial_{k}^x }{ 1_x } \ska{ \partial_{d}^x }{ \partial_{d'}^x } }{ \Vert 1_x \Vert^2 }
                 + 2\frac{ \ska{ \partial_{k}^x }{ 1_x } \ska{ \partial_{d}^x }{ f_x } \ska{ \partial_{d'}^x }{ 1_x } }{ \Vert 1_x \Vert^3 }\,.
              \end{aligned}
    \end{equation}
\end{theorem}

In terms of Christoffel symbols the covariant derivative $\bar \nabla$ on $\overline{\cM}$ is expressed in the
local chart $\varphi$ by
\begin{equation*}
	\bar \nabla_{\partial_d} \partial_{d'} = \sum_{h=1}^D \sum_{h'=1}^D \bar\Lambda^{hh'}\bar\Gamma_{dd'h'} \partial_h = \bar\Lambda^{-1}(\bar\Gamma_{dd'1}, ..., \bar\Gamma_{dd'D})^T\,.
\end{equation*}
Here $\bar\Lambda^{dd'}$ denotes the $(d, d')$-entry of the inverse of $\bar\Lambda$ in the coordinates $(\cU, \varphi)$.
This formula can be used to extend the covariant derivative to any vector field, i.e.,
\begin{equation}\label{eq:covariantDeriv}
	\bar \nabla_{V} W
	= \sum_{d=1}^D\sum_{d'=1}^D V_d \bar\nabla_{\partial_d} W_{d'}\partial_{d'}
	= \sum_{d=1}^D\sum_{d'=1}^D V_d \Big( \partial_d(W_{d'})\partial_{d'}
						+   W_{d'}\bar\nabla_{\partial_d}\partial_{d'} \Big)\,.
\end{equation}
In particular, if all $\Gamma_{dd'd''} = 0$, $d,d',d''\in \{1,\ldots,D\}$, then
\begin{equation}\label{eq:covariantDerivStat}
	\bar \nabla_{V} W
	= \sum_{d=1}^D\sum_{d'=1}^D V_d \partial_d(W_{d'})\partial_{d'}\,.
\end{equation}
This happens, if $\bar\Lambda(s) = \bar\Lambda(s')$ for all $s,s'\in \overline{\cM}$.

The last geometric quantity required to compute LKCs is the Riemannian curvature tensor $\bar R$.
The curvature tensor has in local coordinates the entries
\begin{equation}\label{eq:CurvatureTensor}
	\bar R_{lkdd'}^\varphi = \partial_l \bar\Gamma_{kdd'} - \partial_k \Gamma_{ldd'}
					+ \sum_{m,n = 1}^D \big( \bar\Gamma_{ldm}\bar\Lambda^{mn}\bar\Gamma_{kd'n} - \bar\Gamma_{kdm}\bar\Lambda^{mn}\bar\Gamma_{ld'n}  \big)
\end{equation}
for $i,j,k,l\in\{1,\ldots, D\}$, compare \cite[eq. (7.7.4)]{Adler2007}.
By now, we derived almost all quantities to state the Riemannian curvature tensor in local coordinates.
The missing quantities are the derivatives of the Christoffel symbols which can be found
in the next Lemma.
\begin{lemma}\label{lem:DiffChristoffel}
	The difference of the derivatives of Christoffel symbols in the Riemannian curvature tensor \eqref{eq:CurvatureTensor}
	of a normalized field can be expressed as:
	\begin{equation*}
	\begin{split}
		\partial_l &\bar\Gamma_{kdd'}(x) - \partial_k \bar\Gamma_{ldd'}(x)\\
	&= \frac{ \ska{ \partial_{k}^x \partial_{d}^x }{ \partial_{l}^x \partial_{d'}^x } - \ska{ \partial_{l}^x \partial_{d}^x }{ \partial_{k}^x \partial_{d'}^x } }{\Vert 1_x \Vert}\\
	   &~ ~ ~ ~ ~- \Vert 1_x \Vert^{-2}\Big(~ ~ ~
	   					\ska{ \partial_{l}^x }{ 1_x }\ska{ \partial_{k}^x \partial_{d}^x }{ \partial_{d'}^x }
	                  - \ska{ \partial_{k}^x }{ 1_x }\ska{ \partial_{l}^x \partial_{d}^x }{ \partial_{d'}^x }\\
&~ ~ ~ ~ ~ ~ ~ ~ ~ ~ ~ ~ ~ ~ ~ ~ ~ ~ ~ ~~ ~ + \ska{ \partial_{k}^x \partial_{d}^x }{ \partial_{l}^x }\ska{ \partial_{d'}^x }{ 1_x }
	                  						- \ska{ \partial_{l}^x \partial_{d}^x }{ \partial_{k}^x }\ska{ \partial_{d'}^x }{ 1_x }\\
&~ ~ ~ ~ ~ ~ ~ ~ ~ ~ ~ ~ ~ ~ ~ ~ ~ ~ ~ ~~ ~ + \ska{ \partial_{k}^x \partial_{d}^x }{ 1_x }\ska{ \partial_{l}^x\partial_{d'}^x }{ 1_x }
	                  						- \ska{ \partial_{l}^x \partial_{d}^x }{ 1_x }\ska{ \partial_{k}^x\partial_{d'}^x }{ 1_x }\\
&~ ~ ~ ~ ~ ~ ~ ~ ~ ~ ~ ~ ~ ~ ~ ~ ~ ~ ~ ~~ ~ + \ska{ \partial_{k}^x }{ 1_x }\ska{ \partial_{l}^x\partial_{d'}^x }{ \partial_{d}^x }
	                  						- \ska{ \partial_{l}^x }{ 1_x }\ska{ \partial_{k}^x\partial_{d'}^x }{ \partial_{d}^x  }\\
&~ ~ ~ ~ ~ ~ ~ ~ ~ ~ ~ ~ ~ ~ ~ ~ ~ ~ ~ ~~ ~ + \ska{ \partial_{l}^x }{ \partial_{d}^x }\ska{ \partial_{k}^x }{ \partial_{d'}^x }
	                  						- \ska{ \partial_{k}^x }{ \partial_{d}^x }\ska{ \partial_{l}^x }{ \partial_{d'}^x  }\\
&~ ~ ~ ~ ~ ~ ~ ~ ~ ~ ~ ~ ~ ~ ~ ~ ~ ~ ~ ~~ ~ + \ska{ \partial_{l}^x }{ \partial_{d}^x }\ska{ \partial_{k}^x }{ \partial_{d'}^x }
	                  						- \ska{ \partial_{k}^x }{ \partial_{d}^x }\ska{ \partial_{l}^x }{ \partial_{d'}^x  }\\
&~ ~ ~ ~ ~ ~ ~ ~ ~ ~ ~ ~ ~ ~ ~ ~ ~ ~ ~ ~~ ~ + \ska{ \partial_{k}^x }{ \partial_{l}^x\partial_{d'}^x }\ska{ \partial_{d}^x }{ 1_x }
	                  						- \ska{ \partial_{l}^x }{ \partial_{k}^x\partial_{d'}^x }\ska{ \partial_{d}^x }{ 1_x }
	                     \Big)\\
	   &~ ~ ~ ~ ~+ 2\Vert 1_x \Vert^{-3}\Big(~ ~ ~
	   					\ska{ \partial_{l}^x }{ 1_x }\ska{ \partial_{k}^x\partial_{d}^x }{ 1_x }\ska{ \partial_{d'}^x }{ 1_x }
	                  - \ska{ \partial_{k}^x }{ 1_x }\ska{ \partial_{l}^x\partial_{d}^x }{ 1_x }\ska{ \partial_{d'}^x }{ 1_x }\\
&~ ~ ~ ~ ~ ~ ~ ~ ~ ~ ~ ~ ~ ~ ~ ~ ~ ~ ~ ~~ ~
					  + \ska{ \partial_{k}^x }{ 1_x }\ska{ \partial_{d}^x }{ \partial_{l}^x }\ska{ \partial_{d'}^x }{ 1_x }
	                  - \ska{ \partial_{l}^x }{ 1_x }\ska{ \partial_{d}^x }{ \partial_{k}^x }\ska{ \partial_{d'}^x }{ 1_x }\\
&~ ~ ~ ~ ~ ~ ~ ~ ~ ~ ~ ~ ~ ~ ~ ~ ~ ~ ~ ~~ ~
					  + \ska{ \partial_{k}^x }{ 1_x }\ska{ \partial_{d}^x }{ 1_x }\ska{ \partial_{l}^x\partial_{d'}^x }{ 1_x }
	                  - \ska{ \partial_{l}^x }{ 1_x }\ska{ \partial_{d}^x }{  }\ska{ \partial_{k}^x\partial_{d'}^x }{ 1_x }\\
&~ ~ ~ ~ ~ ~ ~ ~ ~ ~ ~ ~ ~ ~ ~ ~ ~ ~ ~ ~~ ~
					  + \ska{ \partial_{l}^x }{ 1_x }\ska{ \partial_{d}^x }{ 1_x }\ska{ \partial_{k}^x }{ \partial_{d'}^x }
	                  - \ska{ \partial_{k}^x }{ 1_x }\ska{ \partial_{d}^x }{ 1_x }\ska{ \partial_{l}^x }{ \partial_{d'}^x }
	    \Big)
	\end{split}
	\end{equation*}
\end{lemma}
\begin{proof}
	Simple, but lengthy computation.
\end{proof}

To compute the first LKC $\cL_1$ for a $3$-dimensional manifold the trace of the Riemannian tensor is needed
which in the coordinates $(\cU,\varphi)$ can be expressed in terms of the entries of the inverse of the square
root of the Riemannian metric and the Riemannian tensor as
	\begin{equation}\label{eq:trace_Riemann}
	\begin{split}
		{\rm tr}\big( \bar R\, \big) &= \sum_{i,j,k,l=1}^3 \bar R^\varphi_{ijkl}
			\Bigg(~
				\frac{ \bar\Lambda_{1i}^{-1/2}\bar\Lambda_{1j}^{-1/2}\bar\Lambda_{1k}^{-1/2}\bar\Lambda_{1l}^{-1/2}}{ 2 }\\
				&~~~~~~~~~~~~~~~~~~~~~~~+ \frac{ \bar\Lambda_{2i}^{-1/2}\bar\Lambda_{2j}^{-1/2}\bar\Lambda_{2k}^{-1/2}\bar\Lambda_{2l}^{-1/2}}{ 2 }\\
			  &~~~~~~~~~~~~~~~~~~~~~~~+ \frac{\bar\Lambda_{3i}^{-1/2}\bar\Lambda_{3j}^{-1/2}\bar\Lambda_{3k}^{-1/2}\bar\Lambda_{3l}^{-1/2}}{ 2 }\\
			  &~~~~~~~~~~~~~~~~~~~~~~~+ \bar\Lambda_{1i}^{-1/2}\bar\Lambda_{2j}^{-1/2}\bar\Lambda_{1k}^{-1/2}\bar\Lambda_{2l}^{-1/2}\\
			  &~~~~~~~~~~~~~~~~~~~~~~~+ \bar\Lambda_{1i}^{-1/2}\bar\Lambda_{3j}^{-1/2}\bar\Lambda_{1k}^{-1/2}\bar\Lambda_{3l}^{-1/2}\\
			  &~~~~~~~~~~~~~~~~~~~~~~~+\bar\Lambda_{2i}^{-1/2}\bar\Lambda_{3j}^{-1/2}\bar\Lambda_{2k}^{-1/2}\bar\Lambda_{3l}^{-1/2}
			~\Bigg)\,.
	\end{split}
	\end{equation}

\subsection{ Definition of LKCs }\label{A:LKCdefn}
Lipshits Killing Curvatures $\cL_1,\ldots,\cL_D$ are the intrinsic volumes of a compact $D$-dimensional
Whitney stratified (WS) manifold $\big( \cM, \Lambda \big)$ isometrically embedded into
$ \big( \overline{\cM}, \bar\Lambda \big)$.
Here $\Lambda$ and $\bar\Lambda$ denote Riemannian metrics of $\cM$  and $\overline{\cM}$. They are related by
$\bar\Lambda\circ\iota = \Lambda$ where $\iota:\,\cM\rightarrow \overline{\cM}$ is the embedding and
hence $\cM$ is isometrically embedded into $\overline{\cM}$.
In this section we make the formula for the LKCs \cite[Definition 10.7.2]{Adler2007}
of an at most $3$-dimensional
WS manifold more explicit. From the definition of the LKCs it is easy to deduce
\begin{equation}
	\cL_{D-1} = \vol_{D-1}\big( \partial_{D-1}\cM \big) ~ ~ ~ \text{ and } ~ ~ ~ \cL_D = \vol_D\big( \partial_D\cM \big)\,.
\end{equation}
Here $\partial_{d}\cM$  denotes the $d$-dimensional stratum of $\cM$ and
the volume is the volume of the Riemannian manifold $(\partial_d\cM, \Lambda\vert_{\partial_d\cM})$
where $\Lambda\vert_{\partial_d\cM}$ is the restriction of $\Lambda$ to $\partial_d\cM$.
 
For a $3$-dimensional manifold it remains to compute $\cL_1$, which we derive
from \cite[Definition 10.7.2]{Adler2007} and some geometric computations in the next proposition. 
\begin{proposition}\label{prop:L1general}
    Let $(\overline{\cM},\bar\Lambda)$ be a closed Riemannian manifold of dimension $3$ and
    $\cM \subset \overline{\cM}$ be a compact WS manifold of dimension $3$
    isometrically embedded into $\overline{\cM}$. Then
    \begin{equation}\label{eq:L1general}
    \begin{split}
            \cL_1 &= \frac{1}{2\pi} \int_{\partial_{1} \cM}\int_{\mathbb{S}\big(\cT_s\partial_{1}\cM^\bot\big)}  \alpha(v)\mathcal{H}_{1}(dv)\mathcal{H}_{1}(ds)\\
            &\quad\quad+
            \frac{1}{2\pi} \int_{\partial_{2} \cM} \bar\Lambda\Big(\bar \nabla_{U_s} U_s+
            \bar \nabla_{V_s}V_s, N_s\Big) \mathcal{H}_{2}(ds)\\
            &\quad\quad-
            \frac{1}{2\pi} \int_{\partial_{3} \cM} {\rm Tr}^{\cT_s\partial_{3}\cM}( \bar R ) \mathcal{H}_{3}(ds)
    \end{split}
    \end{equation}
    Here  $\mathcal{H}_{1}(dv)$ is the volume
    form induced on the sphere
    $$
    \mathbb{S}\big(\cT_s\partial_{1}\cM^\bot\big) = \left\{ v \in \cT_s\overline{\cM} ~\vert~ \bar\Lambda(v,v) = 1
    		~\wedge~ \bar\Lambda(v,w) = 0\text{ for }w\in \cT_s\partial_{1}\cM \right\}
    $$
    by $\bar\Lambda$ and $\mathcal{H}_{d}(ds)$ the volume form of $\partial_d\cM$.
    Moreover, $\alpha(v)$ denotes the normal
    Morse index given in \cite[Scn 9.2.1]{Adler2007} and $U, V, N$ is a piecewise
    differentiable vector field on $\partial_2\cM$ such that
    $U_s,V_s$ form an orthonormal basis for $\cT_s\partial_2\cM$ for all $s\in \partial_2\cM$ and
    $N$ is outward pointing normal vector field.
\end{proposition}
\begin{remark}
	In the special case that the metric $\bar\Lambda$ is constant, it holds that
	\begin{equation*}
	\begin{split}
            \cL_1 &= \frac{1}{2\pi} \int_{\partial_{1} \cM}\int_{\mathbb{S}\big(\cT_s\partial_{1}\cM^\bot\big)}  \alpha(v)\mathcal{H}_{1}(dv)\mathcal{H}_{1}(ds)\\
            &\quad+
            \frac{1}{2\pi} \int_{\partial_{2} \cM} \bar\Lambda\Big(\bar \nabla_{U_s} U_s+
            \bar \nabla_{V_s}V_s, N_s\Big) \mathcal{H}_{2}(ds)\,,
    \end{split}
	\end{equation*}	 
	since the curvature tensor $\bar R$ vanishes. If $\cM = \cM_\cV$ is a voxel manifold than even
	\begin{equation*}
            \cL_1 = \frac{1}{2\pi} \int_{\partial_{1} \cM}\int_{\mathbb{S}\big(\cT_s\partial_{1}\cM^\bot\big)}  \alpha(v)\mathcal{H}_{1}(dv)\mathcal{H}_{1}(ds)\,,
	\end{equation*}
	as $\nabla_{U_s} U_s = \nabla_{V_s} V_s = 0$ on $\partial_{2} \cM$
	by \eqref{eq:covariantDerivStat} and $\bar\Lambda$ being constant.
\end{remark}

Using the above Theorem we can now derive an expression for $\cL_1$ of a $3$-dimensional voxel manifold.
\begin{theorem}\label{thm:L1VM}
	Let $\cM_\cV$ be a $3$-dimensional voxel manifold. Then
	\begin{equation}\label{eq:L1Surf}
	\begin{split}
	\cL_1 &= \frac{1}{2\pi} \sum_{ \vert I \vert = 1 }			
				\int_{\cF_I}\Theta(x)\sqrt{ \det\big( \boldsymbol{\Lambda}^{I}(x) \big) } dx^I\\
	&\quad+
	\frac{1}{2\pi} \sum_{ \vert I \vert = 2 } \int_{ \cF_I } \Bigg[
		\left( U_{I_1}(x)^2 + V_{I_1}(x)^2 \right) N^T(x) \begin{pmatrix} \Gamma_{{I_1}{I_1}1}(x)\\ \Gamma_{{I_1}{I_1}2}(x)\\ \Gamma_{{I_1}{I_1}3}(x)\end{pmatrix}\\
	&\quad\quad\quad\quad\quad+ \sum_{k=1}^2V_{I_k}(x)V_{I_2}(x) N^T(x)
	\begin{pmatrix} \Gamma_{{I_k}{I_2}1}(x)\\ \Gamma_{{I_k}{I_2}2}(x)\\ \Gamma_{{I_k}{I_2}3}(x)\end{pmatrix} \Bigg]\sqrt{  \det\big( \boldsymbol{\Lambda}^{I}(x) \big) }dx^I\\
	&\quad-
	\frac{1}{2\pi} \sum_{ v \in \cV } \int_{ \cB_v( \delta ) }
	{\rm Tr}\big( R(x) \big) \sqrt{ \det\big( \boldsymbol{\Lambda}(x) \big) }dx
	\end{split}
	\end{equation}
	Here ${\rm Tr}\big( R(x) \big)$ is the trace of the Riemannian curvature tensor and
	\begin{equation*}
			\Theta(x) = \begin{cases}
					\pi - \beta(x)\,, & \text{if $x$ belongs to a convex edge}\\
					-2\beta(x)\,, & \text{if $x$ belongs to a double convex edge}\\
					\beta(x) - \pi\,, & \text{if $x$ belongs to a convex edge}
			\end{cases}\,,
	\end{equation*}
	compare Appendix \ref{app:InducedRiemann} and especially \eqref{eq:trace_Riemann}.
	Here $\beta(x)$ is defined using the crossproduct $V_x \times N_x = \big( m_1^I(x), m_2^I(x), m_3^I(x) \big)$
	of the elements of the ONB from \eqref{prop:ONF} with $k = I$ by
	\begin{equation*}
		\beta(x) = \arccos\left( \frac{ m_2(x)m_3(x) }{ \sqrt{ m_2^2(x) + m_1^2(x) }\sqrt{ m_3^2(x) + m_1^2(x) } } \right).
	\end{equation*}
	The different types of edges are visualized in Figure \ref{fig:EdgeVisualization}.
\end{theorem}

\begin{figure}[ht]
\begin{center}
\includegraphics[trim=0 0 0 0,clip,width=2.1in]{\figurepath 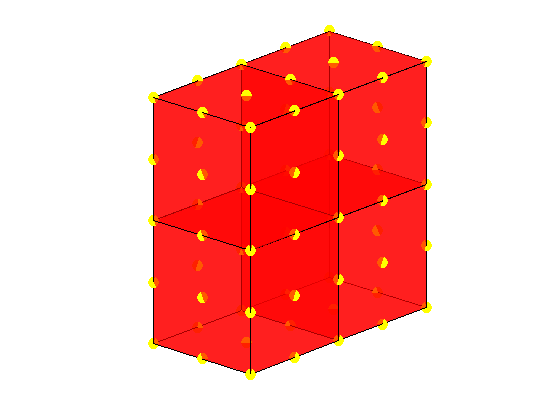}
\includegraphics[trim=0 0 0 0,clip,width=2.1in]{\figurepath 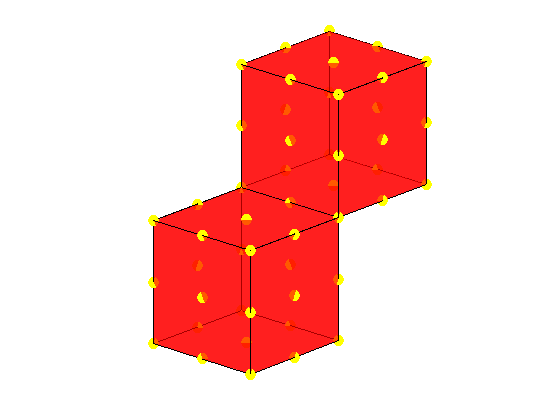}
\includegraphics[trim=0 0 0 0,clip,width=2.1in]{\figurepath 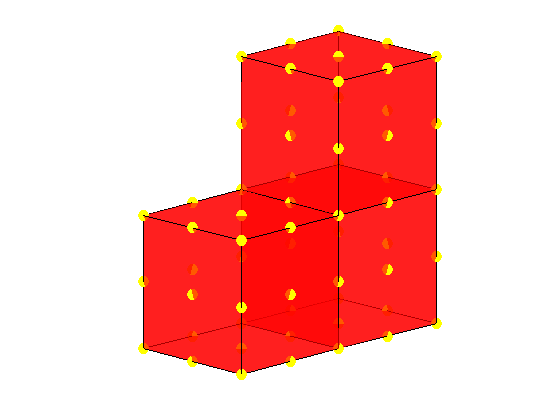}
\end{center}
	\caption{Visualization of the different types of edges appearing in a voxel manifold.
			 In the left voxel manifold all edges are convex. The edge where the
			 two cubes of the voxel manifold in the center are touching is a double
			 convex edge and the same edge in the voxel manifold on the right is a
			 concave edge since a third cube is added.
	\label{fig:EdgeVisualization}}
\end{figure}

\section{Proofs of Results in the Appendix}\label{app:proofsApp}
\subsection{Proof of Theorem \ref{thm:RiemannianMetric}}
\begin{proof}
	Interchanging expectation and derivatives yields
	\begin{equation}\label{eq:derivation_riem_metric}
	\begin{aligned}
	\bar\Lambda_{dd'}(z)
	&= \partial_{d}^x\partial_{d'}^y \frac{ \Cov{ f_x }{ f_y } }{ \Vert 1_x \Vert\Vert 1_y \Vert }\Bigg\vert_{(x,y)=(z,z)}\\
	&=  \frac{ \ska{ \partial_{d}^x }{ \partial_{d'}^y } }{ \Vert 1_x \Vert  \Vert 1_y \Vert}
	- \frac{ \ska{ \partial_{d}^x }{ 1_y } \ska{ 1_y }{ \partial_{d'}^y } }{\Vert 1_x \Vert \Vert 1_y \Vert^3}\\
	&\quad -\frac{ \ska{ \partial_{d}^x }{ 1_x } \ska{ 1_x }{ \partial_{d'}^y }  }{\Vert 1_x \Vert^3 \Vert 1_y \Vert} + \frac{ \ska{ 1_x }{ 1_y } \ska{ \partial_{d}^x }{ 1_x } \ska{ 1_y }{ \partial_{d'}^y }}{\Vert 1_x \Vert^3 \Vert 1_y \Vert^3}\Bigg\vert_{(x,y)=(z,z)}\\
	&= \frac{  \ska{ \partial_{d}^z }{  \partial_{d'}^z } }{ \Vert 1_z \Vert^2 }
	- \frac{ \ska{ \partial_{d}^z }{ 1_z } \ska{ 1_z }{ \partial_{d'}^z } }{\Vert 1_z \Vert^4}
	\end{aligned}
	\end{equation}
\end{proof}

\subsection{Proof of Theorem \ref{thm:Christoffel}}
\begin{proof}
	From the definition of the Christoffel symbols one can derive that
	\begin{equation}
	\Gamma_{kdd'}(z) =
	\frac{\partial^3 \tilde\fc(x,y)}{\partial_{x_{k}}\partial_{x_{d}}\partial_{y_{d'}}}\Bigg\vert_{(x,y)=(z,z)}\,,
	\end{equation}
	see also \cite[eq. 12.2.17]{Adler2007}. Here $\tilde\fc(x,y) =
	\frac{  \Cov{f_x}{f_y}  }{ \sqrt{\Cov{f_x}{f_x}} \sqrt{\Cov{f_y}{f_y}} }$.
	Thus, by simply taking another derivate of \eqref{eq:derivation_riem_metric},
	we obtain
	\begin{align*}
	\Gamma_{kdd'}(z) &=
	\partial_k^x \partial_{d}^x \partial_{d'}^y \frac{ \Cov{ f_x }{ f_y } }{ \sqrt{ \Vert 1_x \Vert\Vert 1_y \Vert} }\Bigg\vert_{(x,y)=(z,z)}\\
	&=	  \frac{ \ska{ \partial_k^x \partial_{d}^x }{ \partial_{d'}^x } }{ \Vert 1_x \Vert^2 }
	- \frac{ \ska{ \partial_{k}^x }{ 1_x } \ska{ \partial_{d}^x }{ \partial_{d'}^x } }{ \Vert 1_x \Vert^4 }
	- \partial_{k}^x \frac{ \ska{  \partial_{d}^x }{ 1_y } \ska{ 1_y }{  \partial_{d'}^y } }{\Vert 1_x \Vert \Vert 1_y \Vert^3}\\
	&\quad -\partial_{k}^x \frac{ \ska{  \partial_{d}^x }{ 1_x } \ska{ 1_x }{  \partial_{d'}^y }  }{\Vert 1_x \Vert^3 \Vert 1_y \Vert}
	+ \partial_{k}^x \frac{ \ska{ 1_x }{ 1_y } \ska{ \partial_{d}^x }{ 1_x } \ska{ 1_y }{ \partial_{d'}^y }}{\Vert 1_x \Vert^3 \Vert 1_y \Vert^3}\Bigg\vert_{(x,y)=(z,z)}\\
	&=    \frac{ \ska{ \partial_k^z \partial_{d}^z }{ \partial_{d'}^z } }{ \Vert 1_z \Vert^2}
	- \frac{ \ska{ \partial_{k}^z }{ 1_z } \ska{ \partial_{d}^z }{ \partial_{d'}^z } }{ \Vert 1_z \Vert^4 }
	- \frac{ \ska{ \partial_{k}^z\partial_{d}^z }{ 1_z } \ska{ 1_z }{ \partial_{d'}^z } }{\Vert 1_z \Vert^4} \\
	&\quad 
	+ \frac{ \ska{ \partial_{k}^z }{ 1_z } \ska{ \partial_{d}^z }{ 1_z } \ska{ 1_z }{ \partial_{d'}^z } }{\Vert 1_z \Vert^6}
	- \frac{ \ska{ \partial_{k}^z \partial_{d}^z }{ 1_z } \ska{ 1_z }{  \partial_{d'}^z } + \ska{  \partial_{d}^z }{ \partial_{k}^z } \ska{ 1_z }{  \partial_{d'}^z } }{\Vert 1_z \Vert^4 }\\
	&\quad - \frac{\ska{  \partial_{d}^z }{ 1_z } \ska{ \partial_{k}^z }{  \partial_{d'}^z } }{\Vert 1_z \Vert^4} + 3\frac{ \ska{  \partial_{k}^z }{ 1_z } \ska{  \partial_{d}^z }{ 1_z } \ska{ 1_z }{  \partial_{d'}^z }  }{\Vert 1_z \Vert^6} \\
	&\quad + \frac{\ska{ \partial_{k}^z }{ 1_z } \ska{  \partial_{d}^z }{ 1_z } \ska{ 1_z }{  \partial_{d'}^z }  }{ \Vert 1_z \Vert^6}
	+ \frac{ \ska{ \partial_{k}^z\partial_{d}^z }{ 1_z } \ska{ 1_z }{ \partial_{d'}^z }}{ \Vert 1_z \Vert^4 }
	+ \frac{ \ska{ \partial_{d}^z }{ \partial_{k}^z } \ska{ 1_z }{ \partial_{d'}^z }}{ \Vert 1_z \Vert^4 }\\
	&\quad -3\frac{\ska{ \partial_{k}^z}{ 1_z } \ska{ \partial_{d}^z}{ 1_z } \ska{ 1_z }{\partial_{d'}^z}}{\Vert 1_z \Vert^6}\\
	&=    \frac{ \ska{ \partial_k^z \partial_{d}^z }{ \partial_{d'}^z } }{ \Vert 1_z \Vert^2}
	- \frac{ \ska{ \partial_{k}^z }{ 1_z } \ska{ \partial_{d}^z }{ \partial_{d'}^z } }{ \Vert 1_z \Vert^4 }
	- \frac{ \ska{ \partial_{k}^z\partial_{d}^z }{ 1_z } \ska{ 1_z }{ \partial_{d'}^z } }{\Vert 1_z \Vert^4} \\
	&\quad - \frac{\ska{  \partial_{d}^z }{ 1_z } \ska{ \partial_{k}^z }{  \partial_{d'}^z } }{\Vert 1_z \Vert^4}
	+ 2\frac{\ska{ \partial_{k}^z}{ 1_z } \ska{ \partial_{d}^z}{ 1_z } \ska{ 1_z }{\partial_{d'}^z}}{\Vert 1_z \Vert^6}\\
	\end{align*}
\end{proof}

\subsection{Proof of Proposition \ref{prop:L1general}}
\begin{proof}
Using $D-d' = D'$ we have that the LKCs of a WS manifold $\cM$ are defined by
\begin{align}\label{eq:LKCgeneral}
\begin{split}
	 \cL_d =& \sum_{d'=d}^D \tfrac{1}{(2\pi)^{\frac{d'-d}{2}}} \sum_{l=0}^{\lfloor \frac{d'-d}{2} \rfloor}  \tfrac{(-1)^lC(D', d-d'-2l)}{l!(d-d'-2l)!}\\
	&\times \int_{\partial_{d'} \cM}\int_{\mathbb{S}(\cT_s\partial_{d'}\cM^\bot)} \!\!\!\!\!\!\!\!\!\!\!\!\!{\rm Tr}^{\cT_s\partial_{d'}\cM}\Big( R^lS^{d'-d-2l}_{\nu_{D'}} \Big) \alpha(\nu_{D'})\mathcal{H}_{-1}(d\nu_{D'})\mathcal{H}_{d'}(ds)\,,
\end{split}
\end{align}
compare \cite[Definition 10.7.2]{Adler2007}.
This formula requires further explanations. The constant $C(m,i)$ is defined in \cite[eq. (10.5.1), p.233]{Adler2007}, i.e.,
\begin{equation}
	C(m,i) = \begin{cases}
				\frac{(2\pi)^{\frac{i}{2}}}{s_{m+i}}, \quad &m+i>0,\\
				1,\quad &m=0
			  \end{cases}\, ~~\text{with }~~ s_m =\frac{ 2\pi^{m/2}}{\Gamma(m/2)}\,,
\end{equation}
which implies $C(m,0) = \Gamma(m/2) / 2 / \pi^{m/2} $.
Moreover, $\mathcal{H}_{D-d-1}$ is the volume form on $\mS(\cT_s\partial_{d}\cM^\bot)$ and $\mathcal{H}_{d}$ the volume form of $\partial_{d}\cM$. $R$ denotes the Riemannian curvature tensor of the different strata on $\cM$ depending
on the strata the integral integrates over. In particular, note that $\bar R = R$ for $\partial_3\cM$.

From this $\cL_1$ simplifies to
\begin{align*}
	 \cL_1 &= C(3, 0) \int_{\partial_{1} \cM}\int_{S(\cT_s\partial_{1}\cM^\bot)} {\rm Tr}^{\cT_s\partial_{1}\cM}\Big( R^0S^{0}_{\nu_{2}} \Big) \alpha(\nu_{2})\mathcal{H}_{1}(d\nu_{2})\mathcal{H}_{1}(ds)\\
	 &~~~~+ \frac{C(1, 1)}{\sqrt{2\pi}}  \int_{\partial_{2} \cM}\int_{S(\cT_s\partial_{2}\cM^\bot)} {\rm Tr}^{\cT_s\partial_{2}\cM}\Big( R^0S^{1}_{\nu_{1}} \Big) \alpha(\nu_{1})\mathcal{H}_{0}(d\nu_{1})\mathcal{H}_{2}(ds)\\
		&~~~~+ \frac{C(0, 2)}{4\pi} \int_{\partial_{3} \cM}\int_{S(\mathbb{O})} {\rm Tr}^{\cT_s\partial_{3}\cM}\Big( R^0S^{2}_{\nu_{0}} \Big) \alpha(\nu_{0})\mathcal{H}_{-1}(d\nu_{0})\mathcal{H}_{3}(ds)\\
		&~~~~- \frac{C(0, 0)}{2\pi} \int_{\partial_{3} \cM}\int_{S(\mathbb{O})} {\rm Tr}^{\cT_s\partial_{3}\cM}\Big( R^1S^{0}_{\nu_{0}} \Big) \alpha(\nu_{0})\mathcal{H}_{-1}(d\nu_{0})\mathcal{H}_{3}(ds)\\
		&= \frac{1}{2\pi} \int_{\partial_{1} \cM}\int_{S(\cT_s\partial_{1}\cM^\bot)} \alpha(\nu_{2})\mathcal{H}_{1}(d\nu_{2})\mathcal{H}_{1}(ds)\\
	 &~~~~+ \frac{1}{2\pi} \int_{\partial_{2} \cM}\int_{S(\cT_s\partial_{2}\cM^\bot)} {\rm Tr}^{\cT_s\partial_{2}\cM}\Big( 1\cdot S^{1}_{\nu_{1}} \Big) \alpha(\nu_{1})\mathcal{H}_{0}(d\nu_{1})\mathcal{H}_{2}(ds)\\
		&~~~~+ \frac{1}{4\pi} \int_{\partial_{3} \cM} {\rm Tr}^{\cT_s\partial_{3}\cM}\Big( 1\cdot S^{2}_{\mathbb{O}} \Big) \mathcal{H}_{3}(ds)\\
		&~~~~- \frac{1}{2\pi} \int_{\partial_{3} \cM} {\rm Tr}^{\cT_s\partial_{3}\cM}\Big( R^1S^{0}_{\mathbb{O}} \Big) \mathcal{H}_{3}(ds)
\end{align*}
This can be further simplified as follows:
\begin{align*}
		\cL_1 &=  \frac{1}{2\pi} \int_{\partial_{1} \cM}\int_{S(\cT_s\partial_{1}\cM^\bot)}  \alpha(\nu_{2})\mathcal{H}_{1}(d\nu_{2})\mathcal{H}_{1}(ds)\\
	 &~~~~+ \frac{1}{2\pi} \int_{\partial_{2} \cM}\int_{S(\cT_s\partial_{2}\cM^\bot)} {\rm Tr}^{\cT_s\partial_{2}\cM}\Big( S^{1}_{\nu_{1}} \Big) \alpha(\nu_{1})\mathcal{H}_{0}(d\nu_{1})\mathcal{H}_{2}(ds)\\
		&~~~~- \frac{1}{2\pi} \int_{\partial_{3} \cM} {\rm Tr}^{\cT_s\partial_{3}\cM}( \bar R ) \mathcal{H}_{3}(ds)\\
	&= \frac{1}{2\pi} \int_{\partial_{1} \cM}\int_{S(\cT_s\partial_{1}\cM^\bot)}  \alpha(\nu_{2})\mathcal{H}_{1}(d\nu_{2})\mathcal{H}_{1}(ds)\\
	 &~~~~+ \frac{1}{2\pi} \int_{\partial_{2} \cM} g\big(\bar \nabla_{e_1(s)} e_1(s), \tilde \nu(s)\big) + g\big(\bar \nabla_{e_2(s)}e_2(s), \tilde \nu(s)\big) \mathcal{H}_{2}(ds)\\
		&~~~~- \frac{1}{2\pi} \int_{\partial_{3} \cM} {\rm Tr}^{\cT_s\partial_{3}\cM}( \bar R ) \mathcal{H}_{3}(ds)\\
	&= \frac{1}{2\pi} \int_{\partial_{1} \cM}\int_{S(\cT_s\partial_{1}\cM^\bot)}  \alpha(\nu_{2})\mathcal{H}_{1}(d\nu_{2})\mathcal{H}_{1}(ds)\\
	 &~~~~+ \frac{1}{2\pi} \int_{\partial_{2} \cM} g\big(\bar \nabla_{e_1(s)} e_1(s) + \bar \nabla_{e_2(s)}e_2(s), \tilde \nu(s)\big) \mathcal{H}_{2}(ds)\\
		&~~~~- \frac{1}{2\pi} \int_{\partial_{3} \cM} {\rm Tr}^{\cT_s\partial_{3}\cM}( \bar R ) \mathcal{H}_{3}(ds)
\end{align*}
Here $\tilde\nu(s)$ is the inward pointing normal at $x$ in $\partial_2\cM$ and $e_1(s), e_2(s)$ an orthonormal basis of $\cT_s\partial_2\cM$ and used Remark (10.5.2) \cite[p.233]{Adler2007}, i.e.,
\begin{equation}
	S^j_\mathbb{O} = \begin{cases}
				1, \quad &j=0,\\
				0,\quad &\text{otherwise}
			  \end{cases}\,.
\end{equation}
\end{proof}

\section{Proofs of the Results in the Main Manuscript}\label{app:ProofsMain}
\subsection{Proof of Proposition \ref{prop:HoelderCont}}
\begin{proof}
	Define $q = p / (p-1)$ if $p>1$ and $q = \infty $ if $p=1$. Using the triangle inequality and
	H\"older's inequality yields for the charts $(\overline{U}_\alpha,\overline{\varphi}_\alpha)$,
	$\alpha \in \{1,\ldots,P\}$, in the atlas of $\overline{\cM}$ covering $\cM$,
	all $x,y \in \overline{\varphi}\big( \overline{U}_\alpha \big) \cap \cM$ that
 	\begin{align*}
    	\vert \tilde X_\alpha( x ) - \tilde X_\alpha( y ) \vert
    		&= \left \vert \sum_{ v \in \cV } \Big( K\big( \overline{\varphi}^{-1}_\alpha(x), v \big) - K\big( \overline{\varphi}^{-1}_\alpha(y), v \big) \Big) X(v) \right\vert \\
    		&\leq \sqrt[q]{ \sum_{ v \in \cV } \Big\vert  K\big( \overline{\varphi}^{-1}_\alpha(x), v \big) - K\big( \overline{\varphi}^{-1}_\alpha(y), v \big) \Big\vert^q } \sqrt[p]{ \sum_{ v \in \cV } X(v)^p }\\
    		&\leq \sqrt[q]{ \sum_{ v \in \cV } A^q } \sqrt[p]{ \sum_{ v \in \cV} X(v)^p } \big\Vert x - y \big\Vert^\gamma\\
    		&\leq L \big\Vert x - y \big\Vert^\gamma\,.
	\end{align*}
	Here $A$ bounds the H\"older constants of $K\big( \overline{\varphi}^{-1}_\alpha(\cdot), v \big)$,
	for all $\alpha \in \{1,\ldots, P\}$ and all $ v \in \cV $ from above,
	and $ L = |\mathcal{V}|A\sqrt[p]{ \sum_{ v \in \cV} X(v)^p } $.
	If $p=1$ then the statement with the $q$-th root is
 	the maximum over $\cV$ instead of the $q$-norm. The result follows as by assumption
 	$\mathbb{E}\left[L^p \right]$ is finite.
\end{proof}

\subsection{Proof of Proposition \ref{prop:GaussG3}}
\begin{proof}
	The functions $K\big(  \overline{\varphi}^{-1}_\alpha(\cdot), v \big)$
	are Lipshitz continuous for each $v\in\cV$ since they are $\cC^1$
	and $\overline{\cM}$ is compact. Because $\cV$ is finite and $\alpha \in \{1,\ldots, P\}$,
	there exists an $M > 0$ that bounds all the Lipschitz constants of the functions
	$K\big(  \overline{\varphi}^{-1}_\alpha(\cdot), v \big)$.
	Thus, applying Proposition \ref{prop:HoelderCont}
	with $ p = 2 $ and $ \gamma = 1 $ yields for the charts
	$(\overline{U}_\alpha,\overline{\varphi}_\alpha)$, $\alpha \in \{1,\ldots,P\}$,
	in the atlas of $\overline{\cM}$ covering $\cM$, all
	$x,y \in \overline{\varphi}\big( \overline{U}_\alpha \big) \cap \cM$ such that
	$ 0 < \Vert x - y \Vert < 1$ that
 \begin{align*}
    \mathbb{E}\left[ \left( \tilde X_\alpha(x) - \tilde X_\alpha(y) \right)^2 \right]
    				\leq \mathbb{E}\left[ L^2 \Vert x - y \Vert^2 \right]
    				=    \mathbb{E}\left[ L^2 \right] \Vert x - y \Vert^2\,.
 \end{align*}
 The claim follows since $ x^2 \leq \left(\log |x|\right)^{-2} $ for $ 0 < x < 1 $.
\end{proof}

\subsection{Proof of Proposition \ref{prop:nondegen}}
\begin{proof}
	As the property is local, we can w.l.o.g.
	assume that $\cM$ is a compact domain in $\mathbb{R}^D$
	and drop the chart notation for simplicity.
	Given $ x \in \cM $, suppose that there exist sets of real
	constants $ a, a_i, a_{jk}, c$ ($ 1 \leq i \leq D, 1 \leq j \leq k \leq D $)
	such that 
	\begin{equation*}
		a\tilde X(s) + \sum_{i = 1}^D a_{i} \tilde X_i(s) + \sum_{ 1 \leq j \leq k \leq D} a_{jk} \tilde X_{jk}(s) = c,
	\end{equation*}
	which implies that 
	\begin{equation*}
	\begin{split}
		a \sum_{v \in \mathcal{V}_s} K(s, v)X(v)
			&+  \sum_{i = 1}^D a_{i} \sum_{v \in \mathcal{V}_s} \partial_{i}^sK(s,v)X(v)\\
			&+ \sum_{ 1 \leq j \leq k \leq D} a_{jk}  \sum_{v \in \mathcal{V}_s} \partial_{jk}^sK(s,v)X(v)= c\,.
	\end{split}
	\end{equation*}
	Non-degeneracy of $(X(v): v \in \mathcal{V}_s) $ then implies that for all $ v \in \mathcal{V}_s $
	\begin{equation*}
		a K(s, v) + \sum_{i = 1}^D a_i \partial_i^sK(s, v)
				  + \sum_{ 1 \leq j \leq k \leq D} a_{jk} \partial_{jk}^sK(s, v) = 0,
	\end{equation*}
	which by the linear independence constraint implies that the constants are all zero.
	This proves non-degeneracy of $ \big(Y(s), \nabla Y(s), \big(\nabla^2 Y(s) \big)\big) $.
	
	For the normalized field $\tilde X / \sqrt{ \Var[\tilde X]} = \tilde X / \sigma$, we note that 
	\begin{equation*}
	\nabla	\frac{X}{\sigma}
			= \frac{\nabla X}{\sigma} - \frac{X \nabla \sigma}{\sigma^2}
			= \frac{\nabla X}{\sigma} - \frac{\nabla \sigma}{\sigma}\left( \frac{X}{\sigma} \right)
	\end{equation*}
	and 
	\begin{equation*}
	\nabla^2 \frac{X}{\sigma}
			= \frac{\nabla^2 X}{\sigma} - \frac{2\left( \nabla X \right)^T(\nabla \sigma)}{\sigma^2}
			- \frac{X\nabla^2 \sigma}{\sigma^2} + \frac{2(\nabla \sigma)^T(\nabla \sigma)X}{\sigma^3}\,.
	\end{equation*}
	Hence we can transform $ \big(X(s), \nabla X(s), \mathbb{V}\big(\nabla^2 X(s)\big)\big) $
	into $ \big(Z(s), \nabla Z(s), \mathbb{V}\big(\nabla^2 Z(s)\big)\big) $ using an invertible matrix. Thus,
	$
	\big(Z(s), \nabla Z(s), \mathbb{V}\big(\nabla^2 Z(s)\big)\big)
	$
	is non-degenerate by Lemma A.2 from \cite{Davenport2022}.
\end{proof}

\subsection{Proof of Proposition \ref{prop:GaussKernelLinearIndependence}}
\subsubsection{Establishing non-degeneracy of the isotropic kernel and its derivatives under linear transformations}

\begin{lemma}\label{lem:nondegeniso}
	Suppose that $ \mathcal{V} $ satisfies the conditions of Proposition \ref{prop:GaussKernelLinearIndependence} and let $ K^* =e^{-\left\lVert s-v \right\rVert^T/ 2}  $ be the $ D $-dimensional isotropic Gaussian kernel. Then given constants $ c, a_i, a_{jk} $ (for $ 1 \leq i \leq D $ and $ 1\leq j \leq k \leq D $), $  s \in \mathbb{R}^D $ and an invertible symmetric matrix $ \Omega' \in \mathbb{R}^{D\times D}$ such that 
	\begin{equation}\label{lem:kernelomega}
	\begin{split}
	c K^*( \Omega's, \Omega'v )  &+ \sum_{j = 1}^D a_j\partial_{j}^xK^*( \Omega's, \Omega'v )\\
	&+ \sum_{ 1 \leq j \leq k \leq D} a_{jk}\partial_{jk}^x K^*( \Omega's, \Omega'v ) = 0,
	\end{split}
	\end{equation}
	for all $ v \in \mathcal{V} $ then $ c = a_i = a_{jk} = 0 $ for $ 1 \leq i \leq D $ and $ 1\leq j \leq k \leq D $.
\end{lemma}
\begin{proof}
	For all $ v\in \mathcal{V} $, letting $ s^* = \Omega's $ and dividing \eqref{lem:kernelomega} by $ e^{-\left\lVert s-v \right\rVert^2/2} $, it follows that 
	\begin{equation*}
	\begin{split}
	c  &+ \sum_{j = 1}^D a_j\left(s^*_j - \sum_{l = 1}^D \Omega'_{jl} v_l\right)\\
	   &+ \sum_{ 1 \leq j \leq k \leq D} a_{jk} \left( \left(s^*_j - \sum_{l = 1}^D \Omega'_{jl} v_l\right)\left(s^*_k - \sum_{l = 1}^D \Omega'_{kl} v_l\right)  - \delta_{jk} \right) = 0.
	\end{split}
	\end{equation*}
	In particular,
	\begin{equation}\label{eq:lincor2}
	\begin{split}
	c  &+ \sum_{j = 1}^D a_j\left(s^*_j - \sum_{l = 1}^D \Omega'_{jl} v_l\right)\\
	   &+ \sum_{ 1 \leq j, k \leq D} a'_{jk} \left( \left(s^*_j - \sum_{l = 1}^D \Omega'_{jl} v_l\right)\left(s^*_k - \sum_{l = 1}^D \Omega'_{kl} v_l\right)  - \delta_{jk} \right) = 0.
	\end{split}
	\end{equation}
	where $ a'_{jk} = a_{jk}/2 $, $ j \not= k $ and $ a'_{jj} = a_{jj} $. 
	
	For $ i \in \{ 1, \dots, D \}$, fixing $ (v_1, \dots, v_{i-1}, v_{i+1}, \dots, v_D) $,
	we can view (\ref{eq:lincor2}) as a quadratic in $ v_i $.
	As such the only way that it can have more than two distinct solutions is
	if where the coefficient of $ v_i^2 $ is zero, i.e., 
	\begin{equation}\label{eq:ii}
		\sum_{1 \leq j,k \leq D} a_{jk}'\Omega'_{ji}\Omega'_{ki} = (\Omega' A' \Omega')_{ii} = 0\,.
	\end{equation}
	Similarly the coefficient of $ v_i$ must be zero, i.e.,
	\begin{equation}\label{eq:vn}
		\sum_{j=1}^D a_j \Omega'_{ji} + \sum_{j,k} \Omega'_{ji}\left(s_k^* -\sum_{m \not= i} \Omega'_{km}v_m\right) + \sum_{j,k} \left(s_k^* -\sum_{l \not= i} \Omega'_{jl}v_l\right)\Omega'_{ki} = 0.
	\end{equation}
	Now allowing $ v_n $ to vary for some $ n \not= i $, by the same logic, the coefficient of $ v_n $ in \eqref{eq:vn} is equal to zero, i.e., 
	\begin{equation*}
		\sum_{j,k} a_{jk}'\Omega'_{ji}\Omega'_{kn} + \sum_{j,k} a_{jk}'\Omega'_{jn}\Omega'_{ki} = (\Omega' A' \Omega')_{in} + (\Omega' A' \Omega')_{ni} = 2(\Omega' A' \Omega')_{in} = 0.
	\end{equation*}
	As such $ (\Omega' A' \Omega')_{in} = (\Omega' A' \Omega')_{ni} = 0 $ for all $ i \not=n $.
	Combining this with \eqref{eq:ii}, it follows that $ \Omega' A' \Omega' = 0$.
	In particular $ A' = 0 $ as $ \Omega' $ is invertible.
	Thus, the remaining linear equation in $v_i$ from \eqref{eq:lincor2} can only have more than
	one solution if  
	\begin{equation*}
		\sum_j a_j \Omega'_{ji} = 0.
	\end{equation*}
	Therefore
	$ (\Omega' a)_i = 0$, where $ a = (a_1, \dots, a_D)^T $. Since this holds for all $ i $ and $ \Omega' $ is invertible, we obtain $ a = 0 $. Finally this implies that $ c = 0. $
\end{proof}

\subsubsection{Establishing Proposition \ref{prop:GaussKernelLinearIndependence}}
\begin{proof}
	We can write $ K(s,v) = K^*(\Sigma^{-1/2}s,\Sigma^{-1/2}v)$ where $ K^* =e^{-\left\lVert s-v \right\rVert^2/2} $ is the isotropic Gaussian kernel. Arguing as in proof of Lemma 1 of \cite{Davenport2022} (taking $ \phi(s) = \Sigma^{-1/2}s $ and $ \varphi $ to be the identity in their notation), for each $ s \in S $ we have 
	\begin{equation*}
	\mathbb{V}\big(\nabla^2 K(s, v) \big) = L\big(\Sigma^{1/2} \otimes \Sigma^{1/2}\big)R \mathbb{V}\Big(\nabla^2 K^*\big( \Sigma^{-1/2}s, \Sigma^{-1/2}v \big) \Big)
	\end{equation*}
	where $\nabla^2$ as usual always denotes the Hessian with respect to the first argument
	and $ L \in \mathbb{R}^{D(D+1)/2 \times D^2}
	$ is the elimination matrix and $ R \in \mathbb{R}^{D^2 \times D(D+1)/2} $ is
	the duplication matrix, the precise definitions of which can be found in
	\cite{Magnus1980}. The matrix $L \big(\,  \Sigma^{1/2} \otimes  \Sigma^{1/2} \,\big)R$
	is invertible by Lemma 4.4.iv of \cite{Magnus1980}, and the fact that $ \Sigma^{1/2}$ is invertible.
	
	Moreover $ \nabla K\big(\Sigma^{-1/2}s, v\big) = \Sigma^{-1/2} \nabla K^*\big(\Sigma^{-1/2}s, \Sigma^{-1/2}v\big) $. As such there is an invertible linear transformation between the vector 
	\begin{equation*}
	\Big( K(s,v) , \nabla K( s, v ), \mathbb{V}\big(\nabla^2 K( s, v) \big) \Big) 
	\end{equation*}
	and the vector
	\begin{equation*}
	\Big( K^*\big(\Sigma^{-1/2}s,\Sigma^{-1/2}v\big), \nabla K^*\big( \Sigma^{-1/2}s, \Sigma^{-1/2}v \big), \mathbb{V}\big(\nabla^2 K^*\big( \Sigma^{-1/2}s, \Sigma^{-1/2}v\big)\big) \Big).
	\end{equation*}
	In particular if there exists constants $ c, a_d, a_{kl} $ such that $ 1 \leq d \leq D $ and $ 1\leq k\leq l \leq D $
	 (with at least one of them being non-zero) such that 
	\begin{equation}
	c K(s,v)  + \sum_{d = 1}^D a_d\partial_{d}^s K(s,v ) + \sum_{ 1 \leq k \leq l \leq D} a_{kl} \partial_{kl}^s K( s,v ) = 0,
	\end{equation}
	then there existing corresponding constants
	$ (c^*, a^*_d, a^*_{kl}) $ (with at least one of them being non-zero) such that 
	\begin{equation}
	\begin{split}
	c^* K^*\big( \Sigma^{-1/2}s, \Sigma^{-1/2}v \big)  &+ \sum_{d = 1}^D a^*_d\partial_{d}^s K^*\big( \Sigma^{-1/2}s, \Sigma^{-1/2}v \big)\\
	&+ \sum_{ 1 \leq  k\leq l \leq D} a^*_{kl} \partial_{kl}^s K^*\big( \Sigma^{-1/2}s, \Sigma^{-1/2}v \big) = 0.
	\end{split}
	\end{equation}
	Applying Lemma \ref{lem:nondegeniso} yields a contradiction and thus establishes the result.
\end{proof}

\subsection{Proof of Proposition \ref{cor:geometry-prop}}
\begin{proof}
	Note that for a SuRF $(\tilde X, X, K, \mathcal{V})$ we obtain the following identity
	\begin{equation}
    	\ska{\partial_d K_s}{ K_{s'} } =  \Cov{ \partial_d^s\tilde X(s)}{ \tilde X(s') } = \ska{\partial_d^s}{1_{s'}}\,.
	\end{equation}
	Hence the result for a normalized SuRF is a Corollary of Theorem \ref{thm:RiemannianMetric}.
\end{proof}

\subsection{Proof of Theorem \ref{thm:GKF}}
\begin{proof}
	In order to apply Theorem 12.4.2 from \cite{Adler2007} we need to prove that
	the assumptions \textbf{(G1)}-\textbf{(G3)} hold.
	The assumption that $\tilde X_\alpha$ is Gaussian with almost surely
	$C^2$-sample paths clearly holds by the assumption that $X$ is a Gaussian field on
	$\cV$ and
	$K(\cdot, v) \in C^3\big( \overline{\cM} \big)$ for all $v\in\cV$.
	The non-degeneracy condition follows from Proposition \ref{prop:nondegen}.
	The last assumption that there is an $\epsilon > 0$ such that
	$$
		\mathbb{E}\!\left[
			\big( \partial_{dd'} \tilde X_\alpha(x) -\partial_{dd'} \tilde X_\alpha(y)\big)^2
			\right]
		\leq K \big\vert \log\Vert x - y \Vert \big\vert^{-(1+\gamma)}
	$$
	for some $K>0$, all $d,d' \in \{1,...,D\}$ and for the charts $(\overline{U}_\alpha,\overline{\varphi}_\alpha)$, $\alpha \in \{1,\ldots,P\}$,	in the atlas of $\overline{\cM}$ covering $\cM$,
	all $x,y \in \overline{\varphi}\big( \overline{U}_\alpha \big) \cap \cM$ such that
	$\vert x - y\vert < \epsilon$ is established in Proposition \ref{prop:GaussG3}.
\end{proof}

\subsection{Proof of Theorem \ref{thm:L1VM}}
\paragraph{ Computation of $ \Theta(x) =  \int_{\mathbb{S}\big(\cT_x\partial_{1}\cM_\cV^\bot\big)}  \alpha(\nu)\mathcal{H}_{1}(d\nu)$ for Voxel Manifolds.}
In order to compute $\Theta(x)$ we need to introduce the normal Morse index $\alpha(\nu)$.
We specialize here to the case of  $\cM_\cV$ being a voxel manifold embedded into $\mR^3$,
yet the exact same concept is defined for any WS manifold, compare  \cite[Scn 9.2.1]{Adler2007}.

For any $x \in \cM_\cV$ and any direction $\nu \in \mathbb{S}\big(\cT_x\cM_\cV\big)$
the normal Morse index  is one minus the local Euler characteristic (EC) of the intersection of $\cM_\cV$,
the $\delta$-ball centered at $x$ and the affine
plane $\{ \lambda\in \mR^3:~\lambda^T \nu + x + \epsilon = 0 \}$ for $\epsilon >0$.
If $\delta$ is sufficiently small, this EC does not depend on $\epsilon$ provided
that $\epsilon$ is small enough. Although this definition sounds complicated at first, it can be easily computed
for all $x$ and $\nu$ for a voxel manifold. Note that $\cM_\cV$ or more precisely an open neighbourhood of it
might be endowed with a different Riemannian metric than
the standard Riemannian metric on $\mR^3$, which
in our case is the induced Riemannian metric $\boldsymbol{\Lambda}$
by a unit-variance random field $f$. In this case
$\nu \in \mathbb{S}\big(\cT_x\mR^3\big)$ are vectors $\nu\in \mR^3$ such that $\nu^T \boldsymbol{\Lambda}(x)\nu=1$.

In the case that $x\in\partial_3\cM_\cV$ and $\nu\in \mathbb{S}\big(\cT_x\mR^3\big)$ it
is obvious that $\alpha(\nu) = 0$ because the intersection of the affine plain defined by $\nu$ and the $\delta$-ball
is always homeomorphic to a filled disk, which has EC $1$. Similarly, if
$x\in\partial_2\cM_\cV$, then $\alpha(\nu) = 0$ for all $\nu\in \mathbb{S}\big(\cT_x\mR^3\big)\setminus\{\nu_{out}\}$.
Here $\nu_{out}$ is the unique outside pointing normal (w.r.t. the $\boldsymbol{\Lambda}$ metric) at $x\in\partial_2\cM_\cV$ and it holds that
$\alpha(\nu_{out})=0$ as  intersection of the affine plain defined by $\nu_{out}$ for a small enough $\delta$-ball
is again homeomorphic to a filled disk.
The interesting cases happen at the edges of the voxel manifold, i.e., for $x\in\partial_1\cM_\cV$.
Here the behavior of $\alpha(\nu)$ can be classified
by the three types of possible edges: the convex, the double convex and the concave edge.
These cases are shown in Figure \ref{fig:EdgeVisualization} in Appendix \ref{A:LKCdefn}.

The behavior of $\alpha$ for directions $\nu \in \mS(\cT_x\partial_{1}\cM_\cV^\bot)$
is demonstrated in Figure \ref{fig:alpha} within the
hyperplane $x + \cT_x\partial_{1}\cM_\cV^\bot$ with $\cM_\cV$. Here we show the two possible
intersection scenarios of the hyperplane (bold green line) orthogonal to the direction
$-\nu\in\mS(\cT_x\partial_{1}\cM_\cV^\bot)$ (green arrow) with
a small $\delta$-ball (dotted black line) and the voxel manifold $\cM_\cV$ from the definition of $\alpha$.
In particular, it can be seen that $\alpha(\nu)$ is constant, if $-\nu$ is inside the normal cone
$\mathcal{N}_x\cM_\cV$ (for a definition see \cite[p.189]{Adler2007}) and constant on
$\mS(\cT_x\partial_{1}\cM_\cV^\bot) \setminus \mathcal{N}_x\cM_\cV$ independent on the type of edge
to which $x$ belongs.
The geometric embedding of the the intersection of the (geometric) normal cone $x+\mathcal{N}_x\cM_\cV$
with the hyperplane $x + \cT_x\partial_{1}\cM_\cV^\bot$ is represented by the red shaded areas.
From this we deduce that on a convex edge the EC of the intersection of the green
hyperplane with the $\delta$-ball and $\cM_\cV$ is $0$, if
$-\nu\in \mS(\cT_x\partial_{1}\cM_\cV^\bot)$, as the intersection
is empty and $1$ otherwise because the intersection is homeomorphic to a disk.
Similar, it holds that the EC of the intersection
is $2$ for $x$ on a double convex or concave edge, if $-\nu\in \mS(\cT_x\partial_{1}\cM_\cV^\bot)$,
as the intersection is homeomorphic to the disjoint union of two disks and $1$ else because the intersection
is homeomorphic to a disk. Therefore we obtain that $\alpha(\nu)$ for $\nu\in\mS(\cT_x\partial_{1}\cM_\cV^\bot)$
is given by
	\begin{equation}
	\alpha( \nu ) = \begin{cases}
	~~1, \text{ if } x \text{ on convex edge and } -\nu\in \left( \mathcal{N}_x\cM_\cV \right)^\circ \\
	-1, \text{ if } x \text{ on a double convex edge and } -\nu\in \left( \mathcal{N}_x\cM_\cV \right)^\circ \\
	-1, \text{ if } x \text{ on concave edge and } -\nu\in \left( \mathcal{N}_x\cM_\cV \right)^\circ \\
	~~0,  \text{ else }
	\end{cases}\,.
	\end{equation}
Using this we can compute the function $\Theta(x)$.

	\begin{figure}[t]
		\begin{center}
			\hspace{1cm}\textbf{Convex Edge}	\hspace{4cm}		\textbf{Double Convex Edge}\\
			\includegraphics[width=0.4\textwidth]{\figurepath 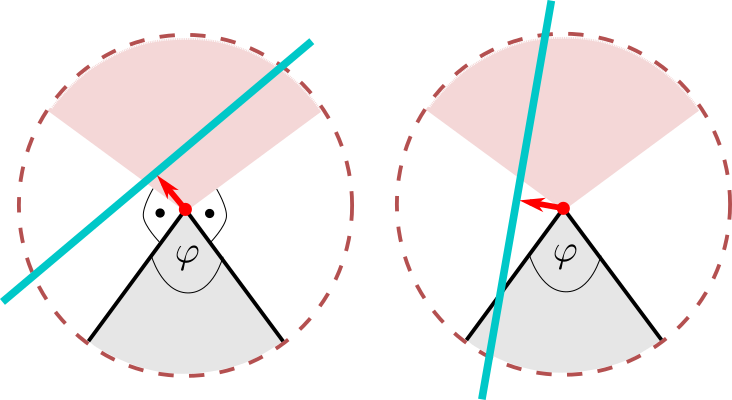}\quad\quad\quad\quad
			\includegraphics[width=0.4\textwidth]{\figurepath 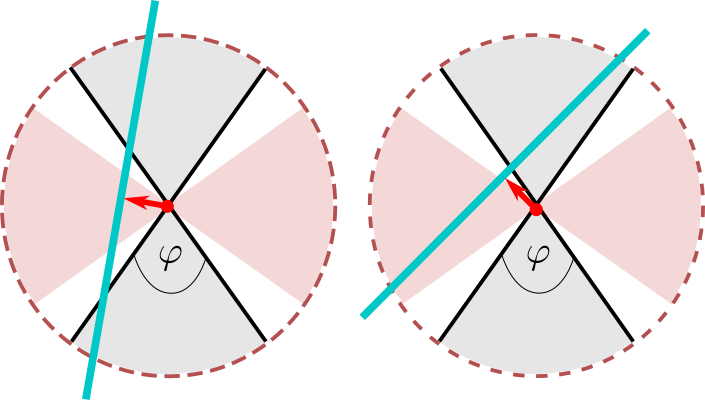}\\
			\textbf{Concave Edge}\\
			\includegraphics[width=0.4\textwidth]{\figurepath 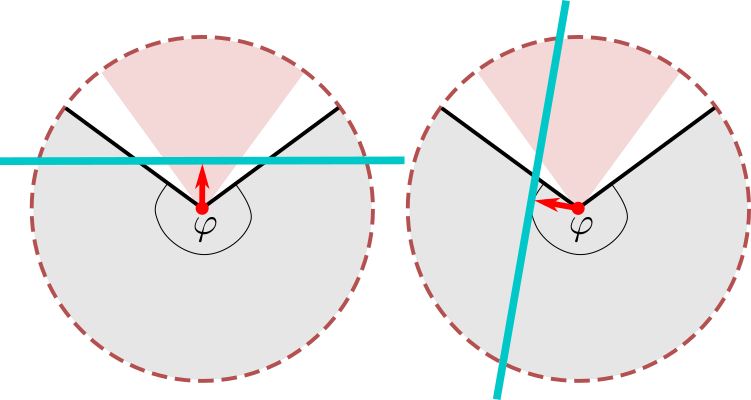}
		\end{center}
		\caption{The two different scenarios of intersections of the hyperplane at the
				 three different types of edges appearing in a voxel manifold $\cM_\cV$ illustrated
				 within the hyperplane $x + \cT_x\partial_{1}\cM_\cV^\bot$. The
				 dashed red line is the boundary of an $\delta$-ball centered at $x$.
				 The grey shaded area
				 belongs to $\cM_\cV$ and the red shaded area is the (geometric) normal cone
				 $x+\mathcal{N}_x\cM_\cV$ which needs to be orthogonal to the boundary of $\cM_\cV$
				 in $x + \cT_x\partial_{1}\cM_\cV^\bot$. The fact that the opening angles $\varphi$ are not necessarily
				 in $\{ \pi/2, 3\pi/2\}$ appear as $x + \cT_x\partial_{1}\cM_\cV^\bot$ is not
				 orthogonal to the edge $x$ belongs to.
				 The red arrow is the unit direction
				 $-\nu\in \mS(\cT_x\partial_{1}\cM_\cV^\bot)$ scaled by $\epsilon$ and the turquoise line is the hyperplane defined by $-\nu$.}\label{fig:alpha}
	\end{figure}
\FloatBarrier

\begin{lemma}\label{thm:theta}
	Let $\cM_\cV$ be a voxel manifold, $x\in\partial_1\cM_\cV$ and $U_x,V_x,W_x$ denote an ON frame for $\cT_x \overline{\cM_\cV}$. Define $M = V_x\times W_x = \big( m_1(x), m_2(x), m_3(x) \big)^T$. Define
	\begin{equation*}
	\beta(x) = \arccos\left( \frac{ m_2(x)m_3(x) }{ \sqrt{ m_2^2(x) + m_1^2(x) }\sqrt{ m_3^2(x) + m_1^2(x) } } \right).
	\end{equation*}
	Then we obtain
	\begin{equation*}
			\Theta(x) = \begin{cases}
					\pi - \beta(x)\,, & \text{if $x$ belongs to a convex edge}\\
					-2\beta(x)\,, & \text{if $x$ belongs to a double convex edge}\\
					\beta(x) - \pi\,, & \text{if $x$ belongs to a convex edge}
			\end{cases}\,.
	\end{equation*}
\end{lemma}
\begin{proof}
	Let $\beta(x)$ be the minimum of $\varphi$ and $2\pi-\varphi$ where $\phi$ is the opening angle
	within $\cM_\cV$ from Figure \ref{fig:alpha}.
	This yields
	\begin{equation*}
	\begin{split}
	\int_{\mS(\cT_x\partial_{1}\cM_\cV^\bot)}  \alpha(\nu_{2})\mathcal{H}_{1}(d\nu_{2})
	&=  \int_{0}^{2\pi}  \1_{ -\sin(t) E_1 - \cos(t)E_2 \in \left( \mathcal{N}_x\cM_\cV \right)^\circ  } dt\\
	&= \begin{cases}
					\pi - \beta(x)\,, & \text{if $x$ on a convex edge}\\
					-2\beta(x)\,, & \text{if $x$ on a double convex edge}\\
					\beta(x) - \pi\,, & \text{if $x$ on a convex edge}
			\end{cases}\,,
	\end{split}
	\end{equation*}
	because the characteristic function is only $1$,
	if $-\nu$ belongs to the normal cone and hence the integral is equal to the opening
	angle of the normal cone (red shaded area in Fig. \ref{fig:alpha}).
	
	It remains to compute the angle $\beta(x)$ which is obtained by computing the angle between the intersection of
	the affine plane $x + \cT_x\partial_{1}\cM_\cV^\bot$ and the boundary of $\cM_\cV$
	at $x$. We only treat the case of $x \in \cM_\cV$ lying on a convex edge.
	Double convex and concave edges follow analogously. Since the tangent space and $\mR^3$ can be identified
	we assume w.l.o.g. that $x=0$ and the voxel is given by the set
	$
	\{ y\in\mR^3:~y_3 \leq 0, y_2 \leq 0 \}
	$.
	(note we extend the edge infinitely, which does not make a difference in this argument)
	Its boundary is given by the set
	$
	A \cup B = \{ y \in \mR^3:~y_3 = 0, y_2 \leq 0 \} \cup \{ y\in\mR^3:~y_3 \leq 0, y_2 = 0 \}
	$.
	The edge to which $x$ belongs is given by
	$
	E = \{ y \in \mR^3:~y_3 = 0, y_2 = 0 \}
	$.
	An orthonormal basis at $x$ with $U_x$ spanning $\cT_x E$ is given in
	Proposition \ref{prop:ONF}. The plane in which the unit circle
	$\mS^1\big( \cT_xE^\bot \big)$ lies is given by the linear span of $V_x,W_x$,
	which we denote by
	$
	\cF = \big\{\, y\in\mR^3:~m_1(x) y_1 + m_2(x) y_2 + m_3(x) y_3 = 0 \,\big\}
	$
	for some $m_1(x),m_2(x),m_3(x)\in\mR$. The intersection $\cF\cap A$ and $\cF\cap B$ are given by
	\begin{equation}
	\begin{split}
	A \cap \cF &= \big\{\, y\in \mR^3:~ m_1(x)y_1 + m_2(x)y_2 = 0 ~\wedge~ y_2\leq0  ~\wedge~ y_3=0 \big\}\\
	B \cap \cF &= \big\{\, y\in \mR^3:~ m_1(x)y_1 + m_3(x)y_3 = 0 ~\wedge~ y_3\leq0  ~\wedge~ y_2=0 \big\}
	\end{split}
	\end{equation}
	By construction $m_1(x) \neq 0$, since otherwise $\cF\cap A = \cF\cap B = E$
	and hence $V,W$ cannot be both orthogonal to $E$ which contradicts the
	assumption that $U_x,V_x,W_x$ form an orthonormal basis for $\cT_x\overline{\cM_\cV}$.
	Thus, if $m_1(x) \neq 0$ we have that $( m_2(x)/m_1(x), -1, 0 )$ and $( m_3(x)/m_1(x), 0, -1 )$
	are vectors in the intersection, which we can identify with tangent directions along $E$. Thus,
	\begin{equation}
	\begin{split}
	\varphi &= \arccos\left( \frac{ \frac{m_2(x)m_3(x)}{m_1^2(x)} }{ \sqrt{ \frac{m_2^2(x)}{m_1^2(x)} + 1 } \sqrt{ \frac{m_3^2(x)}{m_1^2(x)} + 1 }} \right)\\
	 &= \arccos\left( \frac{ m_2(x)m_3(x) }{ \sqrt{ m_2^2(x) + m_1^2(x) } \sqrt{ m_2^2(x) + m_1^2(x) } } \right)\,.
	\end{split}
	\end{equation}
	The same formula holds true if $m_1(x)=0$ and $m_2(x)\neq 0$ and $m_3(x)\neq 0$.
\end{proof}

\paragraph{ Computation of $\int_{\partial_{2} M} \boldsymbol{\Lambda}_x\Big(\overline{\nabla}_{U_x} U_x+\overline{ \nabla}_{V_x}V_x, N_x\Big) \mathcal{H}_{2}(dx)$  for Voxel Manifolds }
Assume w.l.o.g. that $\cT_x\cM_\cV$ for $x\in\partial_2 \cM_\cV$ is contained in the $x_1$-$x_2-$plane.
An orthonormal frame is given by $U_x,V_x,W_x$ from \eqref{prop:ONF} 
and by construction $N_x=\pm W_x$, where the sign depends on whether $W_x$
is inward or outward pointing.
Using the coordinate representations $U_x = \sum_{d=1}^3 U_d(x)\partial_d$ and
$V_x = \sum_{d=1}^3 V_d(x)\partial_d$, linearity and product rule for the covariant
derivative and formula \eqref{eq:covariantDeriv}, we obtain
\begin{equation*}
\begin{aligned}
 \boldsymbol{\Lambda}_x\big( \overline{\nabla}_{U_x} U_x, N_x \big) 
&=  N^T_x  \boldsymbol{\Lambda}(x) U_1(x) \begin{pmatrix} \partial_1 U_{1}(x)\\0\\0 \end{pmatrix} + 
U_1^2(x) N^T_x \begin{pmatrix} \Gamma_{111}(x)\\ \Gamma_{112}(x)\\ \Gamma_{113}(x)\end{pmatrix}\\
&=  U_1^2(x) N^T_x \begin{pmatrix} \Gamma_{111}(x)\\ \Gamma_{112}(x)\\ \Gamma_{113}(x)\end{pmatrix}\,.
\end{aligned}
\end{equation*}
Here the second equality is due to the fact that $N_x \propto \boldsymbol{\Lambda}^{-1}E_3$, if represented as a vector. Similarly it holds that
\begin{equation*}
\begin{split}
 \boldsymbol{\Lambda}_x\big( \overline{\nabla}_{V_x} V_x, N_x ) 
&=
V_1^2(x) N^T_x \begin{pmatrix} \Gamma_{111}(x)\\ \Gamma_{112}(x)\\ \Gamma_{113}(x)\end{pmatrix}  + 
V_2^2(x) N^T_x \begin{pmatrix} \Gamma_{221}(x)\\ \Gamma_{222}(x)\\ \Gamma_{223}(x)\end{pmatrix}\\
&\quad\quad+ V_1(x)V_2(x) N^T_x \begin{pmatrix} \Gamma_{121}(x)\\ \Gamma_{122}(x)\\ \Gamma_{123}(x)\end{pmatrix}
\end{split}
\end{equation*}

Summarizing this yields the following proposition about the trace of the shape operator along
$\partial_2\cM_\cV$, i.e.,
$\overline{g}\Big(\overline{\nabla}_{U_x} U_x+\overline{\nabla}_{V_x}V_x, N_x\Big)$.
\begin{proposition}
	Let $\cM_\cV$ be a voxel manifold and assume that $x \in \partial_2 \cM_\cV$ such that
	$\cT_x\partial_2 \cM_\cV$ is spanned by $E_k, E_l$. Then
	\begin{equation*}
	\begin{aligned}
	 \boldsymbol{\Lambda}_x\Big(\overline{\nabla}_{U_x} U_x + &\overline{\nabla}_{V_x}V_x, N_x\Big)\\
	&=
	\left( U_k^2(x) + V_k^2 \right) N_x^T \begin{pmatrix} \Gamma_{kk1}(x)\\ \Gamma_{kk2}(x)\\ \Gamma_{kk3}(x)\end{pmatrix}
	+  V_l^2(x) N_x^T \begin{pmatrix} \Gamma_{ll1}(x)\\ \Gamma_{ll2}(x)\\ \Gamma_{ll3}(x)\end{pmatrix}\\
	&\quad\quad+ V_k(x)V_l(x) N_x^T \begin{pmatrix} \Gamma_{kl1}(x)\\ \Gamma_{kl2}(x)\\ \Gamma_{kl3}(x)\end{pmatrix}
	\end{aligned}
	\end{equation*}
\end{proposition}

\subsection{Proof of Theorem \ref{thm:unbiased}}
We prove this result in more generality as the restriction to SuRFs is not necessary. The proof extends that of \cite{Taylor2007} which established the result for the trivial linear model.
To do so assume that $Y_1, \ldots, Y_N\sim Y$ is an iid sample of Gaussian random fields
on a voxel manifold $\cM_\cV$. We restrict here to voxel manifolds to avoid working in a
chart, however, the proof easily generalizes. This sample we represent as the vector
$\boldsymbol{Y} = \big( Y_1, \ldots, Y_N \big)^T$ and therefore the corresponding estimator generalizing
\eqref{eq:riem_SuRFest} is, for $x\in\cM_\cV$,
\begin{equation}\label{eq:GenLambdaEst}
\begin{split}
        \hat{\Lambda}_{dd'}(x) &=
        	\frac{ \Cov{ \partial_{d}\boldsymbol{Y}(x)}{ \partial_{d'}\boldsymbol{Y}(x) } }
        		 { \Var\left[ \boldsymbol{Y}(x) \right] }\\
       	  	&=
        	\Cov{ \partial_{d} \left(\frac{\boldsymbol{Y}(x)}{\sqrt{\Var\left[ \boldsymbol{Y}(x) \right]}} \right)}
        		{ \partial_{d'} \left(\frac{\boldsymbol{Y}(x)}{\sqrt{\Var\left[ \boldsymbol{Y}(x) \right]}} \right)}\\
       	  &\quad\quad\quad\quad\quad- \frac{ \Cov{ \partial_{d}\boldsymbol{Y}(x)}{ \boldsymbol{Y}(x) }
       	  		   \Cov{ \boldsymbol{Y}(x)}{ \partial_{d'}\boldsymbol{Y}(x) }}
       	  		 { \Var\left[ \boldsymbol{Y}(x) \right]^2 } \,.
\end{split}
\end{equation}
Here as in the main manuscript the operation $\Var\big[\cdot\big]$ and $\Cov{\cdot}{\cdot }$ denote the sample variance and sample covariance respectively.
We define the vector of normalized residuals to be $\mathbf{R}(x)= \frac{\boldsymbol{\rm H}\boldsymbol{Y}(x)}{\Vert \boldsymbol{\rm H}\boldsymbol{Y}(x) \Vert}$,
 $x\in\cM_\cV$, where $\Vert \cdot \Vert$ denotes the Euclidean norm and
 $\boldsymbol{\rm H} = I_{N\times N} - \mathds{1}\mathds{1}^T$ with $\mathds{1}^T = (1,\ldots,1) \in \mathbb{R}^N$,
 is the centering matrix. Note that $\mathbf{R}$ does not depend on the (unknown) variance $\Var[Y]$
 as for the sample $\tilde{\boldsymbol{Y}} = \boldsymbol{Y} / \sqrt{\Var[Y]}$ we have that
 $\tilde{\mathbf{R}}(x)=\mathbf{R}(x)$ for all $x\in\cM_\cV$.
 Thus, using
 \begin{equation*}
 	\nabla \mathbf{X}(x) = \begin{pmatrix}
 		\frac{\partial X_1}{\partial x_1}(x) & \ldots & \frac{\partial X_1}{\partial x_D}(x)\\
 		\vdots & \ddots & \vdots \\
 		\frac{\partial X_D}{\partial x_1}(x) & \ldots & \frac{\partial X_D}{\partial x_D}(x)\\
\end{pmatrix} 	 
\end{equation*}
for $ \mathbf{X} = (X_1,\ldots, X_D)^T \in C\big( \cM_\cV, \mathbb{R}^D \big)$,
 we can rewrite \eqref{eq:GenLambdaEst} in matrix terms as
 \begin{equation*}
 	\hat{\boldsymbol{\Lambda}} = \big(\nabla \mathbf{R}\big)^T \nabla \mathbf{R}\,.
 \end{equation*}
 Recall that the true underlying $\boldsymbol{\Lambda}$ is given by
 \begin{equation*}
 	\boldsymbol{\Lambda} = \mathbb{E}\left[
 							\left(\nabla\left( \frac{Y - \mathbb{E}[Y]}{\sqrt{\Var[Y]}} \right)\right)^T
 							\nabla\left( \frac{Y - \mathbb{E}[Y]}{\sqrt{\Var[Y]}} \right) 							
 							\right]
 					= \mathbb{E}\left[ \big(\nabla\tilde{Y}^T \big) \nabla\tilde{Y} \right]\,,
 \end{equation*} 
 where $\tilde{Y} = \frac{Y}{\sqrt{\Var[Y]}}$.
 Using this we obtain the following lemma which the proof of which follows the corresponding proof in \cite{Taylor2007}.
 \begin{lemma}\label{lem:DetExpect}
 Under the assumption and notation described above, we have that 
 	\begin{equation}
 		\mathbb{E}\left[ \sqrt{\det\left( \hat{\boldsymbol{\Lambda}}(x) \right)} \right]
 			=   \sqrt{\det\left( \boldsymbol{\Lambda}(x) \right)}
 	\end{equation}
 \end{lemma}
\begin{proof}
	$\boldsymbol{\rm H}$ is a projection matrix and so $\boldsymbol{\rm H}\boldsymbol{\rm H}=\boldsymbol{\rm H}$ and $\boldsymbol{\rm H}^T=\boldsymbol{\rm H}$. We can expand $\nabla \boldsymbol{R}$ in terms of $\tilde{\boldsymbol{Y}}$ as follows:
	\begin{equation*}
		\nabla\boldsymbol{R}
		= \frac{\nabla \boldsymbol{\rm H}\tilde{\boldsymbol{Y}}}{ \Vert \boldsymbol{\rm H}\tilde{\boldsymbol{Y}} \Vert}
			- \frac{\boldsymbol{\rm H}\tilde{\boldsymbol{Y}}\nabla\big( (\boldsymbol{\rm H}\tilde{\boldsymbol{Y}})^T\boldsymbol{\rm H}\tilde{\boldsymbol{Y}}\big)}{ 2\Vert \boldsymbol{\rm H}\tilde{\boldsymbol{Y}} \Vert^3}
		= \Bigg( I_{N\times N} - \frac{ \boldsymbol{\rm H} \tilde{\boldsymbol{Y}} (\boldsymbol{\rm H}\tilde{\boldsymbol{Y}})^T}{ \Vert \boldsymbol{\rm H}\tilde{\boldsymbol{Y}} \Vert^2} \Bigg)\frac{  \boldsymbol{\rm H}\nabla \tilde{\boldsymbol{Y}} }{\Vert \boldsymbol{\rm H}\tilde{\boldsymbol{Y}} \Vert}\,.
	\end{equation*}
	Now note that
	\begin{equation*}
		\boldsymbol{\rm H}\Bigg( I_{N\times N} - \frac{ \boldsymbol{\rm H} \tilde{\boldsymbol{Y}} (\boldsymbol{\rm H}\tilde{\boldsymbol{Y}})^T}{ \Vert \boldsymbol{\rm H}\tilde{\boldsymbol{Y}} \Vert^2} \Bigg) \boldsymbol{\rm H}
		= \boldsymbol{\rm H} - \frac{ \boldsymbol{\rm H} \tilde{\boldsymbol{Y}} (\boldsymbol{\rm H}\tilde{\boldsymbol{Y}})^T}{ \Vert \boldsymbol{\rm H}\tilde{\boldsymbol{Y}} \Vert^2}
	\end{equation*}
	and therefore it is idempotent with
	\begin{equation*}
		\tr\Bigg( \boldsymbol{\rm H} - \frac{ \boldsymbol{\rm H} \tilde{\boldsymbol{Y}} (\boldsymbol{\rm H}\tilde{\boldsymbol{Y}})^T}{ \Vert \boldsymbol{\rm H}\tilde{\boldsymbol{Y}} \Vert^2}  \Bigg)
		=  \tr\big( \boldsymbol{\rm H} \big) - 1
	\end{equation*}
	Thus, applying Cochran's Theorem \cite[Theorem 3.4.4]{Mardia}, it follows that for all $x \in S$,
	\begin{equation*}
		\big(\nabla\boldsymbol{R}(x)\big)^T \nabla\boldsymbol{R}(x) ~\vert~\tilde{\boldsymbol{Y}}(x) ~\sim~ {\rm Wish}_D\big(\boldsymbol{\Lambda}\Vert \boldsymbol{\rm H}\tilde{\boldsymbol{Y}} \Vert^{-2}, \tr\big( \boldsymbol{\rm H} \big) - 1 \big).
	\end{equation*}
	As such by \cite[Corollary 3.4.1.2]{Mardia}
	\begin{equation*}
	\begin{split}
	 \boldsymbol{W}(x)
	 &\sim \Vert \boldsymbol{\rm H}\tilde{\boldsymbol{Y}} \Vert^2 \boldsymbol{\Lambda}(x)^{-1/2} \big(\nabla \boldsymbol{R}(x)\big)^T\nabla \boldsymbol{R}(x) \boldsymbol{\Lambda}(x)^{-1/2} ~\vert~\tilde{\boldsymbol{Y}}(x)\\
	 &\sim  {\rm Wish}_D\big(I_{N\times N}, \tr( \boldsymbol{\rm H} ) - 1 \big)
	\end{split}
	\end{equation*}
	Taking determinants and rearranging it follows that, unconditionally,
	\begin{equation*}
		\sqrt{\det\Big( \big(\nabla \boldsymbol{R}(x)\big)^T \nabla \boldsymbol{R}(x)  \Big)}
		\sim \sqrt{\det\big( \boldsymbol{\Lambda}(x) \big)\det\big(\boldsymbol{W}(x)\big) V(x)^{-D}}\,,
	\end{equation*}
	where $\boldsymbol{W}(x)$ is independent of $V (x) = \Vert \boldsymbol{\rm H}\tilde{\boldsymbol{Y}} \Vert^2 \sim \chi^2_{\tr( \boldsymbol{\rm H} )}$.
	From the independence, the formula for the moments of $\chi^2$-distributions and Theorem 3.4.8  from
	\cite{Mardia} we obtain
	\begin{equation*}
	\begin{split}
		\mathbb{E}\Bigg[ \sqrt{\det\Big(  \big(\nabla \boldsymbol{R}(x)\big)^T \nabla \boldsymbol{R}(x) \Big)} \Bigg]
		&= \sqrt{\det\big( \boldsymbol{\Lambda}(x) \big)}\mathbb{E}\Big[  \sqrt{\det\big(\boldsymbol{W}(x)\big)} \Big]\\ \mathbb{E}\Big[V(x)^{-\frac{D}{2}}\Big]
		&= \det\big( \boldsymbol{\Lambda}(x) \big)^{1/2}\,.
	\end{split}
	\end{equation*}
\end{proof}
\begin{remark}
	Lemma \ref{lem:DetExpect} also holds for estimates of $ \mathbf{\Lambda} $ in a linear model, compare \cite{Taylor2007}. This follows by a similar proof as, in that setting, the residuals $\mathbf{R}$ which are used to calculate $\hat{\boldsymbol{\Lambda}}$ are obtained by $\mathbf{R} = \mathbf{P}\mathbf{Y}$, where $\mathbf{P}$ is idempotent.
\end{remark}

\begin{proof}[Proof of Theorem \ref{thm:unbiased}]
In what follows we establish the results for $\hat{\cL}_D^{(r)}$. The proof for $\hat{\cL}_{D-1}^{(r)}$
is identical as each $\boldsymbol{\Lambda}^I$ is a $(D-1)\times(D-1)$-submatrix of $\boldsymbol{\Lambda}$.

Using Lemma \ref{lem:DetExpect} and the approximation of integrals of continuous functions by Riemann sums yields
\begin{align*}
	\lim_{r\rightarrow\infty}
			\mathbb{E}\left[\hat{\cL}_D^{(r)} \right]
		&= \lim_{r\rightarrow\infty}\mathbb{E}\left[ \sum_{ x \in \cM_\cV^{(r)} }
				\sqrt{ \det\left( \hat{\mathbf{\Lambda}}(x) \right) } \prod_{d=1}^D \frac{\delta_d}{r+1} \right]\\
		&= \lim_{r\rightarrow\infty}\sum_{ x \in \cM_\cV^{(r)} }
				\mathbb{E}\left[ \sqrt{ \det\left( \hat{\mathbf{\Lambda}}(x) \right) } \right]
									\prod_{d=1}^D \frac{\delta_d}{r+1}\\
		&= \lim_{r\rightarrow\infty}
			\sum_{ x \in \cM_\cV^{(r)} }
				\sqrt{ \det\left( \mathbf{\Lambda}(x) \right)} \prod_{d=1}^D \frac{\delta_d}{r+1}\\
		&= \int_{ \cM_\cV } \sqrt{ \det\left( \mathbf{\Lambda}(x) \right)} dx 
		= \cL_D
\end{align*}
Similarly, using Fubini's theorem which is applicable as $\cM_\cV$ is
compact, we obtain that
\begin{align*}
	\mathbb{E}\left[\lim_{r\rightarrow\infty}\hat{\cL}_D^{(r)} \right]
		&= \mathbb{E}\left[\int_{\cM_\cV} \sqrt{ \det\left( \hat{\mathbf{\Lambda}}(x) \right) } dx \right]\\
		&= \int_{\cM_\cV} \mathbb{E}\left[\sqrt{ \det\left( \hat{\mathbf{\Lambda}}(x) \right) } \right]
									dx\\
		&= \int_{\cM_\cV} \sqrt{ \det\left( {\mathbf{\Lambda}}(x) \right) } dx
		 = \cL_D
\end{align*}
\end{proof}

\subsection{Proof of Theorem \ref{thm:consistency}}
\begin{proof}
	By the assumptions and Proposition \ref{prop:HoelderCont} the SuRF has almost surely $L^2$-H\"older continuous paths and
	therefore the assumptions of Lemma 11 from \cite{Telschow:2022SCB} are satisfied.
	This means that $\hat{\boldsymbol{\Lambda}}$ converges uniformly almost surely to $\boldsymbol{\Lambda}$ over all
	$x\in \overline{\cM_\cV}$. Since the Riemann sum converges to the integral the
	$\lim_{r\rightarrow\infty}\lim_{N\rightarrow\infty}$ statement follows immediately.
	On the other hand the $\lim_{N\rightarrow\infty}\lim_{r\rightarrow\infty}$ statement is a special case of
	Theorem $3$ from \cite{Telschow2020} since for $\cL_2$ and $\cL_3$ from their condition
	\textbf{(R)} only the uniform almost sure convergence
	of $\hat{\boldsymbol{\Lambda}}$ to ${\boldsymbol{\Lambda}}$ is required.
\end{proof}